\documentclass[12pt]{article}
\pdfoutput=1
\usepackage{mheckII}

\usepackage{longtable}

\usepackage{blkarray}

\usepackage{tikz} %%per disegnare schemi e grafici
\usetikzlibrary{matrix}

\theoremstyle{theorem}
\newtheorem{thm}{Theorem}[subsection]
\newtheorem*{thms}{Theorem}
\newtheorem{fact}{Fact}

\newtheorem{corl}{Corollary}[subsection]
\newtheorem{pro}{Proposition}[subsection]

\newtheorem{lem}{Lemma}[section]

\newtheorem*{cri}{Criterion}
\newtheorem*{cridfn}{Criterion/Definition}

\theoremstyle{definition}
\newtheorem{defn}{Definition}[subsection]
\newtheorem{rem}{Remark}
\newtheorem*{srem}{Remark}

\newtheorem{exe}{Example}[section]

\def\Dsl{\,\raise.15ex\hbox{/}\mkern-13.5mu D}

\def\sC{\mathscr{C}}
\def\sD{\mathscr{D}}
\def\sU{\mathscr{U}}

\def\sR{\mathscr{R}}
\def\bs{\mathbf{S}}
\def\bb{\mathbf{B}}
\def\bh{\mathbf{H}}

\date{January, 2019}

\title{Homological Classification of 4d $\cn=2$ QFT. Part I: Rank-1 revisited
}

\authors{Matteo Caorsi\footnote{e-mail: {\tt matteocao@gmail.com}} and Sergio Cecotti\footnote{e-mail: {\tt cecotti@sissa.it}}\vskip 9pt

\centerline{SISSA, via Bonomea 265, I-34100 Trieste, ITALY}}
\abstract{Argyres and co-workers started a program to classify all 4d $\cn=2$ QFT by classifying Special Geometries with  appropriate properties. They completed the program in rank-1. Rank-1 $\cn=2$ QFT are equivalently classified by the Mordell-Weil groups of certain rational elliptic surfaces.

The classification of 4d $\cn=2$ QFT   is also conjectured to be equivalent to the  representation theoretic (RT) classification of all 2-Calabi-Yau categories with suitable properties. Since the RT approach smells to be much simpler than the Special-Geometric one, it is worthwhile to check 
this expectation by reproducing the rank-1 result from the RT side. This is the main 
 purpose of the present paper. 
 Along the route we clarify several issues and learn new details about the rank-1 SCFT. In particular,
we relate the rank-1 classification to \emph{mirror symmetry} for Fano surfaces.  
 
In the follow-up paper we apply the RT methods to higher rank 4d $\cn=2$ SCFT.}

\begin{document}

\maketitle
%\tableofcontents
\newpage

\tableofcontents
%\newpage

\newpage

\section{Introduction}

An intriguing problem in QFT is the  classification of \emph{all} 4d $\cn=2$ models with special emphasis on the ones which do not have a Lagrangian realization. One nice approach,
especially advocated and implemented by Argyres and coworkers \cite{Argyres1,Argyres2,Argyres3,Argyres4,Argyres5,Argyres6,Argyres7,Argyres8,Argyres9,Argyres10,Argyres11}\footnote{\ See also \cite{Caorsi:2018zsq}.} is based on the idea that the classification of 4d $\cn=2$ QFTs is equivalent to the classification of all special geometries having the right properties to be the Seiberg-Witten geometry of a QFT \cite{SW0,SW1,Donagi};
string theory provides strong motivations for this geometric viewpoint \cite{vafaswamp}. Argyres \emph{et al.}\! have completed their program in rank-1, listing all $\cn=2$ SCFTs with a one-dimensional Coulomb branch \cite{Argyres3,Argyres4,Argyres5,Argyres6,Argyres7}. Their rank-1 classification was subsequently  reinterpreted \cite{Caorsi:2018ahl} in terms of the Mordell-Weil groups \cite{bookMW} of rational elliptic surfaces (with section) \cite{miranda} having at least one additive Kodaira singular fiber $F_\infty$ and at most one fiber of the three types $\{II,III,IV\}$ not dual to affine $\widehat{A}\widehat{D}\widehat{E}$ Dynkin graphs. In this approach, the rank-1 classifications is  read from the Oguiso-Shioda table of Mordell-Weil  groups \cite{bookMW,oguiso}.

Physical considerations suggest that there is yet a third approach to $\cn=2$ classification, pursued in a series of unpublished papers by Michele Del Zotto and the second named author \cite{unpublished}.
This approach starts from the general  expectation that the BPS objects (states and operators) of a supersymmetric quantum system are described by suitable triangle categories: the prime example being the BPS branes of  Type II on a Calabi-Yau 3-fold $X$ which are classified by the derived category of coherent sheaves and the derived Fukaya category on $X$ and its mirror manifold $X^\vee$ \cite{Aspinwall:2009isa}. 
If this expectation turns out to be correct, we may replace the hard classification of $\cn=2$ QFTs with the classification of triangle categories of the appropriate kind. So reformulated, the $\cn=2$ classification would become a problem in Representation Theory (RT). This problem smells to be much easier than the geometric program in several ways. At the philosophical level it looks akin of replacing the detailed study of complex manifolds with the computation of a certain homological invariant (albeit a  sophisticated one). Second:  the RT approach turned out to be very efficient to classify special classes of $\cn=2$ QFT, such as complete theories with the BPS-quiver property \cite{Cecotti:2011rv,Alim:2011ae}. Third: the preliminary steps in \cite{unpublished} suggest that the complexity of the RT program  grows less dramatically with the rank $k$ of  QFT than in the other approaches. Fourth: in Part II of this study we shall get partial all-rank classification results using ideas inspired by the RT approach.
Fifth: the RT viewpoint carries the additional bonus that from the knowledge of the categories associated to a given QFT we easily compute all supersymmetry protected physical quantities, even when localization methods fail.

Thus we have four classification problems which are purportedly equivalent:
\begin{center}
\begin{tabular}{c|c}
\textbf{\underline{Physics}} &
\textbf{\underline{Holomorphic Integrable Systems}}\\
Consistent 4d $\cn=2$ QFTs & \begin{minipage}{190pt}
\vskip 10pt
\begin{scriptsize}HyperK\"ahler spaces and holomorphic fibrations with special Lagrangian fibers which are 
generically\\ Abelian varieties 
(plus `regularity' conditions)\\
\phantom{.}
\end{scriptsize}
\end{minipage}\\
\hline
\textbf{\underline{Algebraic Geometry}}
&
\textbf{\underline{Representation Theory}}\\
\begin{minipage}{190pt}
\vskip 10pt
\begin{scriptsize}Smooth projective log-symplectic varieties with Abelian fibration and `regularity' conditions.\\ \underline{Rank-1}:
rational elliptic surfaces with restrictions
\end{scriptsize}
\end{minipage} &
\begin{minipage}{190pt}
\vskip 10pt
\begin{scriptsize} Hom-finite 2-CY categories $\sC$ with rigid objects subject to appropriate `regularity' conditions
\end{scriptsize}
\end{minipage}
\end{tabular}
\end{center}
%\smallskip
\medskip

What is the relation between these four different topics? 

A rank-$k$ special geometry is, in particular, a $k$-dimensional Abelian variety $\mathscr{A}_{/K(\mathcal{R})}$ defined over the field $K(\mathcal{R})$ of fractions  of the chiral ring $\mathcal{R}$. If the $\cn=2$ QFT does not contain free subsectors, the Lang-N\'eron trace \cite{LN} of 
$\mathscr{A}_{/K(\mathcal{R})}$ vanishes.
Then $\mathscr{A}_{/K(\mathcal{R})}$ comes with a God-given finitely-generated Abelian group, its \textit{Mordell-Weil group} $\mathsf{MW}(\mathscr{A}_{/K(\mathcal{R})})$; on its free part  
$\mathsf{MW}(\mathscr{A}_{/K(\mathcal{R})})_\text{free}=\mathsf{MW}(\mathscr{A}_{/K(\mathcal{R})})/\text{(torsion)}$ there is a positive-definite, integral, symmetric form, the \textit{N\'eron-Tate height} \cite{diop}. To a Hom-finite 2-CY category\footnote{\ For background on Calabi-Yau categories we refer to the nice review \cite{CYKeller}.} $\sC$ there is also naturally associated a finitely-generated Abelian group with a canonical symmetric integral pairing: its \textit{Grothendieck group} $K_0(\sC)$ with its \textit{``Euler''} form. In rank-1 the correspondence between the two geometric classifications and the representation-theoretical one turns out to be
\begin{equation*}
\boxed{\ \phantom{\Big|}\text{Mordell-Weil group with N\'eron-Tate height}\ \leftrightarrow\ \text{Grothendieck group with ``Euler'' form}\ }
\end{equation*} 
On the physical side, the group $K_0(\sC)/\text{(torsion)}$ is identified with the flavor weight lattice, and the isometry group of its ``Euler'' form is the Weyl group of the putative\footnote{\ The actual physical flavor group may be different (but with the same weight lattice). See the discussion in ref.\!\cite{Argyres5}.} flavor group $F$. The above statement makes sense for all $k$, not just in rank-1, but we shall resist the temptation of proposing a precise conjecture at this point.

\begin{srem} In general, both groups $\mathsf{MW}(\mathscr{A}_{/K(\mathcal{R})})$ and $K_0(\sC)$ have torsion. Matching torsions (and related structures) on the two sides is a subtle and physically crucial aspect of the correspondence. As in $F$-theory \cite{Morrison:2012ei} the Mordell-Weil torsion is related to the global geometry of the gauge group; in a sense, torsion ``reads between the lines of the gauge theory \cite{Aharony:2013hda}''. We discuss torsion in \textsc{appendix \ref{mwtorsion}}.
\end{srem}

In the unpublished notes \cite{unpublished} the RT program was pushed to some extend (for general rank $k$) getting correct formulae for the dimensions of the Coulomb branch operators,  flavor groups, central charges  $a$, $c$, $\kappa_F$, etc. Yet in its naive form the result  was short of a full classification. For instance, the rank-1 theories with exceptional flavor symmetry
$G_2$, $\text{Spin}(7)$, and $F_4$ were missing. \textit{A posteriori} one sees where the weak point was:
 the methods of \cite{unpublished} were returning not \emph{all} $\cn=2$ models but only  
a sub-class of them which we dub \textit{triangular} QFT. One has the inclusions\footnote{\ Triangular QFT are BPS-quiver $\cn=2$ models with some special property; in particular they have a unique Coulomb branch operator of maximal dimension. So class $\cs[A_1]$ theories are triangular iff their rank is 1. All interacting theories whose Coulomb dimensions are all $<2$ are triangular.}
$$
\Big[\,\text{triangular QFT}\,\Big] \subset \Big[\,\text{BPS-quiver $\cn=2$ theories \cite{Alim:2011kw}}\,\Big] \subset
\Big[\,\text{all $\cn=2$ QFT}\,\Big]
$$
In rank-1 the first inclusion is an equality.

The purpose of the present note is to fill the gap, and complete the RT side of the classification at least for $k=1$. Our immediate aim is to reproduce table 1 of \cite{Argyres3} using homological ideas and techniques.
In the process we shall understand better the procedure of gauging a discrete symmetry in $\cn=2$ QFT, the role of the RG flow, and find a detailed description of most  of the  relevant categories. The RT  approach allows \emph{inter alia} to explicitly compute the vev of generalized BPS Wilson-'t Hooft line operators in rank-1 theories (including the exceptional ones) as \emph{characters} of the corresponding 2-CY categories \cite{palu}. 

We hope that the abstractness of the categorical language will not hide the beauty and intrinsic simplicity of the homological approach. The math notation may look esoteric to some physicists, but the actual computations are (typically) quite simple. In an effort to make the paper readable, we have kept  mathematics at a minimum in the main text (at the cost of precision)
and confined details, computations (and precision) in the long appendices.  
\medskip

The rest of this paper is organized  as follows. In section 2 we present a cartoon of the categories associated to the BPS sector of a 4d $\cn=2$ theory. In section 3 we review the small part of the results of \cite{unpublished} relevant here. In section 4 we  consider the 15 missing SCFT and review base-change of elliptic surfaces and discrete gaugings of $\cn=2$ QFT. We also make some general observation. In section 5  we study in detail the discrete gaugings in our preferred theoretical laboratory: $SU(2)$ with $N_f=4$. This allows us to illustrate the main points of the paper in a simple context where everything can be done explicitly and the physics is well-understood from several viewpoints. 
In section 6 we pause a while to discuss what we have learned from the explicit examples, and to draw some general lesson.
The other discrete gauging of rank-1 SCFT are described in section 7.
Section 8 is devoted to the subtler situations we dub \emph{false}-gaugings. Conclusions are drawn in section 9. The appendices are long, self-contained, and contain all the details. Parts of them may be even readable.  

\section{BPS categories: a cartoon}\label{ccartoon}

In this section we present our physical motivations behind the representation-theoretical approach to $\cn=2$ supersymmetry.   
We shall be rather sketchy, using a rough language and avoiding all technicalities about categories. The interested reader may find proper definitions and a survey of basic theorems in \textsc{appendix \ref{adetails}}.
\medskip

To a given $\cn=2$ QFT there are attached several  triangle categories describing its BPS sector
\cite{Caorsi:2016ebt,Caorsi:2017bnp}.   These categories are fairly well understood when the QFT has the so-called \textit{BPS-quiver property} \cite{Cecotti:2011rv,Alim:2011ae,
Alim:2011kw}.
`Usual' theories (such as Lagrangian models) enjoy this property, but `most' $\cn=2$ QFT do not. For the purpose of
classifying \emph{all} $\cn=2$ QFT we have to extend the categorical framework  to the general case where the BPS-quiver property is no longer valid. This is the main  challenge of the present paper.
\medskip

We beging by sketching how triangle categories
arise from physical considerations in the special case of BPS-quiver theories (see \cite{Alim:2011ae,
Alim:2011kw,denef}  for details).
There are two points of view to consider: the UV perspective where the issue are the  BPS \emph{operators,} and the IR one focused  on the spectrum of BPS \emph{states} in a particular vacuum $|u\rangle$ belonging to a specific BPS chamber in the Coulomb branch.

\subsubsection*{IR picture} Let us start from the IR perspective: here 
the problem is to determine the spectrum of BPS particles of the given $\cn=2$ QFT.
There are two IR approaches: the Seiberg-Witten geometry and the supersymmetric Quantum Mechanics (SQM). We focus on the second one.

In the SQM approach we rephrase the question of the BPS particle spectrum 
%in vacuum $|u\rangle$ 
in terms of the 1d theory along the particle world-line which is a \textsc{susy} Quantum Mechanics with 4 supercharges. To specify the 1d Lagrangian $\mathscr{L}_u$ we need the following data: a gauge group $\mathscr{G}$, the gauge representation content of the chiral fields, a gauge-invariant superpotential $\cw$, and the FI terms \cite{denef}.
4d BPS states, preserving 4 supercharges, are \textsc{susy} vacua of the world-line theory: by standard  arguments they coincide with the vacua  of the 1d $\sigma$-model with target $\cv_u$, the space  of classical vacua of $\mathscr{L}_u$.
Computing the BPS spectrum boils down to the quantization of this reduced 1d $\sigma$-model.

When the 4d $\cn=2$ model has the 
 BPS-quiver property, $\mathscr{L}_u$ is a unitary quiver gauge theory, based on  a 2-acyclic quiver\footnote{\ A quiver $Q$ is \emph{2-acyclic} iff no arrow starts and ends at the same node nor there are pairs of opposite arrows $\leftrightarrows$ between two nodes.} $Q$, and $\cv_u$ is the moduli space of stable representations of $Q$ subjected to the relations given by the $F$-term  equations $\partial\cw=0$ \cite{Alim:2011kw}. In an equivalent language, $\cv_u$ is the space of stable modules of the Jacobian algebra $\cp(Q,\cw)$. (An algebra is called \emph{Jacobian} iff its relations are given by the gradient of a single-trace superpotential: $\partial\cw=0$).  In the 2-acyclic case the nodes of $Q$ are in 1-to-1 correspondence with the conserved electric/magnetic/flavor charges, and the gauge group at the $i$-th node is $U(N_i)$ for the world-line theory of particles carrying $N_i$ units of the $i$-th conserved charge.
   The incidence matrix $B$ of $Q$
   \be\label{difromr}
B_{ij}\equiv\langle S_j,S_i\rangle_D=
\#\big\{\text{arrows $i\to j$ in $Q$}\}-\#\big\{\text{arrows $j\to i$ in $Q$}\}
\ee
 is the Dirac skew-symmetric pairing between the $i$--th and $j$-th charges.
The pair $(Q,\cw)$ is not unique for a given 1d theory: it depends on the chosen Seiberg-duality frame. The \textit{mutation class} of $(Q,\cw)$ is the set of pairs $(Q^\prime,\cw^\prime)$, with $Q^\prime$ a 2-acyclic quiver, which can be obtained from $(Q,\cw)$ by a chain of Seiberg-dualities. Different regions of the Coulomb branch are covered by different pairs $(Q,\cw)$ in the mutation class.
\medskip

%  Note that 2-acyclicity of the quiver follows from a special property of the particular 4d theory and is not forced on us by supersymmetry along the world-line or other general principles.  
  Even when the 4d $\cn=2$ theory does \underline{not} obey the BPS-quiver property, it still makes sense to consider the world-line theory of its BPS particles. It is again a 4-supercharge supersymmetric Quantum Mechanics. As in the BPS-quiver case, the 1d Lagrangian $\mathscr{L}_u$ should be well  defined for all values of the  additive conserved charges, and have a good  behavior for large charges, which is the large-$N$ limit from the 1d perspective.
Comparing with 't Hooft analysis \cite{tHooft:1973alw}, we see that a cheap way to ensure this property is to consider a unitary quiver theory with a gauge-invariant superpotential $\cw$ which is a single-trace operator. This is what happens in the BPS-quiver case, except that in general the quiver $Q$ may be allowed to have loops and 2-cycles. However this is not the most general  possibility: there are many variants of this construction.
For instance, a quiver 1d Lagrangian $\mathscr{L}_u$ with the above single-trace  form may have a discrete symmetry $\bG$ which commutes with \textsc{susy};
gauging it we get a different 1d
theory $\mathscr{L}_u/\bG$ which still satisfies all requirements to be the world-line theory of BPS particles in some 4d $\cn=2$ QFT. We can push this even further: mimicking recent constructions in 4d extended supersymmetry \cite{discretegaugings1,discretegaugings2}, we may gauge a discrete symmetry $\bG$ of the 1d theory which is not a symmetry of its Lagrangian $\mathscr{L}_u$. In 4d one  gauges a suitable finite group $\bG$ of $S$-dualities; correspondingly we may think of gauging a discrete group $\bG$ of 1d Seiberg dualities (which are related to 4d $S$-dualities).

All these constructions start from a covering ($\equiv$ ungauged)  
1d Lagrangian $\mathscr{L}_u$ which is a unitary quiver theory with a single-trace superpotential and then take a quotient by a discrete group $\bG$ of symmetries/dualities. Specifying a covering  Lagrangian $\mathscr{L}_u$ of this class is equivalent to giving an associative algebra $\ca(\mathscr{L}_u)$ with the special property that its relations can be written as the gradient of a holomorphic gauge-invariant  superpotential $\cw$, that is, the algebra $\ca(\mathscr{L}_u)$ must be Jacobian. To avoid divergences, $\ca(\mathscr{L}_u)$ should also be finite-dimensional. To complete the description of the parent Lagrangian $\mathscr{L}_u$, the algebra $\ca(\mathscr{L}_u)$ should be supplemented by a stability function $Z_u$ for its modules which encodes the 1d FI terms. Once such a covering theory is given, to get the full story one has to look  for all finite symmetry groups $\bG$ which may be gauged while preserving supersymmetry. 
\medskip

\subsubsection*{UV picture}
Let us turn to the UV side. For models with the BPS-quiver property, physical considerations \cite{CNV,gaio1,gaio2,YQ}, predict the UV category $\sC$ to be the \textit{cluster category} $\sC(Q,\cw)$ associated to the 2-acyclic quiver with superpotential $(Q,\cw)$ which specifies the 1d Lagrangian $\mathscr{L}_u$. $\sC(Q,\cw)$ is independent of the choice of $(Q,\cw)$ in its mutation class, i.e.\! the microscopic category
$\sC(Q,\cw)$ is a Seiberg-duality invariant.
A cluster category is a special instance of a  2-CY category $\sC$ having  cluster-tilting objects $T$\,\footnote{\ For a nice review of the relevant mathematics, see \cite{keller}.}. (A category $\sC$ is 2-CY iff it has a Serre functor $S$ such that\footnote{\ For  notations and definitions see  \textsc{appendix \ref{adetails}}.} $S\cong \Sigma^2$, with $\Sigma$ the shift functor).

\subparagraph{Triangular case: 4d/2d correspondence.} The situation simplifies if the $\cn=2$ QFT is triangular (in rank-1 all BPS-quiver theories are triangular). These $\cn=2$ models have a $F$-theory construction along the lines of \cite{CNV} and enjoy the \textit{4d/2d correspondence} \cite{CNV}: their UV category $\sC$ is a close cousin of the BPS brane category $\sR$ of a 2d (2,2) model
with $\hat c<2$. The condition $\hat c<2$ should be understood as an upper  dimensional bound\footnote{\ It arises as the condition for the existence of a crepant resolution in 3 dimensions \cite{Gukov:1999ya}.}: indeed the 2d superconformal central charge $\hat c$ has the physical interpretation of a `fractional Calabi-Yau dimension' (think, say, of the Gepner models \cite{gepmod}).

 To construct the relevant categories, in the triangular case one starts from an algebra $\ca$ satisfying certain restrictions.
 The physical requirement $\hat c<2$ becomes the condition\footnote{\ To make the story short we limit ourselves to the case in which the 4d theory is superconformal; the construction of \cite{CNV} extends to the asymptotically-free QFT.} that its derived category $\sD_\ca\equiv D^b\mathsf{mod}\,\ca$ is Calabi-Yau of fractional dimension $a/b<2$, namely in $\sD_\ca$ 
there is an isomorphism
 \be\label{jjjaz1294}
 S^b\cong\Sigma^a,\qquad \hat c\equiv a/b<2,\quad a,b\in\bN,
 \ee 
 where $S$ (resp.\! $\Sigma$) is the Serre (resp.\! shift) functor (see \textsc{appendix \ref{adetails}}).
The 2d and 4d BPS categories are just different orbit categories of $\sD_\ca$:
\be\label{kkaszqwe}
\sR_\ca=\Big(\sD_\ca\big/(\Sigma^2)^\Z\Big)_\text{tr.hull}\qquad
\sC_\ca=\Big(\sD_\ca\big/(S^{-1}\Sigma^2)^\Z\Big)_\text{tr.hull}.
\ee
Here $(\cdots)_\text{tr.hull}$ stands for  `triangular hull', a technicality we dispense with\footnote{\ The interested reader may give a look to the \textsc{appendices}.}.
By construction, in $\sC_\ca$ we have
$S\cong \Sigma^2$, which is the defining property of a 2-CY category, while in $\sR_\ca$ we have  
$\Sigma^2\cong\mathrm{Id}$, i.e.\! the category $\sR_\ca$ is 2-periodic.   
In the physical literature the image of $S$ (resp.\! $\Sigma$) in $\sC_\ca$ is called the \textit{4d quantum monodromy} $\bM$ (resp.\! the \textit{half-monodromy} $\bK$) \cite{CNV,YQ}, while the image of $S$ (resp.\! $\Sigma$)  in $\sR_\ca$ is called the \textit{2d quantum monodromy} $H$ (resp.\! \emph{half-monodromy}) \cite{CV92}. 

\subparagraph{Properties of the quantum monodromies.} All $\cn=2$ QFT, not just the BPS-quiver ones, have quantum monodromy operators $\bM$ and $\bK$ which act on the BPS operators $O$ as  
\be
O\mapsto \bM O\bM^{-1}\quad \text{and}\quad O\mapsto\bK O\bK^{-1}.
\ee
 $\bM$ commutes with flavor group action, $\bM g=g\bM$ for $g\in F$, while $\bK$ inverts\footnote{\ $F$ is compact so a torus times a semi-simple Lie group. In the torus $\theta$ is $-1$; in the semi-simple part $\theta$ is  the involution which acts as $X_\alpha\leftrightarrow X_{-\alpha}$ on the Chevalley generators (cfr.\! \cite{bourbaki} VIII\,\S.4 Proposition 4).} the flavor action, $\bK g= \theta(g)\bK$
  for $g\in F$.
 $\bM$ and $\bK$  have explicit expressions in terms of the BPS spectrum \cite{CNV}. The expressions involve choices, but their adjoint action on the BPS operators is intrinsically defined: the condition that $\bM$ computed from the spectra on the two sides of a wall of marginal stability yields the same action is equivalent to  the Kontsevitch-Soibelman wall-crossing formula \cite{wallcro}.  $\mathrm{Tr}\,\bM$ is the Schur index of the 4d SCFT \cite{Cordova:2015nma}, i.e.\! the vacuum character of the 4d infinite chiral algebra in the sense of \cite{Rastelli}. If our QFT is
a SCFT, the action of $\bM$ on the local Coulomb branch operators $O_i$ is just
\be\label{kqwaz}
\bM O_i\bM^{-1}=e^{2\pi i\Delta_i}\,O_i,
\ee
 where $\Delta_i$ is the conformal dimension of $O_i$. The action on the BPS \emph{line} operators is rather  complicated and interesting: in the BPS-quiver case this action fully reflects the cluster structure\footnote{\ For the notion of \textit{cluster structure} in a 2-CY category see the nice review \cite{reitenrev}.} of $\sC$. 
When the QFT is
  the mass-deformation of a SCFT, $\bM$ has a finite order $o(\bM)$ equal to the smallest integer such that $o(\bM)\,\Delta_i\in\bN$ \cite{CNV}\footnote{\ When the $\cn=2$ theory is Lagrangian (possibly asymptotically-free), $\bM$ acts as a shift of the Yang-Mills angles $\theta_i$ by $2\pi b_i$, where $b_i$ is the coefficient of the beta-function of the $i$-th Yang-Mills coupling.}. In particular, whenever $\Delta_i\in\bN$ for all $i$
we have $\Sigma^2\cong\mathrm{Id}$
in the UV category $\sC$, which then is \emph{symmetric} besides being 2-CY.
In the same fashion, we define the order $o(H)$ of the 2d quantum monodromy.
It is the smallest integer so that $o(H) d_i\in\bN$, where $d_i$ are the conformal dimensions of the 2d chiral primary operators \cite{CV92}.
For a rank-1, 2-acyclic, $\cn=2$ theory the dimension of the Coulomb operator is \cite{Caorsi:2017bnp}
\be
\Delta=\frac{o(H)}{o(\bM)},
\ee
and there are similar expressions for $\Delta_i$ in all ranks.

\subparagraph{Flavor symmetry.}
In full generality, whenever a class of BPS objects is described by a triangle category $\sU$, its conserved quantum numbers factor through the Grothendieck group $K_0(\sU)$. Hence the class $[X]\in K_0(\sU)$ should be seen as the universal conserved quantity, and the free group $K_0(\sU)/\text{(torsion)}$ as the set of all additive quantum numbers. In the BPS-quiver case, the UV group
$K_0(\sC)/\text{(torsion)}$ is identified with 
the flavor weight lattice $\Gamma_\text{flav}$ of the 4d theory
\cite{Caorsi:2017bnp}
\be\label{qqqqa}
\Gamma_\text{flav}\equiv K_0(\sC)/\text{(torsion)}.
\ee 
In particular,
\be\label{zzzaq6}
f\equiv\mathrm{rank }\,F=\mathrm{rank}\,K_0(\sC).
\ee  
The lattice $\Gamma_\text{flav}$ has a natural Weyl-invariant symmetric pairing
which should be part of the categorical description: this is required for the flavor symmetry $F$ to act on the BPS objects (as described by $\sC$) in  the proper way.
For a triangle category $\sC$, the only intrinsic bilinear form on the lattice 
$K_0(\sC)/\text{(torsion)}$ is its Euler pairing.\footnote{\ In the present context the definition of the Euler pairing is slightly subtle \cite{Caorsi:2017bnp}, requiring `cutting techniques' \cite{groK}.}
A triangle category $\sC$ has naturally a symmetric (resp.\! antisymmetric) Euler form iff it has the Calabi-Yau property in even (resp.\! odd) dimension $d$.
So the fact that $\sC$ is CY with $d=2$ looks very natural from the flavor point of view. Indeed, in the BPS-quiver case the Euler form is equal to the canonical inner product in the weight lattice (up to overall normalization, see \textsc{appendix \ref{catfal}}).

% When the $\cn=2$ theory is Lagrangian, $\bM$ acts as a shift of the Yang-Mills angles $\theta_i$ by $2\pi b_i$, where $b_i$ is the coefficient of the beta-function of the $i$-th Yang-Mills coupling. $\bM$ has a canonical square root $\bK^2=\bM$, the half-monodromy, which also has an explicit expression in terms of the BPS spectrum. However $\bK$ inverts the flavor action, $\bK g= g^{-1}\bK$
%  for $g\in F$.
%  
%  In particular, if $\sD$ is the derived category of coherent sheaves on a one-dimensional orbifold, as in the examples of pure $SU(2)$ and $SU(2)$ with $N_f=4$,
%  the coefficient $b$ of the Yang-Mills $\beta$-function is the degree of $K^2$. This yields $-4$ and $0$ for the two examples.
%  \medskip

\subsubsection*{Connecting the UV picture to the IR one}
The IR physics should be determined by the microscopic UV dynamics together with the choice of a particular vacuum $|u\rangle$. The physical  
connection $UV\leadsto IR$ is the RG flow.
For the BPS sector the UV/IR connection has two main aspects: 1) the UV BPS line operators $L_\alpha$ have vev's
$\langle L_\alpha\rangle_u$ which are part of the physics of the particular state $|u\rangle$, and 2) 
the spectrum of BPS particles which are stable in vacuum $|u\rangle$ should be determined by the  dynamics of the microscopic degrees of freedom.

As mentioned above, when the $\cn=2$ QFTs has the BPS-quiver property, the UV category $\sC$ is a cluster category \cite{reitenrev}.  In this case the $UV\leadsto IR$ connection is provided by a cluster-tilting object $T\in\sC$.
$T$ is non-unique, a given choice of $T$ covering just a chamber in  the Coulomb branch.
The chosen cluster-tilting object $T\in \sC$ plays several roles in the connection $\text{UV}\leadsto\text{IR}$. 
First: $T$ determines the 1d Lagrangian/Jacobian algebra $\ca(\mathscr{L}_u)$ in the form  $\ca(\mathscr{L}_u)=\mathrm{End}_\sC(T)$.
Second: $T$ yields the Dirac skew-symmetric pairing on the charges, see eqn.\eqref{difromr}.
Third: $T$ defines the {cluster characters} $\langle-\rangle_T\colon \sC\to \C$ which correspond to $L_\alpha\mapsto \langle L_
\alpha\rangle_u$.
\medskip

Now consider the general case in which our $\cn=2$ theory does \underline{not} necessarily enjoy the BPS-quiver property. In the UV we still have some  triangle category $\sC$ describing BPS operators. To determine the low-energy physics in some chamber of the Coulomb branch we need to associate to the microscopic category $\sC$ (and choice of chamber) the supersymmetric theory on the BPS particle world-line. We argued above that this theory is obtained by gauging a discrete symmetry $\bG$ of some parent 1d supersymmetric model which in turn is  described by a Lagrangian $\mathscr{L}$ with a single-trace superpotential $\cw$. The ungauged
Lagrangian is specified by the finite-dimensional algebra $\ca(\mathscr{L})$. The standard way to produce an algebra out of a linear category $\sC$ is to choose an object $X\in \sC$ and consider the algebra of its endomorphisms $\mathrm{End}_\sC(X)$,
which is finite-dimensional since $\sC$ is Hom-finite. The choice of the object $X\in\sC$ reflects the choice of the chamber. However the pair
$(\sC,X)$ cannot be arbitrary: to preserve 1d supersymmetry we need the algebra $\mathrm{End}_\sC(X)$ to be Jacobian, that is, the 1d interactions must be described by a holomorphic superpotential $\cw$. This is a severe constraint which has a remarkably simple solution:

\begin{thms}[Amiot, Keller \footnote{\ See e.g.\!
\textbf{Theorem 5.6} in the survey  \cite{reitenrev}.}] \label{jacth}All finite-dimensional Jacobian algebras have the form $\mathrm{End}_\sC(T)$ with $\sC$ a 2-CY category and $T\in\sC$ a cluster-tilting object.
\end{thms}
To our knowledge, the converse statement is still an open problem. It is however known to be true under various natural hypothesis on the 2-CY category $\sC$. In particular, for all cluster categories associated to 2-acyclic quivers, 
$\mathrm{End}_\sC(T)$ is Jacobian for all cluster-tilting object $T\in\sC$ \cite{reitenrev}.
\medskip

Note that both $\mathscr{L}$ and $\mathscr{L}/\bG$ have the right properties  to be the 1d theory of some (different)
4d $\cn=2$ model. When the finite group $\bG$ is non-trivial, it is natural and convenient to associate to the model \emph{two} different (Hom-finite) 2-CY categories, $\sC$ and $\sC_\bG$, related by a Galois cover
$\sC\to\sC_\bG$ with deck group $\bG$.
To produce the ungauged 1d Lagrangian/Jacobian algebra $\ca(\mathscr{L})$ the corresponding 2-CY category $\sC$ should have a cluster-tilting object $T$ so that
$\ca(\mathscr{L})\cong\mathrm{End}_\sC(T)$. No such condition is required for 
the covered 2-CY category $\sC_\bG$.
Thus we shall consider acceptable also (Hom-finite) 2-CY categories $\sC$ without cluster-tilting objects provided there is a $\bG$-covering of 2-CY categories $\tilde\sC\to \sC$ and the cover $\tilde\sC$ has a cluster-tilting object.
However, whenever our $\cn=2$ QFT has a non-trivial flavor group $F$, the corresponding 2-CY categories $\sC$ should satisfy the weaker property of having non-zero rigid objects.
Since all SCFT in table 1 of \cite{Argyres3}
have $F\neq1$ we shall add this condition
as an extra assumption. Dropping it, one finds a handful of extra possibilities with trivial flavor symmetry.

In conclusion: we propose that the UV category $\sC$ of a general 4d $\cn=2$ belongs to a `slight' generalization of the class known to describe QFT with the BPS-quiver property: cluster categories (which, in particular, means 2-CY with cluster-tilting) are replaced by the larger class of 2-CY categories with non-trivial rigid objects $T_\text{max}$ with the additional property of having a Galois cover
$\tilde\sC$ which is a cluster category.
\medskip

There is an important difference between the BPS-quiver and general cases.
In the first situation the rank of the Grothendieck group $K_0(\mathrm{End}_\sC(T_\text{max}))$ is known to be
\be
\mathrm{rank}\,K_0(\mathrm{End}_\sC(T_\text{max}))=\mathrm{rank}\,F+2 k,
\ee 
with $k$ the complex dimension of the Coulomb branch. In the general case the rank of $K_0(\mathrm{End}_\sC(T_\text{max}))$ may be smaller
\be
\mathrm{rank}\, F\leq \mathrm{rank}\,K_0(\mathrm{End}_\sC(T_\text{max}))\leq \mathrm{rank}\,F+2 k,
\ee
and we expect the upper bound to be saturated only in the BPS-quiver case. The lower bound follows from \eqref{zzzaq6}.
In the rank-1 case we remain with just two possibilities:
\be\label{pppqwer}
 \mathrm{rank}\,K_0(\mathrm{End}_\sC(T_\text{max}))= \begin{cases} f+2 & \text{BPS-quiver}\\
 f & \text{otherwise.}
 \end{cases}
\ee
Let us illustrate the physical origin of this different behavior in a simple example: the rank-1 SCFT with Coulomb branch dimension $\Delta=4$ and flavor group $Spin(7)$.
This theory is obtained by gauging a discrete symmetry $\Z_2\subset PSL(2,\Z)\times U(1)_R$ of $SU(2)$ SYM with $N_f=4$ \cite{Argyres7,discretegaugings1,discretegaugings2}. Here $PSL(2,\Z)$ is the $S$-duality group. The $\Z_2$ gauge group rotates electric charges into magnetic ones and viceversa. In this SCFT the distinction between electric and magnetic charges is not just conventional, it is \emph{gauge-dependent}. The intrinsic, gauge-invariant, Dirac pairing $\langle-,-\rangle_D$, averaged over the orbits of the $\Z_2$ gauge group, just \textit{vanishes.} To define a non-trivial  electro-magnetic pairing in the IR -- to measure the mutual non-locality of states -- we need to fix a $\Z_2$-gauge in the corresponding category $\mathsf{mod}\,\mathrm{End}_\sC(T_\text{max})$.
One of the reasons why we prefer to work (in the general case) with a Galois pair $(\sC,\sC_\bG)$ of 2-CY categories rather than with the single physical UV category $\sC_\bG$ is to make the gauge-fixing  construction systematic. Fixing the gauge means to choose a `local' lift from $\sC_\bG$ to $\sC$ and perform computations in the much simpler covering category $\bG$.   
 Irrespectively of these technicalities, since $\mathrm{rank}\,K_0(\mathsf{mod}\,\mathrm{End}_\sC(T_\text{max}))$ is always $f+\mathrm{rank}\,\langle-,-\rangle_D$, we get eqn.\eqref{pppqwer}.

\subsubsection*{UV completeness: Our final proposal}

The above considerations suggest the working hypothesis that the classification of  4d $\cn=2$ QFT, say in rank-1, may be traded for the classification of  2-CY categories with rigid objects and cluster covers. However, it is certainly \underline{not} true that all (Hom-finite) 2-CY categories with rigid objects and cluster covers describe some consistent 4d $\cn=2$ QFT. For instance, there exist perfectly good cluster categories associated to non-UV-complete field theories such as $SU(2)$ SYM coupled to 5 fundamentals, or to one adjoint and one fundamental hypermultiplet. Such theories make sense only as low-energy effective descriptions, obtained by integrating out the heavy degrees of freedom of some UV completion.
Triangle categories behave well under RG, and there is a categorical version of `integrating out the heavy stuff'\,\footnote{\ It is called Calabi-Yau reduction \cite{keller} or subfactor construction \cite{reitenrev}.}. Not surprisingly, we may obtain the cluster categories of $SU(2)$ with $N_f=5$,\footnote{\label{curious}\ $SU(2)$ with $N_f=5$ is an effective QFT which is `good' up to a certain cut-off scale $\Lambda$. Its IR category $\mathsf{mod}\,\mathrm{End}_{\sC_{N_f=5}}(T)$ should also be `good' only up to a certain scale, i.e.\! when restricted to modules of bounded Grothendieck class. Remarkably, this very point was noted as a `curious' fact by C.M. Ringel back in 1984 (cfr.\! \textbf{Remark 2.} on page 166 of his celebrated book \cite{ringelbook}).}
or $SU(2)$ with one $\boldsymbol{3}$ and one $\boldsymbol{2}$, by suitable decoupling limits of, say, the rank-1 SCFTs with $\Delta=3$ and flavor group $E_6$ and, respectively, $Sp(6)$. 
%Since $SU(2)$ with $N_f=5$ is just an effective QFT which is `good' only up to a certain cut-off scale $\Lambda$, the corresponding IR description, 
%$\mathsf{mod}\,\mathrm{End}_{\sC_{N_f=5}}(T)$ should also be `good' only up to a certain scale, that is, when restricted to modules whose Grothendieck class is not too big. Remarkably, this very point was noted as a `curious' fact by C.M. Ringel back in 1984 (cfr.\! \textbf{Remark 2.} on page 166 of his celebrated book \cite{ringelbook}).
\medskip

For the purpose of classifying $\cn=2$ QFTs we have to keep only 2-CY categories which correspond to UV-complete theories.
We need a criterion to distinguish 2-CY categories which correspond to UV-complete QFT from the ones which do not.  In the BPS-quiver context the criterion proposed in \cite{CNV} reduces to the bound 
$\hat c<2$, see discussion\footnote{\ Again we limit to the SCFT case. The general criterion encompasses the asymptotically-free case too.} around eqn.\eqref{jjjaz1294}.
%
%
%
%states that a 2-acyclic cluster category is ``UV-complete'' iff the 2-acyclic endo-quivers of its cluster-tilting objects satisfy\footnote{\ For instance, from the Ringel analysis \cite{ringelradius} one concludes that an acyclic quiver satisfies the criterion only if it is Dynkin or affine Dynkin. Of course, this was also understood from the viewpoint of \cite{CV92}. } the Diophantine conditions of the classification program for 2d (2,2) QFT \cite{CV92}, which
%guarantee the existence of a sensible $tt^*$ geometry \cite{ttst},
%for a value of the (2,2) superconformal central charge $\hat c<2$.
% For instance, the cluster category for $SU(2)$ $N_f=5$ should be ruled out, since the associated 2d model is a $\sigma$-model with target a certain orbifold of a genus 2 curve. Since the target Ricci curvature is negative, the 2d theory has Landau poles (this is a simple instance of Perelman's finite-time singularity theorem for Ricci-flows which implies the Poincar\'e conjecture \cite{poin}). In this case the 2d quantum theory does not exist, and the $tt^*$ geometry makes no sense above the Perelman termination scale. In the RT language the 4d/2d criterion refers to the spectral properties of the algebra $\mathrm{End}_\sC(T)$ in the sense, say, of \cite{specpro}. See \cite{Cecotti:2012va,Caorsi:2016ebt,Caorsi:2017bnp} for further details.
%\medskip
%
%
To classify \emph{all} rank-1 4d $\cn=2$
we need to extend this criterion to the non-BPS-quiver case. In line with the stringy arguments which motivated the original criterion \cite{CNV}, and in view of the 4d physics, it is natural to propose a criterion such that the set of ``UV-complete'' 2-CY categories is the \emph{smallest} one which contains the 2-acyclic categories consistent with the 4d/2d correspondence and is closed under the standard physical operations: gaugings which preserve $\cn=2$ \textsc{susy} and  RG-flows. 
%To traslate the original criterion in 
%categorical language, just notice that for a 2d (2,2) SCFT the rational number $\hat c$
%has the physical interpretation of a fractional Calabi-Yau dimension (say, in the context of Gepner models \cite{gepmod}). In RT one also defines a notion of ``fractional Calabi-Yau dimension $a/b$'' for triangle categories with a Serre functor $S$, see \textsc{appendix \ref{fracCY}}. Roughly speaking,
%one has to identify the physical and mathematical notions of `fractional Calabi-Yau dimension', but there is a subtlety:
% the 2d QFT needs not to be a SCFT, it may be just asymptotically-free. In this case what matters is that the $\beta$-function has the correct sign (which guarantees that the $tt^*$ geometry has an holonomy consistent with the $SL_2$-orbit theorem of Hodge theory \cite{CV92,sl2}).
Pragmatically, this amounts to requiring the 4d/2d correspondence to be satisfied by the covering cluster category $\tilde\sC$ of our UV 2-CY category $\sC$. 
To keep the statement as simple as possible, we shall not state the criterion in its most general form (the interested reader may look at \cite{Cecotti:2012va,Caorsi:2016ebt,Caorsi:2017bnp}) but only its  simplified version which applies to rank-1 SCFT.
In our tables we shall insert back the rank-1 asymptotically-free categories.    
So specialized, our `conservative' proposal takes in the RT language
a somewhat esoteric form;
we state it in a slightly cavalier fashion\footnote{\ The conditions in the text are slightly too restrictive. We have to ``close'' the class of categories by adding the categories which arise as ``de-singularizations of singular limits'' of the categories in the \textbf{Criterion/Definition}.
We shall comment on this technicality in section 4.} 

\begin{cridfn}\label{cride} A 2-CY category $\sC$ with non-zero rigid objects is said to be \emph{the UV category of a SCFT of rank-1} if it is (the hull of) an orbit category of the form
\be\label{criterion}
\sC_{G}=\Big(\sD_\ca\big/(G)^\Z\Big)_\text{\rm tr.hull},
\ee
where $\sD_\ca=D^b\mathsf{mod}\,\ca$ is the derived category of an algebra $\ca$ satisfying the conditions below, and $G\colon \sD_\ca\to\sD_\ca$ is a suitable  autoequivalence.
The algebra $\ca=\C \mathring{Q}/I$ has global dimension $\leq2$ and:
\begin{itemize}
\item[a)] \emph{(4d/2d criterion $\hat c<2$)} The derived category $\sD_\ca\equiv D^b\mathsf{mod}\,\ca$ is fractional Calabi-Yau of dimension $a/b<2$;
\item[b)] \emph{(Coulomb dimension 1)} The rank of the exchange matrix $B$ of the $2$-acyclic quiver $Q$ (obtained by adding to $\mathring{Q}$ an inverse arrow per  minimal relation of $I$, and then reducing it by deleting loops and conflicting pairs of arrows $\leftrightarrows$) is 2.
\end{itemize}
\end{cridfn}
\begin{rem} Note that the physical requirement of UV finiteness, item \textit{a)}, implies that
the functor $\mathsf{Tor}^A_2(?,D\ca)\colon\mathsf{mod}\,\ca\to\mathsf{mod}\,\ca$ is nilpotent, equivalently that the category $\sC_\ca$ in eqn.\eqref{kkaszqwe} is \textit{Hom-finite;}\footnote{\ See \textbf{Proposition 4.9}(4) of \cite{amiot}.} thus the absence of infinities in the sense of QFT implies the absence of infinities in the  sense
of categories.
Hence $\sC_\ca$ is the Amiot cluster category which is Hom-finite, 2-CY with $\ca$ as cluster-tilting object. The condition of being 2-CY implies that $G^s=S\Sigma^{-2}$ for some $s\in\bN$, and then $\sC_G$ is also Hom-finite and 2-CY but may or may not have a cluster-tilting object.
See \textsc{appendix B} for more details. 
\end{rem}

\subsection{Kodaira type of a 2-CY category}

Our basic claim is that the classification of 2-CY categories of the form \eqref{criterion} is the classification of rank-1 $\cn=2$ SCFTs.
The $\cn=2$ SCFTs (or, more generally, the $\cn=2$ QFTs) have two kinds of 
classifications: \emph{fine} and \emph{coarse-grained}. The rank-1 coarse-grained classes are in one-to-one correspondence \cite{Argyres3,Caorsi:2018zsq} with the Kodaira's singular elliptic fibers $\cf$ \cite{koda,kod1,kod2} having additive reduction (i.e.\!\! \emph{non}-semi-stable) and Euler characteristic $e(\cf)\leq 10$; the SCFT (i.e.\! semi-simple) ones are listed in table \ref{ffffhad}.
There is a forgetful map
$\text{(fine class)}\mapsto \text{(coarse-grained class)}$,
hence a map 
\be
\text{(2-CY category of the form \eqref{criterion})}\mapsto \text{(additive Kodaira fiber)}.
\ee
The Kodaira fiber $\cf$ associated to a 2-CY category will be called its
\emph{Kodaira type}. Kodaira type is a basic tool in the application of 2-CY categories to the study of $\cn=2$ QFT. In view of the extension of present work to rank $k>1$, we shall define the Kodaira type in general not just for the categories satisfying our \textbf{Criterion/Definition}. 

We recall that the Kodaira fibers \cite{koda,kod1,kod2}
\be\label{kkkaqwp}
I_b,\ I^*_b\ (b\geq0),\ II,\ III,\ IV,\ II^*, III^*,\ IV^*
\ee
corresponds to the conjugacy classes of elements $\rho\in SL(2,\Z)$ consistent with the strong monodromy theorem: \textit{i)} they are
quasi-unipotent i.e.\! $(\rho^m-1)^n=0$ for some $m,n\in\bN$, and \textit{ii)} the unipotent element $\rho^m$ satisfies the $SL_2$-orbit theorem \cite{sl2}. The fiber is additive iff $m>1$; all types \eqref{kkkaqwp} are additive but $I_b$. Physically, these conditions corresponds to UV completeness. 
 
\subparagraph{Kodaira type of an algebra of finite global dimension.} 
The Kodaira type of a (finite-dimensional, basic) algebra of finite global dimension is a far-reaching  refinement of its spectral invariants 
(for background see e.g.  \cite{spectral}).
For fixed spectral data, the Kodaira type takes value in a finite group $H_\bK$
whose structure depends on the arithmetics of the spectral field $\bK$ (see below for examples).

Let $\ca=\oplus_{i=1}^\nu P_i$ with $P_i$ indecomposable, $P_i\not\cong P_j$ for $i\neq j$, and $\sD_\ca\equiv D^b\mathsf{mod}\,\ca$. Consider the  integral $\nu\times \nu$ matrix
\be
\mathbf{S}^{-1}_{ij}=\dim\sD_\ca(P_i,P_j).
\ee
Since $\ca$ has finite global dimension, $\bs^{-1}$ is a unit in $GL(\nu,\Z)$ so its inverse $\bs$
exists and has integral entries.
In facts, $\bs$ is the transpose of the Euler form in the basis of the simple modules $S_i\in\mathsf{mod}\,\ca$
\be
\bs_{ij}=\langle S_j,S_i\rangle_E \equiv\sum_{k\in\Z}(-1)^k\dim \sD_\ca(S_j, \Sigma^k S_i).
\ee
If,  in addition, $\ca$ is triangular, i.e.\! $\ca\cong \C \mathring{Q}/I$ with $\mathring{Q}$ a quiver without oriented cycles, $\bs$ is an upper-triangular matrix with 1 on the main diagonal, and physicists call it the \emph{Stokes matrix} \cite{CV92} since they like to see it as
the monodromy datum of a Sato-Miwa-Jimbo
isomonodromic system of PDE's \cite{satomiwa,cvising}.
A triangular algebra $\ca$ is ``physically nice'' iff $\bs$ produces a regular solution of the PDE's.\footnote{\ To get a flavor of the PDE regularity conditions, see \cite{ttstar1,ttstar2,ttstar3} where the $\Z_\nu$ symmetric case is analyzed in detail for $\nu\leq 5$.}

 Returning to the general case, we define 
$\bb$ as the antisymmetric part of the Euler form
\be
\bb_{ij}= \langle S_i, S_j\rangle_\text{anti}\equiv\langle S_i, S_j\rangle_E-\langle S_j,S_i\rangle_E= (\bs^t-\bs)_{ij}.
\ee
$\sD_\ca$ has a Serre auto-equivalence $S$. The 2d quantum monodromy $\bh$ is the action of $S$ in the Grothendieck group $K_0(\sD_\ca)$ ($\equiv$ the free Abelian group on the classes $[S_i]$,
$i=1,\dots, \nu$)
\be\label{zzz0123}
[SS_i]= [S_j]\,\bh_{ji}\quad\Rightarrow\quad \bh=(\bs^t)^{-1}\bs\in SL(\nu,\Z)
\ee 
(mathematicians prefer to use the Coxeter element $-\bh$).

A triangular algebra $\ca$ is ``physically nice'' iff $\bh^{-1}$, seen as a putative monodromy, is consistent with  the strong monodromy theorem. For $\ca$ triangular, we say that $\sD_\ca$ is \emph{strong} if these conditions are fulfilled, and \emph{numerically CY} if $\bh$
has finite order. We have the inclusions
\be
\text{($\sD_\ca$ fractional CY)}\subset \text{($\sD_\ca$ numerical CY)}\subset \text{($\sD_\ca$ strong).}
\ee 
\begin{exe} A canonical (or squid)
algebra of type $\{p_1,\cdots, p_t\}$ is strong iff its Euler characteristic $\chi\equiv 2-\sum_{a=1}^t(1-1/p_a)\geq0$, and fractional CY iff $\chi=0$.
\end{exe}

We go back to the general case ($\ca$ is not necessarily triangular).
One has
\be
\bh^t\, \bb\, \bh=\bb.
\ee
Let $\Gamma_\text{fl}\subset K_0(\sD_\ca)$ be the sublattice of
elements $x$ such that $\langle x,-\rangle_\text{anti}=0$; $\Gamma_\text{fl}$ coincides with the sublattice of elements such that $\bh\, x=x$. Since $\bb$ is skew-symmetric, $\bh$ has no non-trivial Jordan block associated to the eigenvalue 1. Then we may find a $\Z$-equivalent basis
such that
\be\label{kkkaqwe1za}
\bb=\left(\begin{array}{c|ccc}
\mathbf{b} & & 0 &\\\hline
&0 &&\\
0 & &\ddots &\\
&&& 0\end{array}\right)\qquad 
\bh=\left(\begin{array}{c|ccc}
\mathbf{h} & & 0 &\\\hline
&1 &&\\
0 & &\ddots &\\
&&& 1\end{array}\right)
\ee
where: \textit{i)} $\mathbf{b}$ is a non-degenerate  integral $2k\times 2k$ skew-symmetric matrix, \textit{ii)}
1 is not an eigenvalue of the integral matrix $\mathbf{h}$, and
\be
\text{\it iii)}\qquad \mathbf{h}^t\,\mathbf{b}\,\mathbf{h}=\mathbf{b},
\ee
that is, $\mathbf{b}$ is an element of the arithmetic group
$Sp(\mathbf{b},\Z)$; e.g.\!\! if $\mathbf{b}$ is a multiple of the standard symplectic matrix, $\mathbf{h}$ belongs to the Siegel modular group $Sp(2k,\Z)$. The conjugacy class of $\mathbf{h}$ in $GL(2k,\C)$
yields the \emph{spectral invariants} of $\ca$ \cite{spectral}.  Its conjugacy class in the arithmetic group
$Sp(\mathbf{b},\Z)$ is a much finer derived 
invariant of $\ca$ which encodes precious physical information. We call this invariant the \emph{Kodaira type} of $\ca$.

\begin{exe} For generic $k$, the simplest case is when $Sp(\mathbf{b},\Z)\cong Sp(2k,\Z)$, and $\mathbf{h}$ is regular of finite order, i.e.\! its minimal equation is $\mathbf{h}^m=1$ with\footnote{\ $\phi(-)$ stands for Euler's totient function.} $\phi(m)=2k$. In this case
the Kodaira type takes value in a  group $H_\bK$ whose order is \cite{Caorsi:2018zsq}
\be
|H_\bK|=\frac{2^{\phi(m)/2}}{Q_\bK}\,h_\bK^-,
\ee
where $Q_\bK$ and $h^-_\bK$ are, respectively, the Hasse unit index and the relative class number of the cyclotomic field $\bK$ of $m$-th roots of unity.
\end{exe} 

For an algebra $\ca$ which satisfies conditions \textit{a)},\,\textit{b)} of our
\textbf{Criterion/Definition} the story simplifies dramatically:  $k=1$, so the arithmetic group is just $SL(2,\Z)$ and 
$\mathbf{h}\in SL(2,\Z)$ has finite-order $m=2,3,4,6$. The class number $h_\bK$ of the relevant cyclotomic fields is automatically $1$; then
\be
H_\bK=\begin{cases}
1 &\text{for }\bK=\bQ\ \text{i.e. }m=2,\\
GL(2,\Z)/SL(2,\Z)\cong\Z_2 &\text{for }\bK\neq\bQ\ \text{i.e. }m=3,4,6.
\end{cases}
\ee
Thus the algebras $\ca$ of \textbf{Criterion/Definition} have 7 possible Kodaira types\footnote{\ 11 types if we count the asymptotically-free cases, i.e.\! $\sD_\ca$ strong but not necessarily numerically CY.}
naturally identified with the additive Kodaira fibers with semi-simple monodromy. For $m=2$ we have a single type, $I_0^*$, while for each
$m=3,4,6$ we have two types
\be
\begin{tabular}{|cc|cc|cc|}\hline
$m=3$ & $IV$, $IV^*$ &
$m=4$ & $III$, $III^*$ & $m=6$ & $II$, $II^*$\\
\hline
\end{tabular}
\ee
We shall say that two fiber types in the same $H_\bK$ orbit (i.e.\! same $m$ for $k=1$) are each other conjugate ($I_0^*$ is self-conjugate). Passing from
a fiber type $\cf\neq I_0^*$ to its conjugate corresponds to a Kodaira quadratic transformation of the corresponding elliptic fibration. We write $\cf^*$ for the conjugate type of $\cf$
(with the convention $(I_0^*)^*=I_0^*$). See table \ref{ffffhad}.

\begin{table}
\centering
\begin{tabular}{l|ccccccc}\hline\hline
Kodaira type $\cf$ & $II^*$ & $III^*$ & $IV^*$ & $I_0^*$ & $IV$ &
$III$ & $II$\\
conjugate Kodaira fiber $\cf^*$ & $II$ & $III$ & $IV$ & $I_0^*$ & $IV^*$ & $III^*$ & $II^*$\\ 
conjugate Euler number $e(\cf^*)$ & 2 & 3 & 4 & 6 & 8 & 9 & 10\\
-(reduced conjugate intersection form) & - &$A_1$& $A_2$& $D_4$ & $E_6$ & $E_7$ & $E_8$
\\
Coulomb operator dimension $\Delta$ & $\tfrac{6}{5}$ &  
$\tfrac{4}{3}$ & $\tfrac{3}{2}$
& 2 & $3$ & $4$ & $6$\\
$b(\cf)$ & 1 & 1& 1 & 2 & 3 & 3 & 3\\
\hline\hline 
\end{tabular}
\caption{\ Coarse-grained classification of rank-1 SCFT by
Kodaira fibers of semi-simple type.
By the \emph{reduced conjugate intersection form} we mean the intersection matrix $C_i\cdot C_j$ of the irreducible curves $C_i\subset \cf^*$ which do not cross the zero-section of the elliptic fibration; the Cartan matrix of a simply-laced Lie algebra is denoted by the same symbol as the Lie algebra. Sometimes we use the Lie algebra in the third row to label the  Kodaira type of $\ca$.}\label{ffffhad}
\end{table}  

\begin{exe} The Kodaira type $\cf$ of the Dynkin algebras of type $A_2$, $A_3$ and $D_4$ is, respectively,
$II^*$, $III^*$ and $IV^*$. The Kodaira type of a tubular algebra of type $(2,2,2,2)$ is $I_0^*$. Note that in all cases $\nu+e(\cf)=12$, i.e.\! $\nu=e(\cf^*)$. 
Examples of algebras of type $IV$, $III$, $II$ are the \emph{del Pezzo algebras} defined in the next section.
\end{exe}

\subparagraph{Kodaira type of 
the category $\sC_G$.} We return to
2-CY of the form \eqref{criterion} which satisfy condition \textit{a)} but not necessarily \textit{b)}.  The auto-equivalence $G\colon \sD_\ca\to\sD_\ca$ defines a matrix $\mathbf{G}$
\be
[G S_a]=[S_b]\mathbf{G}_{ba}.
\ee
Since the category is 2-CY, we have
$G^s=S\Sigma^{-2}$ for some $s\in\bN$, so $\mathbf{g}^s=\mathbf{h}$, and $\mathbf{g}^t\mathbf{b}\mathbf{g}=\mathbf{b}$
so $\mathbf{g}\in Sp(\mathbf{b},\Z)$.
By the Kodaira type of $\sC_G$ we mean the conjugacy class of $\mathbf{g}$ in the arithmetic group $Sp(\mathbf{b},\Z)$.
Note that the Kodaira type of the Amiot cluster category $\sC_\ca$ coincides with that of the algebra $\ca$.
The Kodaira types of $\sC_G$ and $\sC_\ca$ are related by the local base change $z\to z^s$ formulae given e.g.\! in table (IV.4.1) of
\cite{miranda}. 

The Kodaira type of the rank-1 SCFT described the category $\sC_G$ coincides with Kodaira type of $\sC_G$ as described in this section.

\begin{rem} For $k=1$ we have another invariant $b\in\bN$ given by
$\det\mathbf{b}=b^2$. One gets
\be
b(\cf)+b(\cf^*)=4\quad \text{and also}\quad 1/\Delta(\cf)+1/\Delta(\cf^\prime)=1.
\ee 
\end{rem}

\medskip

The rest of this paper is dedicated to show that the 2-CY categories satisfying the above \textbf{Criterion/Definition} are in
natural bijection with the 28 interacting rank-1 SCFT listed in table 1 of \cite{Argyres3}. 

\section[Appetizer: derived category of index-1 Fano surfaces]{Appetizer: derived category of index-1 Fano surfaces\\
(a.k.a.\! Minahan-Nemeshanski SCFT)}

Our problem is to find all solutions to the conditions in the \textbf{Criterion/Definition} of the previous section. Some special  solutions can be found in a cheap way;
this holds in particular when the algebra $\ca$ is derived equivalent to a hereditary category $\ch$ (see \textsc{appendices B, D}). As we shall review in section 3 and 6, 
the hereditary case produces in rank-1\footnote{\ Together with the 4 asymptotically free $SU(2)$ SQCD with $N_f\leq3$; for simplicity in the text we limit to SCFTs.} the Argyres-Douglas models of types $A_2$, $A_3$, $D_4$ and $SU(2)$ SQCD with $N_f=4$ of respective Kodaira type $II$, $III$, $IV$ and $I^*_0$ (the rank of $K_0(\ch)$ is the Euler number of its Kodaira fiber, table
\ref{ffffhad}).
As an appetizer we present a less known class of cheap solutions where $\ca$ has wild representation-type
and Kodaira type $IV^*$, $III^*$ and $II^*$.

\subsection{Algebraic-Geometric viewpoint}

Suppose $X$ is a smooth projective surface which is Fano
(i.e.\! $-K_X$ is ample) of index-1
(i.e.\! $-K_X$ is primitive in the Picard lattice) and such that the
anticanonical model of $X$ is a
hypersurface of degree $d$ in some weighted projective space
$\bP(w_0,w_1,w_2,w_3)$. Since $X$ is an index-1 Fano we have the equality
\be
w\equiv\sum_{i=0}^3w_i=d+1.
\ee

We claim that each surface $X$ with the above properties yields a (derived equivalence class of) triangular algebra $\ca$ of global-dimension 2 such that their derived category $\sD_\ca\equiv D^b\mathsf{mod}\,\ca$ is fractional Calabi-Yau of dimension
\be\label{fffrqa}
\hat c=\frac{a}{b}=\frac{2(2d-w)}{d}\equiv \frac{2(d-1)}{d},\quad\text{i.e.}\quad S^d\cong \Sigma^{2(d-1)}\ \text{in }\sD_\ca.
\ee
Note that this implies
\be
\Sigma^2\cong (S^{-1}\Sigma^2)^d\quad\Rightarrow\quad \Sigma^2\cong\mathrm{Id}\ \text{in }\sC_\ca,
\ee 
so that the Amiot cluster category $\sC_\ca$ is symmetric. As discussed around eqn.\eqref{kqwaz}, this implies that the Coulomb dimensions $\Delta_i\in\bN$. 

\medskip

By classification, the surfaces $X$ satisfying the above conditions are del Pezzo, in fact $\bP^2$ blown-up in $k=6,7,8$ points in very general position (\!\!\cite{dolgachev} especially \S.\,8.3.2). See
table \ref{delpezzo}. 

\begin{table}
\begin{tabular}{c|ccccc}\hline\hline
$k\equiv\#$ blown-up points & $(w_0,w_1,w_2,w_3)$ & $d$ & $\hat c$ & root lattice & $\Delta\equiv(1-\hat c/2)^{-1}\equiv d$ \\\hline 
6 & $(1,1,1,1)$ & $3$ & $4/3$ & $E_6$ & 3\\
7 & $(1,1,1,2)$ & $4$ & $3/2$ & $E_7$ & 4\\
8 & $(1,1,2,3)$ & $6$ & $5/3$ & $E_8$ & 6\\\hline\hline
\end{tabular}
\caption{Index-1 Fano surfaces whose anti-canonical model is a weighted projective hypersurface.
$\hat c$ is the CY dimension of the associated algebra.
$\Delta$ is the Coulomb branch dimension of the corresponding rank-1 SCFT, while the root lattice described in \cite{dolgachev} \S.8.2 (see fifth column) dictates its flavor symmetry. The root lattice in the fifth column corresponds to the Kodaira type of the associated algebra $\ca$ (cfr.\! third row of table \ref{ffffhad}).}
\label{delpezzo}\end{table}

\medskip

The claim follows from a few well-known facts that we now recall. 
\begin{defn} Let $\sD$ a ($\C$-linear, Hom-finite) triangle category.
An ordered sequence of objects $\{E_1,E_2,\dots, E_r\}$ of $\sD$ is said to be \emph{strongly exceptional} iff
\be
\begin{aligned}
\sD(E_i,E_j[k])=0&&&\text{for }k\neq0\ \text{and all }i,j\\
\sD(E_i,E_i)\cong \C,&&& \sD(E_i,E_j)=0\ \text{for }i>j.
\end{aligned}
\ee
The sequence $\{E_1,E_2\dots,E_r\}$ is called \emph{full} iff  it generates the triangle category $\sD$, i.e.\! if the smallest full triangle sub-category of $\sD$ containing the $E_i$ is $\sD$:
\be
\sD\cong\big\langle E_1,E_2,\cdots, E_r\big\rangle.
\ee
\end{defn}

If $\{E_1,E_2,\cdots, E_r\}$ is a full strongly exceptional sequence, the object 
\be
\ce\equiv\bigoplus_{i=1}^r E_i\in \sD
\ee
 is a tilting object of $\sD$ with the special property  that its endo-algebra  $\cb\cong\mathrm{End}(\ce)$ is (finite-dimensional and) triangular. The basic statement of tilting theory \cite{tiltingX} is the  equivalence of triangle categories
\be
\sD\cong D^b\mathsf{mod}\,\cb. 
\ee

\begin{fact}[e.g.\! \cite{orlov1,orlov2}]\label{fact1} Let $X$ be a del Pezzo 
surface obtained by blowing up $k$ points in very general position in $\bP^2$; we write $\mathsf{coh}\,X$ for the category of the coherent sheaves on $X$, and
$\sD_X=D^b\mathsf{coh}\,X$ for its bounded derived category. It follows from the description of $X$ as the blow-up of the plane that $\sD_X$ contains several full strongly exceptional sequences {\rm\cite{orlov2}}. For instance a convenient one is
\begin{multline}\label{jjjqawe}
\big\{E_1,E_2\cdots, E_{k+3}\big\}=\\
=\Big\{\co_{\ell_1}(-1),\co_{\ell_2}(-1),\cdots,
\co_{\ell_k}(-1), \pi^*\co[1], \pi^*\ct(-1)[1],\pi^*\co(2)[1]\Big\}
\end{multline} 
where $\pi\colon X\to \bP^2$ is the obvious dominant morphism, $\ell_a$ ($a=1,2,\dots, k$) are the exceptional $(-1)$ lines, and $\ct$ is the tangent bundle of $\bP^2$.
\end{fact} 
\begin{proof} Computation of $\dim\sD(E_i,E_j[k])$, see \textsc{appendix C}.
\end{proof}

The triangular algebra $\cb\equiv \mathrm{End}(\ce)$ has ``Cartan matrix'' (in the 4d/2d language \cite{CV92}: the inverse of the Stokes matrix $S$)
\be
S^{-1}_{ij}=\dim \sD(E_i,E_j)
\ee
which is upper-triangular with 1's on the main diagonal
(the reason why $\cb$ is called ``triangular''); see \textsc{appendix C} for explicit expressions.
 Let us draw the acyclic quiver (with relations) of the algebra $\cb=\mathrm{End}(\ce)$ for the full strong exceptional sequence \eqref{jjjqawe}
 \bigskip

 \be\label{uuuqw12}
 \begin{gathered}
 \xymatrix{\co_{\ell_1}(-1) \ar[r] & \pi^*\co[1]\ar@<0.4ex>[dd]\ar[dd]\ar@<-0.4ex>[dd]\\
 \vdots && \pi^*\co(2)[1]\ar@{..>}@<0.4ex>[ul]\ar@{..>}[ul]\ar@{..>}@<-0.4ex>[ul] \ar@/_4pc/@{..>}@<0.2ex>[ull]\ar@/_4pc/@{..>}@<-0.2ex>[ull]
 \ar@/^4pc/@{..>}@<0.2ex>[dll]\ar@/^4pc/@{..>}@<-0.2ex>[dll]\\
 \co_{\ell_k}(-1)\ar[uur] & \pi^*\ct(-1)[1]\ar@{..>}[uul]\ar@{..>}[l]\ar@<0.4ex>[ur]\ar[ur]\ar@<-0.4ex>[ur]
 }
 \end{gathered}
 \ee
 \bigskip
 
\noindent  where (as always) dashed arrows stand for minimal relations in the opposite direction.

\begin{fact}\label{fact2} Let $Y$ be a smooth hypersurface of degree $d$ in the weighted projective space $\bP(w_0,w_1,w_2,w_3)$
satisfying the relation
\be\label{jjz125m}
w\equiv\sum_i w_i=d+1,
\ee
i.e.\! $Y$ is (in particular) a index-1 Fano surface.
Let $(E_1,\cdots,E_r)$ be a full strong exceptional sequence in $D^b\mathsf{coh}\,Y$
 whose last object $E_r\equiv \cl[s]$  is a \emph{shift of a line bundle} $\cl$.
 Then the full triangle sub-category
 \be
 \sD_\ca\equiv \big\langle E_1,E_2,\cdots, E_{r-1}\big\rangle \subsetneq \sD_Y
 \ee 
is fractional Calabi-Yau of dimension $a/b=2(2d-w)/d$. I.e.
\be
S^d\cong \Sigma^{2(2d-w)}\quad\text{in }\sD_\ca.
\ee
Moreover one has
\be
\sD_\ca\cong D^b\mathsf{mod}\,\ca\quad\text{where }\ca=\mathrm{End}(\ce^*),\quad \ce^*=\bigoplus_{i=1}^{r-1} E_i.
\ee
\end{fact}
\begin{proof} Applying to the sequence the shift $\Sigma^{-s}$ and twisting it by the line bundle $\co_Y\cl^{-1}$ we may assume $E_r=\co_Y\equiv \pi^*\co(1)$.
Then, in view of eqn.\eqref{jjz125m}, the statement is a special case of \textbf{Corollary 4.2} in \cite{kuz}.
\end{proof}

Let $X$ be a del Pezzo surface which is also a smooth hypersurface in some $\bP(w_0,w_1,w,_2,w_3)$
(then automatically $w=d+1$). We apply \textbf{Fact \ref{fact2}} to the full strong exceptional sequence in eqn.\eqref{jjjqawe}.
 The quiver $\mathring{Q}$ of the algebra
$\ca=\mathrm{End}(\ce^*)$ may be obtained from the quiver of the algebra $\cb=\mathrm{End}(\ce^*\oplus\cl[s])$, eqn.\eqref{uuuqw12},
 by deleting in \eqref{uuuqw12}
the rightmost node associated to the shifted line bundle $\cl[1]\equiv\pi^*\co(2)[1]$.  
 
 \be\label{uuuqw13}
\mathring{Q} = \begin{gathered}
 \xymatrix{\bullet_1\ar[rr] && \omega\ar@<0.4ex>[dd]\ar[dd]\ar@<-0.4ex>[dd]\\
 \vdots && \\
 \bullet_k\ar[uurr] && \alpha\ar@{..>}[uull]\ar@{..>}[ll]
 }
 \end{gathered}
 \ee
$\ca$ is then a finite-dimensional triangular algebra of global-dimension 2
whose derived category $\sD_\ca$ is fractional CY of the dimension in eqn.\eqref{fffrqa}. We give a name to the algebras so constructed.

\begin{defn} A \emph{del Pezzo algebra of type} $E_k$ ($k=6,7,8$) is a global-dimension 2 algebra $\ca_k$ of the form 
\be
\ca_k\cong\mathrm{End}_{\sD_{X_k}}\!\!\left(\bigoplus_{i=1}^{k+2}E_i\right),\qquad \left[\begin{aligned} &X_k\ \text{a smooth del Pezzo}\\ 
&\text{surface of degree }9-k,\end{aligned}\right.
\ee
where $\{E_1,E_2,\cdots, E_{k+2},\co_{X_k}\}$ is a full strong exceptional sequence in $\sD_{X_k}$.
del Pezzo algebras $\ca_k$ depend on continuous parameters, namely the complex moduli of $X_k$. Two del Pezzo algebras associated to the same surface $X_k$ are derived equivalent.
\end{defn}

The del Pezzo cluster categories 
\be
\sC_{E_k}= \big(D^b\mathsf{mod}\,\ca_k\big/(S^{-1}\Sigma^2)^\Z\big)_\text{tr.hull}\quad k=6,7,8,
\ee
 describe (respectively) 
the Minahan-Nemesjansky SCFT of type $E_6$, $E_7$ and $E_8$ whose existence here we deduce from the structure of the derived category of sheaves on del Pezzo surfaces.
The quivers of these theories are given by (the mutation class of) the completion of the quiver $Q$ \eqref{uuuqw13} obtained from $\mathring{Q}$ by making solid the dashed lines; $Q$ is equipped with the superpotential
\be\label{superW}
W_\mathrm{lin}=\sum_a \rho_a r_a
\ee
where $\rho_a$ is the arrow which replaces the $a$-th dashed arrow and $r_a$ the associated minimal relation.\footnote{\ Cfr.\! \textbf{Theorem 6.12} of \cite{comke}.}
We have recovered the well known BPS-quiver description of these SCFT \cite{Alim:2011kw,Cecotti:2013sza}.

\begin{rem} Had we started with the more widely used full strongly exceptional sequence
\be
\Big\{\co_{\ell_1}(-1),\co_{\ell_2}(-1),\cdots,
\co_{\ell_k}(-1), \pi^*\co[1], \pi^*\co(1)[1],\pi^*\co(2)[1]\Big\},
\ee
instead of the one in eqn.\eqref{jjjqawe} we would have ended with the quiver
 \be\label{uuuqw1345}
 \mathring{Q}^\prime=\begin{gathered}
 \xymatrix{\bullet_1\ar[rr] && \omega\ar@<0.4ex>[dd]\ar[dd]\ar@<-0.4ex>[dd]\\
 \vdots && \\
 \bullet_k\ar[uurr] && \alpha\ar@<0.2ex>@{..>}[uull]
 \ar@<-0.2ex>@{..>}[uull]\ar@<-0.2ex>@{..>}[ll]
 \ar@<0.2ex>@{..>}[ll]
 }
 \end{gathered}
 \ee 
 whose completion $Q^\prime$ is equivalent to $Q$ by the elementary mutation at the node $\omega$.
\end{rem}

\begin{rem} The algebra $\ca_k=\mathrm{End}(\ce^*_k)$ makes sense for all $k\leq 8$, except that $\sD_{A_k}$ is not fractional Calabi-Yau for $k\leq 5$. The ``minimal deviation'' from being fractional CY is obtained for $k=5$, i.e.\! $X$  is the blow-up of $\bP^2$ at 5 generic points;
$X$ is no longer a hypersurface, but it is still a complete intersection of two quadrics \cite{dolgachev}. Thus formally $\Delta=2$, while the would-be flavor group is $E_5\equiv SO(10)$.
This has a physical meaning: the quiver \eqref{uuuqw13}
with $k=5$ belongs to the mutation class of $SU(2)$ SQCD with $N_f=5$, a QFT affected by Landau poles, so not defined as a QFT in its own right. However, as discussed on page \pageref{curious}, $SU(2)$ SQCD with $N_f=5$ makes sense as a low-energy effective theory up to some cut-off (Ringel's `curious fact', see footnote \ref{curious}).
 The presence of Landau poles obviously spoils the CY property; a part for that, $\Delta=2$ and $F=SO(10)$ are the correct answers for this formal QFT. 
\end{rem}

\subsection{Physicist's viewpoint}

The categorical constructions used in the previous subsection were originally introduced in the context of (homological) mirror symmetry, i.e.\! they are a rephrasing of the 4d/2d correspondence of \cite{CNV}. It is just the relation between the BPS-brane categories of the gauged linear $\sigma$-model with target the del Pezzo surface $X$ and   
the Landau-Ginzburg (2,2) model with quasi-homogeneous superpotential the equation $W(X_0,X_1,X_2,X_3)=0$ of $X$ in the weighted projective space,
which flows in the IR to a (2,2) SCFT with $\hat c=2(d-1)/d$. Since $X$ is Fano, the $\sigma$-model is not conformal, and flowing it to the IR we loose chiral primaries, which is the physical rationale for deleting the node of the quiver which gets ``massive''. More or less by definition, the IR fixed point has Calabi-Yau dimension $\hat c$.

\section{Review of \cite{unpublished} (in rank 1)}

We start by reviewing the small part of the unpublished work \cite{unpublished} which refers to rank-1 theories (for a related discussion see \cite{Cecotti:2013sza}). To distinguish the several rank-1 QFTs we either refer to them by the number of the corresponding entry in table 1 of Argyres \emph{et al.} \cite{Argyres3}  or by the pair  $(\Delta, F)$ where $\Delta$ is the dimension of their Coulomb branch and $F$ the maximal flavor symmetry group;
an asymptotically-free theory will be written $(a\!f,F)$.

\subsection{The 16 rank-1 2-acyclic QFT}\label{4d2d}

We start by considering the 2-CY categories with the properties required in our \textbf{Criterion/Definition} which are Amiot cluster categories, i.e.\! such that  $G=S\Sigma^{-2}$ in eqn.\eqref{criterion}. They are in particular 2-acyclic i.e.\! have the BPS-quiver property.

We saw in \S.2 that the 2d quantum monodromy $H$ of an algebra $\ca$ with $\mathrm{gl.dim}\,\ca\leq2$ which satisfies conditions \textit{a)},\textit{b)} of \textbf{Criterion/Definition},
such that $\sD_\ca$ is fractional Calabi-Yau has special spectral properties. We recall the definitions.
Let $P_i$ ($i=1,\dots, \nu$) be the indecomposable projective modules of $\ca$. The matrix $\dim\sD_\ca(P_i,P_j)$ has determinant $\pm1$, so its inverse exists and we have \cite{CV92}
\be
S^{-1}_{ji}= \dim\sD_\ca(P_i,P_j),\qquad B=S-S^t,\qquad H=(S^{-1})^tS.
\ee
Under the conditions of the \textbf{Criterion/Definition}, 
$B$, $H$ are $\Z$-equivalent to
\be\label{numcon}
B= m\begin{pmatrix} 0 & 1\\
-1 & 0\end{pmatrix}\bigoplus \boldsymbol{0},\ \ m\in\Z_{\geq1}\qquad H=\overline{H}\bigoplus\boldsymbol{1},\ \  \overline{H}\in SL(2,\Z)\ \text{and torsion}.
\ee
We shall say that an algebra $\ca$
with $\mathrm{gl.dim}\,\ca\leq 2$
is \emph{numerically CY of rank-1} iff its matrix $\dim\sD_\ca(P_i,P_j)$
satisfies \eqref{numcon}. This implies that for certain integers $a$, $b$ we have $[S^b X]=[\Sigma^a X]$ for all $X\in\sD_\ca$, but in general the equality may be not lifted from the Grothendieck group to the derived category, so the
notion of numerical CY is weaker than fractional CY.

A numerically CY algebra $\ca$
with $\mathrm{rank}\,B=2$
has a well-defined semi-simple
Kodaira type, i.e.\! the conjugacy class of $\overline{H}$ in $SL(2,\Z)$.

As a first step we can look for triangular algebras which are numerically CY. This means we start with an integral upper-triangular matrix $S$ with 1's on the diagonal, and solve the
numerical conditions above seen
as Diophantine equations in the entries of $S$. These are precisely the same Diophantine equations as in the classification program of 2d (2,2) SCFT \cite{CV92}. Modulo some subtlety we shall dwell momentarily, if a solution $S$ to these equations correspond to an actual (2,2) SCFT, then an actual fractional CY also exists in the form of its brane category just as it did in the del Pezzo case of section 3 (seen from the physics side).
\medskip

Ref.\!\!\cite{unpublished} aims to 
solve the Diophantine conditions by the interplay of the 4d and 2d physics.
Let us describe the results relevant for rank-1.
Given a (possibly empty) set of positive integers
$A=\{a_1,a_2,\cdots,a_f\}$ and 
an integer $q$
we define the quiver $Q(A;q)$ to be
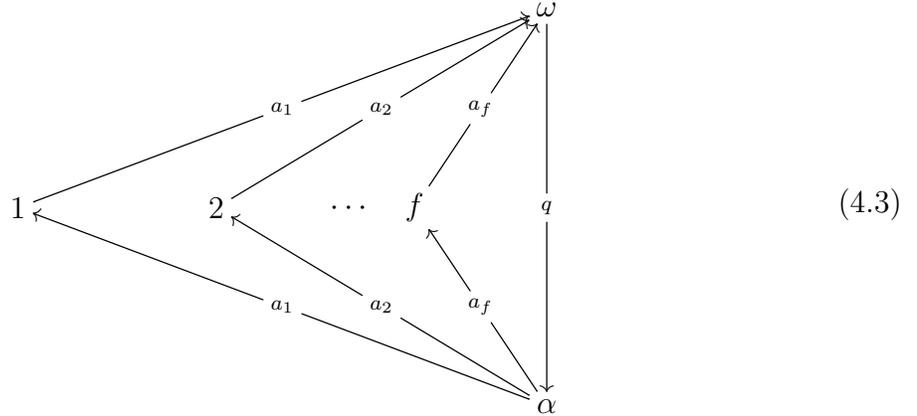
\be\label{Qaqquivers}
\begin{gathered}
\begin{xy} 0;<1pt,0pt>:<0pt,-1pt>:: 
(200,0) *+{\omega} ="0",
(200,150) *+{\alpha} ="1",
(150,75) *+{f} ="2",
(125,75) *+{\cdots} ="3",
(75,75) *+{2} ="4",
(0,75) *+{1} ="5",
"0", {\ar|*+{\scriptstyle q}"1"},
"2", {\ar|*+{\scriptstyle a_f}"0"},
"4", {\ar|*+{\scriptstyle a_2}"0"},
"5", {\ar|*+{\scriptstyle a_1}"0"},
"1", {\ar|*+{\scriptstyle a_f}"2"},
"1", {\ar|*+{\scriptstyle a_2}"4"},
"1", {\ar|*+{\scriptstyle a_1}"5"},
\end{xy}
\end{gathered}
\ee
where the symbol $-\,a_j\!\to$ stands for $a_j$ directed arrows between the correspondent  pair of nodes ($a_j$ negative means $|a_j|$ arrows in the opposite direction). The rank of its exchange matrix $B_{ij}$ is 2. Then, if the quiver describes a $\cn=2$ QFT, the rank of its flavor group is
\be
\mathrm{rank}\,F=f.
\ee
The type $(p,\bar q)$ of the quiver $Q(A;q)$ is defined to be
\be
p=\sum_{i=1}^f a_i^2,\qquad 0\leq \bar q=\begin{cases} p-q &\text{if }q<0\\
q & \text{if }p-q<0\\
\min\{q,p-q\} &\text{otherwise}.
\end{cases}.
\ee 
We have the mutation equivalence
\be
Q(A,q)\sim Q(A,p-q),
\ee
so, without loss, we may assume $q\equiv \bar q$.

\begin{table}
\begin{center}
\begin{tabular}{c|cccccc}\hline\hline
$(p,\bar q)$ & (0,1) & (1,0)  & (2,1) & (0,2)  & (1,2) & (2,2)\\
$\cf_{(p,q)}$ & $II$ & $III$ & $IV$ &
 &  & 
\\
$(\Delta,F)$ & $(\tfrac{6}{5},-)$ & $(\tfrac{4}{3},SU(2))$ & $(\tfrac{3}{2},SU(3))$& $(a\!f,-)$ & $(a\!f,SO(2))$& $(a\!f,SO(4))$\\
$\#$ & 28 & 27 & 26 & - & - &-\\
$a/b$ & 1/3 & 2/4 & 2/3 &
\\\hline
$(p,\bar q)$  & (3,2) & (4,2) & (6,3) & (7,3) & (8,3) &\\
$\cf_{(p,q)}$ &  & $I_0^*$ & $IV^*$ &
$III^*$ & $II^*$ & \\
$(\Delta,F)$  & $(a\!f,SO(6))$ &
$(2, F)$ & $(3,F)$ & $(4,F)$ & $(6,F)$\\
$\#$ & - & 23,24,25 & 19,20 & 12,13 & 1,2,3\\
$a/b$ & & 2/2 & 2/3 & 2/4 & 2/6
\\\hline\hline
\end{tabular}
\caption{\label{allowedtypes} The values of $(p,\bar q)$ such that the quiver $Q(A,q)$ satisfies the requirements of 4d/2d correspondence. For a SCFT $\Delta$ is the conformal dimension. $\#$ is the entry number
of the (mass-deformed) SCFT in table 1 of ref.\!\cite{Argyres3}. For a SCFT $a/b$ is the fractional CY dimension.}
\end{center}
\end{table}

To each $Q(A,q)$ quiver we associate some triangular algebra $\ca\equiv \ca(A,q)$.
If $A=\varnothing$, $Q$ is the acyclic $q$-Kronecker quiver (unique in its mutation class). In this case $\ca$ is just the hereditary algebra $\C Q$ (global dimension 1).
If $A\neq\varnothing$ we have choices. A first possibility is to replace the $q$ vertical arrows by dashed ones; the remaining solid arrows form an acyclic quiver $\mathring{Q}^{(1)}$. The dashed arrows are taken to represent (generic) minimal relations generating an admissible ideal $I^{(1)}$ in the path algebra $\C\mathring {Q}^{(1)}$.
The algebra $\ca^{(1)}= \C\mathring Q^{(1)}/I^{(1)}$ has global dimension $\leq 2$.
A second possibility is to dash the arrows starting at the node $\alpha$; we get a different quiver $\mathring{Q}^{(2)}$, ideal $I^{(2)}$ and algebra $\ca^{(2)}= \C\mathring Q^{(2)}/I^{(2)}$. Note that the matrices $S$ for $\ca^{(1)}$, $\ca^{(2)}$ are $\Z$-equivalent, so the spectral conditions does not distinguish between the two.

One shows that the Diophantine conditions are satisfied if and only if the type $(p,q)$ of $Q(A,q)$ is as in table \ref{allowedtypes} where for completeness we reinstated the four rank-1 asymptotically-free theories which satisfy weaker spectral conditions.\footnote{\ The spectral conditions of \cite{CNV} state that the
2d quantum monodromy $H$ should have spectral radius 1 and, in the $a\!f$ case \cite{CV92}, a unipotent part consistent with the $SL_2$-orbit theorem of Hodge theory \cite{sl2}. Equivalently the conjugacy classof $\overline{H}$ is in the Kodaira list.} 
Moreover, all algebras of given type $(p,q)$ have the same Kodaira type $\cf_{(p,q)}$ see table. Note that the Euler number of the Kodaira fiber is
\be
e(\cf_{(p,q)})=p+2.
\ee
We say that a quiver $Q(A,q)$ whose type is in the table is \emph{admissible}.
In the table we list some of the properties of the putative QFT described by a given quiver. $F$ stands for some unspecified flavor group which depends on the specific set $A$ not just on the type of the $Q(A,q)$ quiver. Note that  an integer $p\leq 3$ can be written in a unique way as a sum of squares, $p=1+\cdots+1$, so for $p\leq 3$ the type determines uniquely the flavor group which is written in the table. Note that all rank-1 models with
$\Delta <2$ or asymptotically-free are
 covered by table \ref{allowedtypes}. 
 \medskip

In the case $A=\{1^f\}$ the numerical CY algebras we got are well known:
\begin{itemize}
\item for (0,1) and (1,0) hereditary of Dynkin type $A_2$ and $A_3$;
\item for (2,1) a tilted-algebra of $D_4$ type;
\item for $(p,q)=(N,2)$
$\ca^{(1)}$ is a \emph{canonical} algebra of type $\{2,2,\cdots,2\}$ ($N$ 2's) and $\ca^{(2)}$ a \emph{squid} algebra of the same type (the two being derived equivalent);
\item for $(p,q)=(p,3)$ $\ca^{(2)}$ is a \emph{del Pezzo} algebra of type $E_p$
(section 3).
\end{itemize}
In particular, in all these cases there exists and ideal $I$
such that the algebra $\C\mathring{Q}/I$ is actually fractional Calabi-Yau of the given dimension, not just numerically so.

\subparagraph{Subtleties for $A\neq\{1^f\}$.} The case with $A\neq\{1^f\}$
is subtler in many respects.
The first example is $Q(\{2\},2)$ i.e.\! the Markoff quiver
which is well known to correspond to $SU(2)$ $\cn=2^*$ \cite{Alim:2011kw}. \emph{A priori} it is not clear that one may find an ideal $I$ so that the algebra $\C\mathring{Q}/I$ is fractional CY and not just numerically CY.
The simplest choice of relations
\be
\begin{gathered}\xymatrix{\bullet\ar@<0.4ex>[rr]^{a_1}\ar@<-0.4ex>[rr]_{a_2}&&\bullet\ar@<0.4ex>[rr]^{b_1}\ar@<-0.4ex>[rr]_{b_2}&&\bullet\ar@<0.8ex>@{..>}@/^1.3pc/[llll]\ar@<-0.8ex>@{..>}@/_1.3pc/[llll]}\end{gathered}\qquad b_1a_1=b_2a_2=0,
\ee 
which yields
a gentle algebra, is certainly \emph{not} fractional CY. 
On the other hand a good 2d (2,2) theory associated to the Markoff quiver exists, namely the Landau-Ginzburg model with superpotential the Weierstrass function $\wp(X)$
(with identifications, $X\sim X+1$,
$X\sim X+\tau$). Correspondingly there exists a 4d $\cn=2$ SCFT,
namely $SU(2)$ $\cn=2^*$.\footnote{\ Note that the coupling spaces of both 4d and 2d models (apart for a mass scale) are the moduli of elliptic curves.}
In fact it is known that there exists a
$W$ on the Markoff quiver which produces a nice 2-CY cluster category\footnote{\ With the further subtlety that the path algebra of the quiver should be replaced by the completion of the path algebra with respect the $\mathfrak{m}$-topology, where $\mathfrak{m}$ is the arrow ideal.}. The superpotential $W_\text{lin}$ in eqn.\eqref{superW}, linear in the $\rho_a$, is a \emph{singular} limit of the good one which contains also terms of higher order in the $\rho_a$'s. This phenomenon should be compared with the special-geometric viewpoint that we shall review in .... The geometry for $\cn=2^*$ is a desingularization of a singular limit of the geometry for 
$SU(2)$ with $N_f=4$. The algebra with $W_\text{lin}$ may be seen as the singular limit, while adding the higher term to $W$ the desingularization process.

The physical reason why the model is subtle, is the existence of hypermultiplets which are everywhere light on the Coulomb branch; switching the mass deformation off, this means that we have a Higgs branch fibered over the generic point in the Coulomb branch (an enhaced Coulomb branch in the terminology of \cite{Argyres6}).
We shall give some more detail on this topic below.
\medskip

The story should remain true for all admissible quivers $Q(A,q)$ with
$A=\{2,1^{f-4}\}$;    
we call the associated SCFT \emph{Argyres-Wittig models} since they were first described in ref.\!\cite{wittig}. This is quite natural, from the geometrical side, since the situation is locally the same as for the $\cn=2^*$ model.
In fact, the SCFT are known to exist.
\medskip

There remains a last admissible
quiver $Q(\{2,2\},3)$.
It corresponds to entry 3 in table 1 of \cite{Argyres3}. However this entry is shaded in color in \cite{Argyres3} since its effective existence is doubtful.
From our categorical viewpoint, its existence also looks problematic, in the sense, that the numerical CY algebra is not expected to be truly fractional CY, and the desingularization which worked in the Argyres-Wittig cases is hardly sufficient to regularize the 2-CY category.

The main physical reason to doubt its existence \cite{Argyres3} is that no meaningful RG flow seems to originate from this would-be QFT.
In the present context, there is an element with the same flavor. All admissible quivers $Q(A,q)$ -- except $Q(\{2,2\},3)$ -- have the property that if we delete a node $\neq \alpha,\beta$ we either get another
admissible $Q(A,q)$ quiver or the quiver
of a low-energy effective theory with a cut-off such as $SU(2)$ coupled to five $\boldsymbol{2}$'s, or to one $\boldsymbol{3}$ and one $\boldsymbol{2}$. In all cases the spectral radius of the 2d monodromy remains $1$.
This is no true for $Q(\{2,2\},3)$; indeed $Q(\{2\},3)$ has spectral radius $>1$, so is not expected to correspond to a QFT, not even a formal one. With these \emph{caveats} we keep $Q(\{2,2\},3)$ in our tables.
\medskip

The reader may wonder how general is our ansatz $Q(A,q)$ for the quiver. We carried over extensive searches for quivers with the right spectral properties, an found (of course) lots of them; at closer inspection they all turned out to be mutation equivalent to one in the form $Q(A,q)$.
We believe that this is indeed true in general.
The correspondence with singular fiber configurations of elliptic surfaces provides  further evidence for this expectation.  

\subsubsection{Comparison with
rational elliptic surfaces} The quivers
$Q(A,q)$  which satisfy the above spectral requirements are easily seen in one-to-one correspondence with the allowed configurations of Kodaira singular fibers in a rational elliptic surface with section\footnote{\ Tables of allowed fibers configurations for rational elliptic surfaces may be found in refs.\!\cite{bookMW,per,mira}.} subject to two conditions: \textit{i)} there is precisely one fiber with additive reduction (the fiber at $\infty$) -- of the conjugate type $\cf^*$
with respect to the SCFT -- all other singular fibers being multiplicative (i.e.\! semi-stable); \textit{ii)} the poles of Kodaira's functional invariant $\mathscr{J}\!(z)$ have orders which are perfect squares: this is essentially the condition called `Dirac quantization of charge' in ref.\!\!\cite{Argyres3}. 
The map (quiver data) $\leftrightarrow$ (Kodaira
fiber configuration) is
\be\label{mmzaqw}
A\ \longleftrightarrow\ \{\cf^*; I_{a_1^2},I_{a_2^2},\cdots, I_{a_f^2}, I_1^2\},
\ee
where $\cf^*$ is the unique additive (non semi-stable) Kodaira fiber with Euler characteristic $e(\cf^*)=10-p=12-e(\cf_{(p,q)})$. In particular the allowed types $(p,\bar q)$ are in one-to-one correspondence with the
un-stable Kodaira fibers $\cf^*$ which may appear in a \emph{rational} elliptic surface 
\be
p = 10-e(\cf^*),\qquad
q=b(\cf)\equiv\begin{cases} 3 & e(\cf^*)<6\\
2 & e(\cf^*)=6+b\\
1 & e(\cf^*)>6 \ \text{and }r=0,
\end{cases}
\ee
where $r$ is the order of pole of $\mathscr{J}\!(z)$ at infinity.
For SCFT $r=0$, and we limit to this case.
The `Dirac quantization' constraint is automatically satisfied by this class of 2-acyclic quivers.
Indeed it is the basic ingredient which led to the ansatz $Q(A,q)$ for the quiver in the first place \cite{unpublished}.

From the point of view of the elliptic surface, the 
models with $A\neq\{1^f\}$ are obtained by making $a_i^2$ singular fibers $I_1$ to coalesce to form 
singular fibers of type $I_{a_i^2}$.
The corresponding algebras are then singular limits at the boundary in the complex moduli space; geometrically, the singularities may be resolved by changing the birational model.
Morally, this reflects the process of 
regularizing the completed algebra
by adding higher order terms in the superpotential $W$. 

The relation with the geometry of the rational elliptic surfaces may be described directly, without reference to the physical considerations of \cite{Caorsi:2018ahl}, as we are going to show.

\subsubsection{Relation with the derived category of rational elliptic surfaces}

There is a strange ``duality'' between the algebras with $A=\{1^f\}$ of conjugate semi-simple Kodaira type, $\cf$ and $\cf^*$. Their numbers $\nu$ of quiver nodes, fractional CY dimensions $\hat c$, and $b$-coefficients satisfy
\be
\nu(\cf)+\nu(\cf^*)=12,\quad \hat c(\cf)+\hat c(\cf^*)=2,\quad b(\cf)+b(\cf^*)=4.
\ee
 One feels that the corresponding derived categories $\sD_\cf$ and $\sD_{\cf^\vee}$ cry to be paired up in a deeper structure. The feeling is correct.

Let $Y$ be a (smooth) rational elliptic surface with a section. It may be seen as $\bP^2$ blown-up in 9 points \cite{miranda}, one of the exceptional $(-1)$ lines playing the role of the base of the fibration. 
For $0\leq k\leq 8$ we consider the following full strong exceptional sequence in $D^b\mathsf{coh}\,Y$
\begin{multline}\label{jjjqawe999}
\big\{E_1^{(k)},E_2^{(k)}\cdots, E_{12}^{(k)}\big\}=\\
=\Big\{\co_{\ell_1}(-1),\cdots,
\co_{\ell_k}(-1), \pi^*\co[1], \pi^*\ct(-1)[1],\pi^*\co(2)[1],
\co_{\ell_{k+1}}[1], \cdots, \co_{\ell_{9}}[1]\Big\}
\end{multline}
Let $\ca^{(k)}\equiv\mathrm{End}(\ce^{(k)})$ be the endo-algebra of the tilting object
\be
\ce^{(k)}=\bigoplus_{i=1}^{12}E_i^{(k)}.
\ee
whose quiver is shown in figure \ref{endoo}. For $k=6,7,8$, the two complementary full subquivers
over the nodes $\{E_1^{(k)},\cdots, E_{k+2}^{(k)}\}$ and, respectively,
$\{E_{k+3}^{(k)},\dots, E_{12}^{(k)}\}$, are the quivers of the $A=\{1\}$ algebras of complementary Kodaira types $\cf$ and $\cf^\vee$ with 
$\cf=IV^*, III^*$ and $II^*$ respectively.
We consider the two full triangulated subcatgories
\begin{align}
\mathscr{A}_k &=\big\langle E_1^{(k)},\cdots,E_{k+2}^{(k)}\big\rangle\subset \sD^b\mathsf{coh}\,Y,  && \mathscr{A}_k \cong D^b\mathsf{mod}\,\ca(\{1^k\},3)
\\
 \mathscr{B}_k&=\big\langle E_{k+3}^{(k)},\dots, E_{12}^{(k)}\big\rangle\subset \sD^b\mathsf{coh}\,Y, && \mathscr{B}_k \cong D^b\mathsf{mod}\,\C Q_k 
\end{align}
(with $Q_6=D_4$, $Q_7=A_3$, $Q_8=A_2$).
$\mathscr{A}_k$, $\mathscr{B}_k$  yield a semiorthogonal decomposition of $D^b\mathsf{coh}\,Y$
 \be
 D^b\mathsf{coh}\,Y=\big\langle \mathscr{A}_k,\mathscr{B}_k\big\rangle,
 \ee
 so, in a sense, the derived category $D^b\mathsf{coh}\, Y$ is obtained by gluing the two fractional Calabi-Yau categories of complementary dimensions.
 
 Geometrically, this correspond to the operation in \cite{Caorsi:2018ahl} of gluing together the special geometries of two complementary rank-1 SCFT to get a compact geometry $Y$, easier to study than the open geometry of a single SCFT.
 
 We expect that $D^b\mathsf{coh}\,Y$ 
to have a semiorthogonal decomposition into two  derived categories of coherent sheaves on weighted projective lines of tubular type $\{2,2,2,2\}$. 

\begin{figure}
\begin{equation*}
 \begin{gathered}
 \xymatrix{E^{(k)}_1 \ar[rr] && E^{(k)}_{k+1}\ar@<0.4ex>[dd]\ar[dd]\ar@<-0.4ex>[dd]&& E^{(k)}_{k+4}\ar@{..>}@<0.2ex>[ll]\ar@{..>}@<-0.2ex>[ll]\ar@{..>}@/^2.4pc/[ddll]\\
 \vdots &&& E^{(k)}_{k+3}\ar@{..>}@<0.4ex>[ul]\ar@{..>}[ul]\ar@{..>}@<-0.4ex>[ul] \ar@/_4.5pc/@{..>}@<0.2ex>[ulll]\ar@/_4.5pc/@{..>}@<-0.2ex>[ulll]
 \ar@/^4.5pc/@{..>}@<0.2ex>[dlll]\ar@/^4.5pc/@{..>}@<-0.2ex>[dlll]\ar[ur]\ar[dr]&\vdots\\
E^{(k)}_k\ar[uurr] && E^{(k)}_{k+2}\ar@{..>}[uull]\ar@{..>}[ll]\ar@<0.4ex>[ur]\ar[ur]\ar@<-0.4ex>[ur]&&E^{(k)}_{12}\ar@{..>}[ll]\ar@<0.2ex>@{..>}@/_2.4pc/[uull]\ar@<-0.2ex>@{..>}@/_2.4pc/[uull]
 }
 \end{gathered}
 \end{equation*}
 \medskip

 \caption{\ The quiver of the endo-algebras $\ca^{(k)}$ derived equivalent to the category of coherent sheaves on a rational elliptic surface with section. }\label{endoo}
 \end{figure}
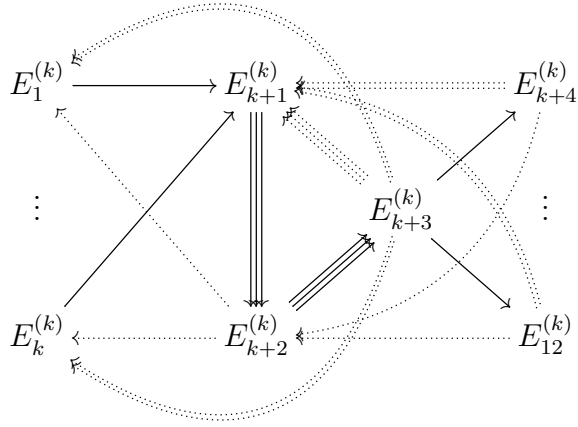

\subsubsection{Flavor symmetry}\label{asqxz1} How we read the flavor group $F$ from
the set $A$?  
If $A=\{1,1,\dots,1\}\equiv \{1^f\}$
and of Kodaira semi-simple type $F$
is the group associated to the corresponding Kodaira fiber, table \ref{ffffhad}. Let us give a graphical rule which extends to the asymptotically-free case.  
$F$
 is the simply-laced Lie group whose Dynkin diagram is a star with 3-arms of lenghts
$(2,p-q,q)$ (counting the valency 3 node in the arm length, so an arm of lenght 1 is no arm at all).
If the length of an arm is negative, the Dynkin graph does not exist, which means that the flavor group is not semi-simple, hence is the Abelian group $U(1)^f$.
If an arm has lenght 0, the central node gets deleted and we remain with the disconnected Dynkin graph $A_1\oplus A_{p-q-1}$ so the flavor group is $SU(2)\times SU(p-q)$, see tables \ref{allowedtypes} for $p\leq 3$, while for the other
 four models with $A=\{1^f\}$
one gets the groups in table \ref{ffffhad}.
The $A=\{1^f\}$ flavor groups are the largest possible  for the given dimension $\Delta$. This correspond to the fact that the $A\neq\{1^f\}$ cases are obtained as singular specializations of the $A=\{1^f\}$ of the same type; then the mass parameters gets specialized to a sublocus, and since the rank of the flavor group is the dimension of mass parameter space, this means that 
$\mathrm{rank}\,F$ gets reduced.

For $A\neq\{1^f\}$ the rule of \cite{unpublished} yields
\begin{equation*}
\begin{tabular}{c|ccccc}\hline\hline
$p$ & 4 & 6 & 7 & 8 & 8\\
$A$ & $\{2\}$ & $\{2,1^2\}$ & $\{2,1^3\}$ &
$\{2,1^4\}$ & $\{2^2\}$\\ 
$(\Delta,F)$ & $(2,Sp(2))$ & $(3,Sp(4)\times U(1))$
& $(4,Sp(6)\times Sp(2))$ & $(6,Sp(10))$ & $(6,Sp(4))$\\\hline\hline
\end{tabular}
\end{equation*}
which (of course) agree with the tables of \cite{Argyres3} under the correspondence \eqref{mmzaqw}.
We now motivate this claims from the categorical viewpoint.

\subsection{Properties of triangular QFT}

We sketch the properties of the triangular QFT in the categorical language.
To make the story shorter, we consider only the superconformal case and focus on rank-1, that is, on the theories with quivers of the form $Q(A,q)$.
We write $\sD_\ca\equiv D^b\mathsf{mod}\,\ca$ for the correspondent bounded derived category. 
$\sD_\ca$ has a Serre functor $S$, i.e.\! an auto-equivalence such that\footnote{\ Throughout the paper $D$ stands for duality of complex vector spaces.}
\be
D\sD_\ca(X,Y)=\sD_\ca(Y,SX),\qquad X,Y\in\sD_A
\ee
Modulo the subtlety
for $A\neq\{1^f\}$, $\sD_A$ is fractional CY of dimension $a/b$, see table \ref{allowedtypes}.
The UV description of a triangular QFT is given by the cluster category $\sC_\ca$, cfr.\! eqn.\eqref{kkaszqwe}.
By construction, $\sC_\ca$ is CY of dimension  2, i.e.\! $S\cong \Sigma^2$ in $\sC_\ca$. $\sC_\ca$ has cluster-tilting objects, e.g.\! $\ca$ seen as a module over itself. 
%Note that the 4d/2d condition $a/b<2$
%is equivalent to $(S^{-1}\Sigma^2)^b\cong \Sigma^{2b-a}$ being degree-increasing.
The 2-CY category $\sC_\ca$
is \emph{symmetric} if it also 
2-periodic i.e.\!
$\Sigma^2\cong\mathrm{Id}$.
The IR description is given by the root category $\sR_\ca$, eqn.\eqref{kkaszqwe}.
%\be
%\sR_\ca=\big(\sD_a/(\Sigma^2)^\Z\big)_\text{tr.hull}.
%\ee 
By construction $\sR_\ca$ is 2-periodic, i.e.\! $\Sigma^2\simeq \mathrm{Id}$.
%
%
%The image of the Serre functor $S$ in the root category $\sR_\ca$ has the physical interpretation of being the 2d quantum monodromy $H$ \cite{CV92},
%while the image of $S$ in the cluster category $\sC_\ca$ is the 4d quantum monodromy $\bM$ \cite{CNV,YQ,Cecotti:2015lab}  associated to the Kontsevitch-Soibelman wall crossing formula \cite{wallcro}.
The Coulomb dimension $\Delta$ is expressed in terms of the orders of the 2d and 4d quantum monodromies
\be\label{dimensiond}
\Delta=\frac{o(H)}{o(\bM)}.
\ee

Since $\Sigma^{2o(\bM)}\cong S^{o(\bM)}\cong \mathrm{Id}$ in $\sC_\ca$, and given that the
 $2o(\bM)$-periodic orbit category
\be
\mathscr{M}_\ca= \left(\sD_\ca/(\Sigma^{2o(\bM)})^\Z\right)_\text{tr.hull}.
\ee
  enjoys a universality property between all $2o(\bM)$-periodic quotient categories of $\sD_\ca$, we get a diagram of projection functors which defines the
  \textit{IR/UV correspondence}
  \be\label{jjaqwe9j}
  \begin{gathered}
  \xymatrix{&& \mathscr{M}_\ca\ar@/_0.7pc/[dll]\ar@/^0.7pc/[drr]\\
  \sR_\ca \ar@{<-->}[rrrr]^{\text{IR/UV correspondence}} &&&& \sC_\ca}
  \end{gathered}\qquad
  \begin{gathered}
 \sR_\ca \cong \mathscr{M}_\ca/\Z_{o(\bM)}\\
 \sC_\ca \cong \mathscr{M}_\ca/\Z_{o(H)}.
  \end{gathered}
  \ee
Pulling back (periodic) Euler characteristics to $\mathscr{M}_\ca$, we get the relations
\be
\langle X,Y\rangle_{\sC_\ca}=o(H)\,\langle X, Y\rangle_{\mathscr{M}_\ca},\qquad
\langle X,Y\rangle_{\sR_\ca}= o(\bM)\,\langle X,Y\rangle_{\mathscr{M}_\ca},
\ee
which, in view of \eqref{dimensiond} yields
\be
\langle X,Y\rangle_{\sC_\ca}=\Delta\,
\langle X,Y\rangle_{\sR_\ca}
\ee
%Comparison with the naive picture, eqn.\eqref{mnzaq}, would give
%\be\label{kkkaqw}
%\kappa_\ca= 2 \Delta,
%\ee
%an identity which holds for all $A=\{1,\dots,1\}$ triangular models but not in general. The point is that one should be careful in identifying the precise currents of which the UV Euler form is computing the 2-point function. We shall dwell on this issue momentarily.
In particular when $o(\bM)=1$,
$\Delta=o(H)\in\bN$ and the cluster category $\sC_\ca$ is \textit{2-periodic}. In this case the IR/UV correspondence reduces to a triangle functor
\be\label{2per}
\sR_\ca\to\sC_\ca\equiv \sR_\ca/(S)^{\Z_\Delta}.
\ee

$o(\bM)$ is defined for all 4d $\cn=2$ SCFT \cite{CNV}, not just for the triangular ones; it is equal to the lcm of the orders of the Coulomb dimensions $\Delta_i$ in $\bQ/\bZ$ \cite{CNV}. Iff all dimensions $\Delta_i$ are integers, $o(\bM)=1$.
In rank-1 we have a single dimension $\Delta$ which may take only 7 values \cite{{Argyres9} ,Caorsi:2018zsq}
\be
\Delta=\tfrac{6}{5},\ \tfrac{4}{3},\ \tfrac{3}{2},\ 2,\ 3,\ 4,\ 6,
\ee
as it follows from the Kodaira types
of numerically CY of rank-1 (table \ref{ffffhad}).
If $\Delta<2$ the SCFT should be an Argyres-Douglas model of type $\mathfrak{g}\in ADE$, whose UV (IR) categories are the cluster (root) Dynkin categories of the same type $\mathfrak{g}$. All other rank-1 SCFT have symmetric 2-CY categories (not necessarily cluster for non triangular models).

\subsubsection{The generic Higgs branch} The first invariant of a $\cn=2$ SCFT is the quaternionic dimension $h$ of the Higgs branch at a generic point along the Coulomb branch.
When $h>0$ one says there is an \emph{enhanced Coulomb branch}\cite{Argyres6}.

The theory of the generic Higgs branch is a subtle and beautiful topic. 
Suffice here to say that for a triangular theory it is controlled by the Bongartz equation \cite{bongartz,Cecotti:2015qha}.
We identify the set $A$ with the vector $\boldsymbol{a}=(a_1,a_2\cdots, a_f)\in \Z^f$.
Then $2h$ is the number of solutions
$\boldsymbol{x}=(x_1,\cdots, x_f)$ to the quadratic equation
\be\label{higgsas}
q^2 \,\boldsymbol{x}\cdot\boldsymbol{x}+(2-q)(\boldsymbol{a}\cdot\boldsymbol{x})^2=q^2
\ee
whose entries are integers.
This yields $h=0$ for all models with
$A=\{1,\dots,1\}$ while for the Argyres-Wittig SCFT half the number of solutions is\footnote{\ The explicit solutions are exhibit in table \ref{solbon}.}
\be\label{rreqw}
\begin{tabular}{c|ccccc}\hline\hline
($A$, $q$) & ($\{2\}$, 2) & ($\{2,1^2\}$, 3) &
 ($\{2,1^3\}$,3) & ($\{2,1^4\}$,3) & ($\{2^2\}$,3)\\
 $h$ & 1 & 2 & 3 & 5 & 2\\\hline\hline
\end{tabular}
\ee
The flavor group of an Argyres-Wittig model has the form
$Sp(2h)\times F^\prime$ with $Sp(2h)$ acting in the natural way on the generic Higgs branch.

%However here there is a subtlety (cfr.\! \S.4.8 of \cite{Alim:2011kw}). The theory associated to the Markoff quiver $Q(\{2\},2)$ is not $SU(2)$ $\cn=2^*$ but its 
%Gaiotto class $\cs$-theory i.e.\!
%$SU(2)$ SYM coupled to an adjoint hypermultiplet plus a free hypermultiplet.
%This is due to the fact that the Markoff quiver and its superpotential have a $\Z_2$ automorphism acting trivially on the nodes
%(so in particular on the conserved charges).
%The module associated to the neutral hyperpmultiplet in the adjoint representation is not invariant under this $\Z_2$, so we get a second hypermultiplet with zero electric and magnetic charge and carrying the same flavor charge as the adjoint hyper. Hence the visible flavor symmetry $Sp(2)$ is the \emph{diagonal}
%subgroup $Sp(2)_{\cn=2^*}\times Sp(2)_\text{free}$. The Euler characteristic computes $\kappa_F$ for the combined system, so to get the one for the interacting sector one has to subtract the free contribution. Hence $\kappa_F=2\Delta-1=3$.
%
%
%
%Then one shows that correct value for the flavor central charge $\kappa_F$ for a rank-1 triangular theory is $\kappa_F=2\Delta-h$.
% for a rank-1 triangular theory. 
%

\subsection{$K_0(\sC_\ca)$, $K_0(\sR_\ca)$ for triangular QFT} \label{jjaszq}
To substantiate our physical picture we have to compute the Grothendieck group of the Amiot cluster category $\sC_\ca$ with its appropriate Euler form and check that they yield the expected flavor symmetry. For the cluster category defined by the quiver $Q(A,q)$, the group $K_0(\sC_\ca)$ is the Abelian group generated by the classes of the projective $\ca$-modules $[P_\alpha]$,
$[P_\omega]$, and $[P_i]$ ($i=1,\dots,f$)
subjected to the two relations
\be
\gcd(a_i,q)\,[P_\alpha]=\gcd(a_i,q)\,[P_\omega],\qquad
q[P_\alpha]+\sum_i a_i [P_i]=0.
\ee
Assume $A\neq\{2\},\{2,2\}$: then $\gcd(a_i,q)=1$, and  $K_0(\sC_\ca)\simeq \Z^f$ is \emph{freely} generated by $[P_\alpha],
[P_1],\cdots, [P_{f-1}]$.

\subsubsection{The flavor weight lattice}
It is convenient to write the above lattice in
the form ($A\neq\{2\},\{2,2\}$)
\be\label{lattgam}
\Gamma_\ca=\left\{\,\sum_{i=1}^f w_i[P_i]\ \bigg|\ w_i\in\frac{a_i}{q}\Z\ \ \text{and}\ \  q\,w_i/a_i\equiv q\,w_j/a_j\bmod q\,\right\}
\ee 
For $\Delta\geq 2$ the Euler quadratic form on $\sC_\ca$ is
\be\label{quaf}
\Delta\cdot \sum_i w_i^2
\ee
(see also \textsc{appendix \ref{catfal}}).
A part for the overall factor $\Delta$, the Grothendieck group \eqref{lattgam} equipped with the Euler quadratic form \eqref{quaf} is the weight lattice of the flavor groups listed in \S.\,\ref{asqxz1}
with their canonical Weyl-invariant inner product.

The overall factor in Euler's form, eqn.\eqref{quaf}, has a simple interpretation. For the SCFT $A=\{1^f\}$, without subtleties, $\Delta$ is one-half the central charge
of the superconformal flavor current algebra, $\Delta=\kappa_F/2$; in other words $\Delta$ is the overall
normalization of the flavor-current
two point function, so that for $a_i=1$
\eqref{quaf} is just the physical  normalization of the flavor weigth quadratic form. 
Inverting the argument, from the Euler form of the cluster category we may read not just the flavor symmetry group $F$, but also its conformal central charge $\kappa_F$ as well as interesting selection rules on the representations of $F$ which may appear (see eqn.\eqref{grorel} below for a typical example). 

For the Argyres-Wittig models, $A=\{2,1^{f-1}\}$,
one has a similar formula, whose meaning is less clear. We write $F=Sp(2h)\times F^\prime$ (cfr.\! eqn.\eqref{rreqw}) and consider the two central charges
$\kappa_{Sp}$ and $\kappa_{F^\prime}$. One has \cite{Argyres6}
\be
2\,\Delta=\kappa_{Sp}+h,\qquad 2\,\Delta=\kappa_{F^\prime}.
\ee

\subsubsection{The flavor root lattice}\label{flaroot}
Let us consider the Grotendieck group of the IR category, $K_0(\sR_\ca)$. In the triangular case, this group contains also the electro-magnetic charges, so to compare with $K_0(\sC_\ca)$ we consider the  sublattice of purely flavor charges namely 
the sublattice $\Lambda_\ca$ of the $S$-invariant classes $[SX]=[X]$,  $X\in\sR_\ca$. 
%Passing from $K_0(\sR_\ca)$ to its sub-lattice $\Lambda_\ca$ corresponds to the operation we described
%physically at the end of section 2 of taking differences of charges of IR states with the same non-flavor charges. 

\subparagraph{The simply-laced case.} We consider first the case $A=\{1^f\}$ which yields a simply-laced $F$ of rank $f$.
In this case $\Lambda_\ca\subset K_0(\sR_\ca)$ corresponds to classes of the form
\be
\sum_i (x_i+y) [S_i] +y\big([S_\alpha]+[S_\omega]\big)
\quad x_i,y\in\Z, \qquad q\,y+\sum_ix_i=0
\ee
where $S_\bullet$ are the simple modules of $\ca$.
We write $\Lambda_\ca$ in the form
\be\label{za12l}
\Lambda_\ca= \left\{\boldsymbol{x}\equiv (x_1,\cdots, x_f)\in\Z^f\;\bigg|\; \sum_{i=1}^f x_i=0\bmod q\right\}
\ee
with a natural pairing given by the Euler form of $\mathsf{mod}\,\ca$ (symmetric when restricted to $\Lambda_\ca$)
\be
\Big\langle \boldsymbol{x}, \boldsymbol{y}\Big\rangle_{\mathsf{mod}\,\ca}= \sum_i x_i\, y_i-\frac{q-2}{q^2}\,\sum_i x_i \cdot\sum_j y_j
\ee
This symmetric form is \emph{integral} in lattices of the form \eqref{za12l} for all $(f,q)$. For $q=2$ or $q$ odd it is also \emph{even.} It is \textit{positive-definite} for 
\be
q^2> f(q-2).
\ee
Suppose $f\geq q$. The elements
\be
\begin{aligned}
\alpha_i&=(0,\cdots,0,\overset{i\text{-th}}{1},-1,0,\cdots,0), &\phantom{mm}&\text{for }i=1,\dots,f-1\\ \alpha_f&=(1,\cdots,1,0,\cdots, 0)
&&\text{with $q$ 1's}
\end{aligned}
\ee
are of square-length 2 and span $\Lambda_\ca$. Hence, under the condition that $q$ is either 2 or odd and
$(q-2)f< q^2\leq f^2$, by Witt's theorem \cite{witt}
$\Lambda_\ca$ is the root lattice of a simply-laced Lie algebra $\mathfrak{f}$ and
$\langle-,-\rangle_{\textsf{mod}\,\ca}$ is its canonical quadratic form.
The conditions are satisfied by all
the allowed types $(p,\bar q)\equiv (f,q)$
listed in table \ref{allowedtypes}; from the
simply-laced Lie algebra $\mathfrak{f}$ we read the flavor Lie group $F$ of the SCFT described by the quiver $Q(\{1^f\},q)$.
With respect to the pairing
$\langle-,-\rangle_{\textsf{mod}\,\ca}$ the two lattices $\Gamma_\ca$ and $\Lambda_\ca$ are each other dual.
%We stress that to recover the physical normalization of the currents one should use the IR/UV correspondence and keep the normalization given by the canonical form on $K_0(\sC_a)$.
%For $\Delta\in\bN$ the correspondence is given by a $\Delta$-to-1 covering functor $\sR_\ca\to\sC_\ca$ \eqref{2per}, and we get an extra factor $\Delta$ from the degree of the cover, which then yields $\kappa_F=2\Delta$, as already remarked. 

\subparagraph{The non-simply-laced case.} We assume $A=\{2,1^{f-1}\}$ with $f>1$. This is the case of the three Argyres-Wittig models $(3,Sp(4)\times U(1))$, $(4,Sp(6)\times SU(2))$ and $(6,Sp(10))$.
 Now $\Lambda_\ca\subset K_0(\sR_\ca)$ is given by the classes of the form
\be
(x_1+2y)[S_1]+\sum_{i\geq2} \big(x_i+y\big) [S_i] +y\big([S_\alpha]+[S_\omega]\big),
\quad x_i,y\in\Z, \quad q\,y+2x_1+\sum_{i\geq2}x_i=0
\ee
i.e.\! the lattice $\Lambda_\ca=\{\boldsymbol{x}=(x_1,\cdots, x_f)\in \Z^f\;|\;(2x_1+\sum_{i\geq2}x_i=0\bmod q\}$ with pairing
\be\label{jaswe}
\Big\langle \boldsymbol{x}, \boldsymbol{y}\Big\rangle_{\mathsf{mod}\,\ca}= \sum_i x_i\, y_i-\frac{q-2}{q^2}\,\Big(2x_1+\sum_{i\geq2} x_i\Big)\Big(2y_1+\sum_{j\geq2} y_j\Big).
\ee

\begin{table}
\begin{footnotesize}
\begin{equation*}
\begin{array}{c|cccc}\hline\hline
(\text{length})^2=1&\pm(e_1+e_{i_1}), &&&
\pm(e_1+e_2+e_3+e_4+e_5)\\
(\text{length})^2=2&\pm(e_{i_1}-e_{i_2}),
& \pm(e_{i_1}+e_{i_2}+e_{i_3}), & \pm(2e_1+e_{i_1}+e_{i_2}),&
\pm(2e_1+e_2+e_3+e_4+e_5)\\\hline\hline
\end{array}
\end{equation*}
\end{footnotesize}
\caption{\label{solbon} Elements of square-length 1 and 2 in the lattice $\Lambda_\ca$ for $A=\{2,1^{f-1}\}$. The $e_i$'s are the elements of the canonical basis of $\Z^f$. The indices $i_k=2,\dots, f$  are \emph{all distinct.} The last column is present only for the Argyres-Wittig model with $\Delta=6$. The first row yields also the list of integral solutions to eqn.\eqref{higgsas}.}
\end{table}

This quadratic form is still integral, symmetric, and positive-definite for  the allowed types $(p,q)$. However it is not \emph{even}: now $\Lambda_\ca$ contains vectors of square-length 1: see table \ref{solbon}. There are exactly $2h$ of them in correspondence with the hypermultiplets spanning the generic Higgs branch.
The elements of square-length 2 in $\Lambda_\ca$ are the short roots of the flavor group, while
twice the vectors associated to
the generic Higgs branch are the long roots of $Sp(2h)$.
Thus for $(A,q)=(\{2,1^2\},3)$ we have $4$ short and $4$ long roots,
so $F=Sp(4)\times\text{(Abelian)}$; for 
$(A,q)=(\{2,1^3\},3)$ we have $12+2$ short and $6$ long roots,
so $F=Sp(6)\times SU(2)$;
and for $(A,q)=(\{2,1^4\},3)$ we have $40$ short and $10$ long roots,
so $F=Sp(10)$.

\subparagraph{The last two cases.}
We return to the 2 cases we have
omitted
\be
(A,q)=(\{2\},2),\qquad (A,q)=(\{2,2\},3).
\ee   
The first model is $SU(2)$ $\cn=2^*$ and its physics is well-understood.
The Grothedieck group of $\cn=2^*$ is generated by $[S_\alpha]$, $[S_\omega]$ and $[S_1]$ subjected to the relations
\be
2[S_\alpha]=2[S_\omega]=2[S_1],
\ee
so that $K_0(\sC_{\cn=2^*})$ is isomorphic to $\Z\otimes\Z_2\otimes\Z_2$ generated by
\be\label{llllc23}
[S_\alpha]+[S_\omega]+[S_1], \qquad [S_\alpha]-[S_1],\qquad [S_\omega]-[S_1].
\ee
This is the correct 't Hooft group in the UV, see \textsc{appendix \ref{mwtorsion}}.
Indeed for a gauge theory with matter the (maximal) UV group is \cite{Caorsi:2017bnp}
\be
\Gamma_\text{flav}\oplus \pi_1(G_\text{eff})\oplus \pi_1(G_\text{eff})^\vee,
\ee
where $G_\text{eff}$ is the quotient of the gauge group $G$ which acts effectively on the local fields. In the $\cn=2^*$ model $G_\text{eff}=SU(2)/\Z_2$ since the local fields are in the adjoint representation. Then $\pi_1(G_\text{eff})=\Z_2$ and we get full agreement between $K_0(\sC_{\cn=2^*})$ and the expected UV group.
In the IR $\Lambda_{\cn=2^*}\cong\Z$
generated by $[S_\alpha]+[S_1]+[S_\omega]$.\medskip

For the quiver $(A,q)=(\{2^2\},3)$
the lattice and pairing is
$\Lambda_\ca$ is 
$$
\Lambda_\ca=\Big\{(x_1+2y)[P_1]+(x_2+2y)[P_2]+y([P_\alpha]+[P_\omega])\;\Big|\; x_i,y\in\Z,\ \ 3y+2x_1+2x_2=0\Big\}
$$
with pairing
\be
\big\langle(x_1,x_2),(y_1,y_2)\big\rangle =
x_1y_1+x_2y_2-\frac{4}{9}(x_1+x_2)(y_1+y_2).
\ee
There are 4 vectors of lenght 1, 
$\pm(2,1)$ and $\pm(1,2)$
and 4 length 2 $\pm(1,-1)$ and $\pm(3,3)$.
$Sp(4)$ as expected.
In $K_0(\sC)$  we have the relations
\be
3[S_\alpha]=2[S_1]+2[S_2], \quad 3[S_\omega]=2[S_1]+2[S_2], \quad 2[S_\alpha]=2[S_\omega],
\ee
so $K_0(\sC)$ is free of rank 2
\be
K_0(\sC_{(\{2^2\},3)})\cong \Z\big([S_\alpha]- [S_1]\big)\oplus \Z\big([S_\alpha]-[S_2]\big).
\ee
In this basis the natural pairing reads
\be
\Big\langle \big([S_\alpha]- [S_a]\big),
\big([S_\alpha]- [S_b]\big)\Big\rangle_\sC= \Delta\, \delta_{ab},\quad \Delta=6,\ a,b,=1,2,
\ee
i.e.\! the $Sp(4)$ weight lattice rescaled by $\Delta$, as expected.
 
\section{The 15 missing SCFT}

In total, the 4d/2d correspondence of \S.\ref{4d2d} produces
12 rank-1 SCFTs and 4 rank-1 $a\!f$ theories. The list of $a\!f$ theories is complete, but in table 1 of \cite{Argyres3}
Argyres \emph{et al.}\! list 27 rank-1 SCFTs (the table has 28 entries, but $SU(2)$ $\cn=2^*$ is listed twice, since that theory may be thought of in two different ways\footnote{\ Related issues are discussed in \textsc{appendix \ref{mwtorsion}}.}).
Thus the $Q(A,q)$ family of quivers describes all rank-1 QFT except 15 SCFT.
%The missing SCFT have $\Delta\geq 3$, i.e.\! they are not complete, in agreement with classification \cite{Cecotti:2011rv}.

To substantiate our claim that classifying the appropriate class of UV categories is equivalent to classifying all $\cn=2$ QFTs,  we have to exhibit, in addition to the family of cluster categories $\sC_\ca$ associated to the $Q(A,q)$ quivers,
other 15 2-CY categories with the  properties required to describe rank-1 SCFTs. Where do we find them?

The clue comes from the classification of  base-changes between rational elliptic surfaces
satisfying the geometric requirements called ``UV and SW completeness'' obtained in \cite{Caorsi:2018ahl} using the tables of \cite{cover}.
All $\Z_k$ gaugings\footnote{\ In ref.\!\cite{Argyres7} it is explained physically why the discrete groups which may be gauged while preserving $\cn=2$ supersymmetry have the form $\Z_k$. From the `arithmetic' perspective of \cite{Caorsi:2018ahl} this follows from the fact that the only consistent base-changes are cyclic, see \cite{cover}. } of a rank-1 SCFT  which preserve $\cn=2$ \textsc{susy} correspond to base-changes of elliptic surfaces (seen as elliptic curves over the field $\C(z)$ of rational functions) but the inverse statement is \emph{false.} There are geometrically sensible base-changes, even symplectic ones (see \S.\ref{reviewgauging}), which are \emph{not} discrete gaugings in the physical sense.
Mathematically these base-changes share most properties of actual gaugings: we dub them \emph{false-gaugings}.
\medskip

We review  base-changes in \S.\,\ref{reviewgauging} below. The result of the analysis is that \emph{all} 15 missing SCFT are either discrete gaugings or false-gaugings: 10 gaugings and 5 false-gaugings. They also exhaust the list of admissible base-changes.
%Armed with this observation, to reproduce table 1 of 
%\cite{Argyres3} from the RT perspective requires two steps:
%\begin{itemize}
%\item[1)] to construct the 2-CY categories with cluster-tilting associated to each discrete gauging and pseudo-gauging;
%\item[2)] to show that there are no other ``UV complete rank-1'' cluster categories (with tilting) besides the $Q(A,q)$ ones and those corresponding to the 15 gaugings/pseudo-gaugings.
%\end{itemize}
%The second step is however tautological in our set-up since we defined the ``UV-completeness condition'' to be precisely the closure of the criterion producing the set $Q(A,q)$
%under the operations of gauging/pseudo-gauging.
%

\subsection{Review of base-change/discrete gauging}\label{reviewgauging}

%In  \cite{Caorsi:2018ahl} it was shown that discrete-gaugings of a rank-1 $\cn=2$ model
%correspond geometrically to base-changes for the corresponding (rational) elliptic surfaces. The opposite statement is \underline{not} true. 

A rational elliptic surface $\ce$ is, in particular, a holomorphic fibration $\ce\to \bP^1$ whose generic fiber is a smooth elliptic curve. The exceptional fibers were classified by Kodaira \cite{kod1,kod2}: they are in one-to-one correspondence with the quasi-unipotent conjugacy classes of $SL(2,\Z)$. The configurations of exceptional Kodaira fibers allowed in a rational elliptic surface (with section) are listed in \cite{bookMW,per,mira}. To fully fix $\ce$, in addition to the fiber configuration, we have to specify a rational function $\mathscr{J}\!(z)$ (the functional invariant \cite{kod1,kod2,miranda}) which satisfies certain properties in relation to the fiber configuration, see \cite{Caorsi:2018ahl} and the references therein for details.
   
The (total) special geometry of a rank-1
$\cn=2$ theory is given by $\ce\setminus \cf^\vee$, where $\cf^\vee$ is the fiber over $\infty\in\bP^1$. UV completeness requires the curve $\cf^\vee$ to be
of \emph{un-stable} type. SW completeness requires, in addition,
that there is no fiber over $\bP^1\setminus\infty$ of types $II$, $III$ and $IV$. We write a fiber configuration
as $\{\cf^\vee;F_{z_1},\cdots, F_{z_s}\}$ where the first fiber is always the one at infinity. As already mentioned, the fiber $\cf^\vee$ encodes the Coulomb dimension $\Delta$
and for $\Delta=2$ also the $\beta$-function coefficient $b$:
\be
\begin{tabular}{c|ccccccc}\hline\hline
$\cf^\vee$ & $II^*$ & $III^*$ & $IV^*$ & $I_b^*$ & $IV$ & $III$ & $II$\\
$\Delta$ & $\tfrac{6}{5}$ & $\tfrac{4}{3}$ &
$\tfrac{3}{2}$ & $2$ & $3$ & $4$ & $6$\\\hline\hline
\end{tabular}\quad\qquad b=\begin{cases} b &\text{for }\cf^\vee=I^*_b\\
0 &\text{otherwise.}
\end{cases}
\ee

A base-change is given by a commutative diagram of the form
\be\label{diagramcov}
\begin{gathered}
\xymatrix{\ce^\prime\ar[d]\ar[r]^{\phi_\circ} & \ce\ar[d]\\
\bP^1\ar[r]^\phi &\bP^1}
\end{gathered}
\ee
with $\phi$ a rational map of degree $n>1$, so that the elliptic fibration $\ce^\prime$ is the pull-back
$\phi^*(\ce)$. A base-change is \emph{symplectic}\footnote{\ For an example of a non-symplectic base change see \textsc{appendix \ref{exnonsimpletic}.}}
 iff the (log-)symplectic structure $\Omega^\prime$ of the pair $(\ce^\prime,\cf^{\vee\,\prime})$ is the pull-back of the one for $(\ce,\cf^\vee)$, i.e.\! 
\be
\Omega^\prime=\phi_\circ^{\,\ast}\Omega.
\ee 
This amounts to saying that the SW differential on the covering special geometry, $\ce^\prime\setminus \cf^{\vee\,\prime}$, is the pull-back of the one on the base $\ce\setminus \cf^\vee$. 
 
One shows \cite{cover} that
$\phi$ must be a cyclic map of order $n$.
UV completeness requires $\infty$ to be a branching point for $\phi$, and we write $\phi\colon z\mapsto z^n$. The coordinate $z$ on $\bP^1$ is then identified with the Coulomb branch coordinate $u$; this yields $u= (u^\prime)^n$, so that -- if the covering is symplectic -- the Coulomb branch dimensions are related by
\be\label{dfffr}
\Delta=n\cdot \Delta^\prime.
\ee
Since the distinction SCFT/$a\!f$ is preserved by the covering, and given that for all $a\!f$ theories $\Delta=2$, we see from \eqref{dfffr} that there is no non-trivial covering for $a\!f$
models, and we restrict to SCFT in the rest of the discussion. In terms of fiber configurations there are 19 allowed coverings between UV/SW complete rational elliptic surfaces consistent with eqn.\eqref{dfffr}. 4 of them are physically interpreted as special limits of another one, so we remain with the 15 coverings listed in table \ref{missing}.
They precisely match the 15 `missing' rank-1 SCFT.

\begin{table}
\begin{center}
\begin{minipage}{400pt}
\begin{tabular}{ccccccc}\hline\hline
$\#$ & $\Delta$ & $F$ & fibers $\ce$ & fibers $\ce^\prime$ & Galois & gauging?\\\hline
4 & 6 & $F_4$ & $\{II;I_0^*,I_1^4\}$ & $\{IV;I_1^8\}$ & $\Z_2$ & $\checkmark$\\
5 & 6 & $Sp(6)$ & $\{II;I_1^*,I_1^3\}$ & $\{IV; I_2, I_1^6\}$ & $\Z_2$\\
6 & 6 & $SU(2)$ & $\{II;I_1^*,I_3\}$ & $\{IV; I_2,I_3^2\}$ & $\Z_2$\\
7 & 6 & $Sp(4)$ & $\{II;I_2^*,I_1^2\}$ & $\{IV; I_4, I_1^4\}$& $\Z_2$ & $\checkmark$\\
8 & 6 & $SU(2)$ & $\{II;I_3^*,I_1\}$ & $\{IV; I_6, I_1^2\}$& $\Z_2$\\
9 & 6 & $SU(2)$ & $\{II;IV^*,I_2\}$ & $\{I_0^*;I_2^3\}$& $\Z_3$ & $\checkmark$\\
10 & 6 & $G_2$ & $\{II;IV^*,I_1^2\}$ & $\{I _0^*; I_1^6\}$& $\Z_3$ & $\checkmark$\\
11 & 6 & $SU(2)$ & $\{II;III^*,I_1\}$ & $\{IV^*; I_1^4\}$& $\Z_4$ & $\checkmark$\\\hline
14 & 4 & $Spin(7)$ & $\{III;I_0^*,I_1^3\}$ & $\{I_0^*, I_1^6\}$& $\Z_2$ & $\checkmark$\\
15 & 4 & $SU(2)\times SU(2)$\footnote{\ The discrete gauged realization has the symmetry $SU(2)\times U(1)$ \cite{Argyres3,Argyres5}.} & $\{III;I_1^*,I_1^2\}$ & $\{I_0^*; I_2, I_1^4\}$& $\Z_2$\\
16 & 4 & $SU(2)$ & $\{III;I_1^*,I_2\}$ &$\{I_ 0^*; I_2^3\}$& $\Z_2$ & $\checkmark$\footnote{\ The two $\ce^\prime$ in entries 16 and 17 both correspond to $SU(2)$ $\cn=2^*$.}\\
17 & 4 & $SU(2)$ & $\{III;I_2^*,I_1\}$ & $\{I_0^*; I_4, I_1^2\}$& $\Z_2$ &$\checkmark$\\
18 & 4 & $SU(2)$ & $\{III;IV^*,I_1\}$ &$\{III^*;I_1^3\}$& $\Z_3$ & $\checkmark$\\\hline
21 & 3 & $SU(3)$ & $\{IV;I^*_0,I_1^2\}$ & $\{IV^*; I_1^4\}$& $\Z_2$ & $\checkmark$\\
22 & 3 & $U(1)$ & $\{IV;I^*_1,I_1\}$ & $\{IV^*; I_2, I_1^2\}$& $\Z_2$\\\hline\hline
\end{tabular}
\end{minipage}
\caption{\label{missing}The 15 rank 1 SCFT which are not described by a $Q(A;q)$ quiver. The first column is the number or the corresponding entry in table 1 of ref.\!\cite{Argyres3}: we shall often identificate a SCFT by this entry number.
The second column gives the Coulomb branch dimension and the third one the flavor group. The fourth column is the configuration of Kodaira exceptional fibers in the corresponding rational elliptic surface $\ce$, the first entry being the fiber at infinity. The fourth column yields the fiber configuration for the rational elliptic surface $\ce^\prime$ covering $\ce$ under change of the base field (following ref.\cite{cover}). The cover
$\ce^\prime\to \ce$ is branched only over the first two exceptional fibers in the configuration. The last column follows from comparison with table 1 of ref.\!\cite{Argyres7} .}
\end{center}
\end{table}

However not all interesting  coverings $\ce^\prime\to\ce$ correspond to discrete gaugings in the sense of refs.\!\!\cite{Argyres7,discretegaugings1,discretegaugings2}. The true gaugings are distinguish by a check-mark $\checkmark$ in the last column of the table.  
By inspection, we note that the unchecked base-changes are precisely the ones whose covering surface $\ce^\prime$ has a fiber configuration which does not satisfy the ``Dirac quantization of charge'' of ref.\!\cite{Argyres3,Argyres4}, equivalently are not described by a 
$Q(A,q)$ quiver. The categorical description should (in particular) clarify the subtle  distinction between \emph{true} and \emph{false} gaugings.
\smallskip

Going through the table of base-changes, we see that for all the 15 `missing' SCFT
\be\label{kkkas1}
\mathrm{rank}\,F=\#\big(\text{exceptional fibers of $\ce$ over $\bP^1\setminus\{0,\infty\}$}\big).
\ee 
This formula is not valid for the 16 QFT described by a $Q(A,q)$ quiver. In that case
\be\label{kkkas2}
\mathrm{rank}\,Q(A,q)=2+\mathrm{rank}\,F=\#\big(\text{exceptional fibers of $\ce$ over $\bP^1\setminus\{\infty\}$}\big),
\ee
as it is clear from the (quiver) $\leftrightarrow$ (fiber configuration) correspondence \eqref{mmzaqw}.

The difference between \eqref{kkkas1} and 
\eqref{kkkas2} reflects the different structure of the Grothendieck group of the UV categories for the two classes of $\cn=2$ theories. This is hardly a surprise:
we already predicted these same  expressions on general grounds, see eqn.\eqref{pppqwer}.  

\section{Warm-up: discrete gaugings of $SU(2)$ with $N_f=4$}\label{za10k}

To orient ourselves it is convenient to start by working out in detail some explicit example in a context where both the physics and the mathematics are well understood.
The perfect set-up
  are the discrete gauging of $SU(2)$ with $N_f=4$.
Three SCFT in the list of \cite{Argyres3} belong to this class:
\be\label{sun4gau}
\begin{tabular}{ll}
$\bullet$\ the ungauged $(2,Spin(8))$ theory &(entry 23)\\
$\bullet$\ a $\Z_2$ gauging produces 
the SCFT $(4,Spin(7))$ &(entry 14)\\
$\bullet$\  a $\Z_3$ gauging produces 
the SCFT $(6,G_2)$ &(entry 10).
\end{tabular}
\ee
Our fist goal is to construct explicitly the three 2-CY categories describing these SCFT in UV. 
We begin by reviewing the categories associated to the SCFT $(2,Spin(8))$.

\subsection{IR category for $SU(2)$ with $N_f=4$}

The BPS objects of 4d $\cn=2$ SQCD with gauge group $SU(2)$ are described by
categories associated to $\mathbb{P}^1$ and its generalizations (the weighted projective lines).
The relation between $SU(2)$ SQCD and weighted projective lines is mirror symmetry. To illustrate the point, we work out the case of $SU(2)$ with $N_f=0,1,2$
referring to the literature \cite{Cecotti:2012va,Cecotti:2015hca,Caorsi:2017bnp} for the general case.  The canonical SW geometry
of $SU(2)$ with $N_f=0,1,2$ corresponds,\footnote{\ We fine-tuned the masses to a convenient value, for simplicity.} respectively, to the curves \cite{Tachikawa:2011yr}
\be\label{ppppq12234}
W_0\equiv p^2-e^x-e^{-x}=0,\quad
W_1\equiv p^2-e^{2x}-e^{-x}=0,\quad
W_2\equiv p^2-e^{2x}-e^{-2x}=0,
\ee
($x\sim x+2\pi i$) with SW differential $\lambda =p\,dx$. The 4d/2d correspondence \cite{CNV} associates to each 4d QFT the 2d (2,2) Landau-Ginzburg model with superpotential $W_{N_f}$ which is easily seen to be the mirror of the $\sigma$-model with target the weighted projective lines 
$\bP(1,1)$, $\bP(2,1)$ and $\bP(2,2)$, respectively. The relation continues to hold for $N_f>2$: the general case involves a weighted projective line $\bX$ of
type $(2,2,\cdots,2)$ with $N_f$ 2's.
The most convenient language to describe the triangle categories associated to $\cn=2$ $SU(2)$ SYM coupled to fundamental matter
is that of coherent sheaves over weighted projective lines\footnote{\ For details and complete references see \textsc{appendix \ref{Acoh}}.}. This formalism has been reviewed in \cite{Cecotti:2012va} and \cite{Cecotti:2015hca}; we shall adhere to the conventions of this last paper.

$SU(2)$ with $N_f=4$ is associated to 
a weighted projective line $\bX$ of tubular type $(2,2,2,2)$. $\mathsf{coh}\,\bX$, the 
Abelian category of coherent sheaves on $\bX$, is hereditary with tilting objects (see \textsc{appendix \ref{Acoh}}).
We write $\tau$ for the auto-equivalence of 
$\mathsf{coh}\,\bX$ given by the tensor product with the dualizing sheaf $\co(\vec \omega)$. Serre duality then reads
\be\label{AR}
\mathrm{Ext}^1(X,Y)\cong D \mathrm{Hom}(Y,\tau X),\qquad X,Y\in \mathsf{coh}\,\bX.
\ee
 For $\bX$ of type $(2,2,2,2)$ one has $\tau^2=\mathrm{Id}$. The fact that $\tau$ has finite order ($\equiv$ the dualizing sheaf is torsion) is equivalent to the vanishing of the Yang-Mills $\beta$-function \cite{Cecotti:2012va}.

The bounded derived category $\sD_\bX\equiv D^b\mathsf{coh}\,\bX$ is just the repetitive category of $\mathsf{coh}\,\bX$ and \eqref{AR} extends to $\sD_\bX$ in the form
\be
D\,\mathrm{Hom}(X,Y)\cong\mathrm{Hom}(Y,\tau \Sigma X)\qquad X,Y\in \sD_\bX,
\ee  
i.e.\! the Serre functor is $S\equiv\tau\Sigma$.
On $\sD_\bX$ we have the (non-unique) telescopic auto-equivalences $T$ and $L$, satisfying the $\cb_3$ braiding relation
(see \textsc{appendix \ref{Acoh}})
\be\label{braid}
TLT=LTL,
\ee
which physically generate the $SL(2,\Z)$ duality group of $SU(2)$ with $N_f=4$ \cite{Cecotti:2015hca}.
 For the telescopic functors we adhere to the conventions of \cite{Cecotti:2015hca}:
in particular $T$ is realized as the twist by the line bundle $\co(\vec x_3)$.
 On the electro-magnetic charges $T$ and $L$ act as the $SL(2,\Z)$ matrices \cite{Cecotti:2015hca}
  \be\label{modtras}
 T \leadsto \begin{bmatrix} 1 & 1\\
 0 & 1\end{bmatrix},\qquad 
 L \leadsto \begin{bmatrix} 1 & 0\\
 -1 & 1\end{bmatrix}.
 \ee 

\subsection{Gaugeable auto-equivalences}

Consider an elliptic curve $E$. Depending on the value of its modulus $\boldsymbol{\tau}$, the group $\mathrm{Aut}(E)$ may be $\Z_2$ (for generic $\boldsymbol{\tau}$), $\Z_4$ (for
$\boldsymbol{\tau}^2=-1$) or $\Z_6$ (for $\boldsymbol{\tau}^3=-1$). 
These groups are generated, respectively, by $-1$,
by $\zeta$ with $\zeta^2=-1$, and by $\xi$ with $\xi^3=-1$.
Since the conformal manifold of $SU(2)$ with $N_f=4$ is isomorphic to the moduli space of elliptic curves, the associated triangle category $\sD_\bX$
must have auto-equivalences corresponding to $-1, \zeta,\xi\in \mathrm{Aut}(E)$. We write them as $M_2$, $M_4$ and $M_6$, respectively; they should satisfy the relations
\be\label{lllasqw}
(M_4)^2=M_2,\qquad (M_6)^3=M_2.
\ee 
This physical argument predicts that no autoequivalence $M$ with $M^n=M_2$ should exist for $0<n\neq1, 2,3$. Of course, the prediction is mathematically correct. $M_2$ corresponds to $-1\in\mathrm{Aut}(E)$, that is, it generates the Weyl group of the physical $SU(2)$ gauge group: $M_2$ acts as $-1$ on the electro-magnetic charges and as $+1$ on the flavor ones. $\zeta$ and $\xi$ then correspond to the additional gauging of a discrete $\Z_2$ (resp.\! $\Z_3$) symmetry. The fact that no other $M^n=M_2$ autoequivalence exists, rules out discrete gauging by other discrete groups.

$M_2$ is identified with the auto-equivalence $\tau^{-1}\Sigma$ \cite{Cecotti:2015hca}. The simplest solutions to \eqref{lllasqw} are\footnote{\ Solutions conjugated in $\mathrm{Aut}(\sD_\bX)$ are equivalent. $M_2$ is central in $\mathrm{Aut}(\sD_\bX)$.}
\be\label{jjjas}
\begin{aligned}
M_2&=\tau^{-1}\Sigma,\\
M_4&= TLT,\\
M_6&=TL.
\end{aligned}
\ee
They satisfy
\eqref{lllasqw} since for all tubular weighted projective
lines of type
$\neq(3,3,3)$ one has
\be\label{relrel}
(TLT)^2=(TL)^3=\tau^{-3}\Sigma,
\ee   
the first equality being a trivial consequence of the braid relation \eqref{braid} and the second one well known (cfr.\!\!\! \textsc{appendix}). For $\bX$ of type $(2,2,2,2)$ $\tau^2=\mathrm{Id}$, and eqn.\eqref{relrel} becomes equivalent to eqn.\eqref{lllasqw}.

%
%\textcolor{blue}{\begin{rem}[\textit{$M_d$ and the Coulomb branch dimension}] By definition, the Coulomb branch dimensions may be read from the quantum monodromy $\bM$ which represents the action of a $U(1)_R$ rotation in $2\pi$.
%In categorical terms, $\bM$ corresponds to the double shift $\Sigma^2$. The $\ell$-fractional monodromy $\bL$ is an operator such that $\bL^\ell=\bM$ [cite]. We have
%$(M_d)^d=\Sigma^2$, so $M_d$ is a $d$-fractional monodromy. If $\bL$ acts on the BPS line operators as the identity (while it cannot be written as a power of an operator acting as 1) it follows
%that the dimensions are precisely $d$. The last statement follows from the quantum monodromy written in terms of the usual finite BPS chamber.
%\end{rem}}
%
%\medskip
%
%\textcolor{green}{We recall that a $\C$-linear triangle category with split idempotents, $\mathscr{C}$, is said to be $n$-Calabi-Yau if we have a bifunctorial isomorphism
%\be
%\mathscr{C}(X,Y)\cong D\mathscr{C}(Y,X[n]),\quad X,Y\in\mathscr{C}.
%\ee
%Note that this implies
%\be\label{symsym}
%\sC(X,Y[k])\cong D\sC(Y,X[n-k]),\quad 0\leq k\leq n.
%\ee}
%

\subsection{UV categories for the gauged models}

The orbit categories
\be
\mathscr{C}_d \overset{\text{def}}{=} \sD_\bX\big/(M_d)^\Z,\qquad d=2,4,6
\ee
are defined to have the same objects as $\sD_\bX$ and morphism spaces
\be\label{morro}
\sC_d(X,Y)=\bigoplus_{k\in\Z} \mathrm{Hom}_\sD(X, M_d^k Y).
\ee
Since $\mathsf{coh}\,\bX$ is hereditary, while the auto-equivalences $M_d$ satisfy the condition of Keller's theorem \cite{kellermainthm2}\footnote{\ See also \textsc{appendix}, \textbf{Theorem \ref{kellermainthm}}.}  the categories
$\sC_d$ are triangulated and the canonical  functors
\be
\sD_\bX\xrightarrow{\ \pi_d\ } \mathscr{C}_d,
\ee
are exact. In facts, the functor $\pi_d$ factorizes through the cluster category\footnote{\ This also follows from the universal property of the cluster category (\textsc{Appendix}, \textbf{Proposition \ref{uniprop}}).} 
$\mathscr{C}$ of $\bX$
\be\label{factfact}
\begin{gathered}
\xymatrix{\sD_\bX\ar[r]_{\pi_2\phantom{mmn}}\ar@/^2pc/[rr]^{\pi_d} & \mathscr{C}\equiv \mathscr{C}_2\ar[r]_{\phantom{mn}g_{d/2}}& \mathscr{C}_d}
\end{gathered}
\ee
$\sC\equiv\sC_2$ is the UV triangle category for 
$SU(2)$ with $N_f=4$ \cite{Cecotti:2015hca}. 
In $\mathscr{C}_d$ we have
\be
\sC_d(X,Y)\cong \sC_d((M_d)^{d/2} X, Y)= \sC_d(\tau^{-1}\Sigma X, Y)\cong D\sC_d(Y,\Sigma^2X),
\ee
so that the $\mathscr{C}_d$ is a (Hom-finite) $2$-CY triangle category.  Moreover, $\Sigma^2=(\tau^{-1}\Sigma)^2\cong\mathrm{Id}$, so the 2-CY category $\sC(d)$ is \emph{symmetric},
as we predicted on physical grounds for a QFT with $\Delta\in\bN$ (eqn.\eqref{2per}).

The exact surjective functors 
\be
g_{d/2}\colon \mathscr{C}\to \mathscr{C}_d,
\ee 
defined by diagram \eqref{factfact} give the quotient with respect  to a subgroup $\Z_{d/2}$ of the automorphism group of $\sC$, the UV category of $SU(2)$ with $N_f=4$.
It is then natural to identify the 2-CY quotient $\sC_d$ with the 2-CY category giving the UV description of the BPS sector for the SCFTs
obtained by gauging the discrete symmetry $\Z_{d/2}\subset SL(2,\Z)\times U(1)_R$ of $SU(2)$ $N_f=4$ i.e.\!  entries 14 and 10 of the table in \cite{Argyres3}.
To show that this identification is correct
we have to check a few facts. In particular
that $\sC_d$ has cluster-tilting objects
with the appropriate properties to reproduce the correct physics in the IR.
\medskip

The cluster category $\sC$ is best seen \cite{clucan} as the category whose objects are the coherent sheaves $X\in\mathsf{coh}\,\bX$ with morphism spaces\footnote{\ We reserve the notations $\mathrm{Hom}(X,Y)$ and $\mathrm{End}(X)\equiv \mathrm{Hom}(X,X)$ for morphism spaces in the category of coherent sheaves.}
\be\label{kkaswe}
\sC(X,Y)\cong \mathrm{Hom}(X,Y)\oplus 
D\mathrm{Hom}(Y,X),\quad X,Y\in\mathsf{coh}\,\bX.
\ee 
We write $S_d$ for the auto-equivalence of $\sC$ induced by
$M_d$ where $d=4,6$. One has
\be
(S_d)^{d/2}=\mathrm{Id}\qquad \text{in }\sC.
\ee
An indecomposable sheaf $X\in\mathsf{coh}\,\bX$ has a well-defined slope $\mu(X)\in \bP^1(\bQ)$ equal to the ratio of the electric and magnetic charge of the associated BPS object.
The telescopic functors $T$, $L$ act on $\mu$ as the modular transformations
\eqref{modtras}. 
Since $\mu(M_2X)=\mu(X)$ the slope is well-defined in $\sC$. One has
\be\label{unbranchedZ}
\mu(S_4X)=-\frac{1}{\mu(X)},\qquad 
\mu(S_6X)=\frac{1}{1-\mu(X)}.
\ee
Since these M\"obius transformations have no fixed points in $\bP^1(\bQ)$, $S_d$ acts freely on the indecomposable sheaves.
This is the crucial fact.
We define the pull-up functor\footnote{\ This is the usual pull-up functor of covering theory, see \textsc{appendix \ref{cov-tech}} for its basic  properties.}
\be
\iota_d\colon \sC_d\to \sC,\qquad \iota_d\colon X\mapsto \bigoplus_{k=0}^{d/2-1}S_d^kX.
\ee 
 Eqn.\eqref{morro} reduces to
 \be\label{kkkaz120}
 \sC_d(X,Y)\cong \sC(X,\iota_dY).
 \ee
 
 \subsubsection{Cluster-tilting and unbranched Galois covers}\label{aaa1098}
Eqn.\eqref{kkkaz120} implies that
\begin{equation}\label{p0as2}
\text{$T\in\sC_d$ is cluster-tilting}\ \Leftrightarrow\ 
\text{$\iota_d T\in\sC$ is cluster-tilting}
\ \Leftrightarrow\ 
\text{the \underline{sheaf} $\iota_dT$ is tilting}.
\end{equation}
Thus the cluster-tilting objects in $\sC_d$ are just the tilting objects of $\mathsf{coh}\,\bX$ in the range of $\iota_d$, i.e.\! the tilting sheaves fixed by $S_d$. They have the form
\be\label{jjjzas}
\iota_d\!\left(\bigoplus_{i=1}^\ell T_i \right)=\bigoplus_{k=0}^{d/2-1}\bigoplus_{i=1}^\ell S_d^k T_i,
\ee
with the $T_i$ indecomposables pair-wise non-isomorphic in $\sC_d$. In particular $\ell d=12$. 

From \eqref{p0as2} we see that to a  cluster-tilting object $T\in\sC_d$ we may associate three finite-dimensional basic algebras
\be
J_d\equiv \sC_d(T,T),\quad J\equiv\sC(\iota_dT,\iota_dT),\quad A\equiv\mathrm{End}(\iota_d T).
\ee
From eqn.\eqref{kkaswe} we see that $J$ is the \emph{trivial extension algebra} of $A$, $J=A\ltimes DA$.

Let us suppose (for the moment) that a cluster-tilting object $T\in \sC_d$ exists and
consider the quiver $Q$ of $J$.
The nodes of $Q$ are in one-to-one correspondence with the indecomposable summands $S^k_dT_i$ in 
eqn.\eqref{jjjzas}. Since $S_d$ is an auto-equivalence of $\sC$ acting freely on the indecomposables, $S_d$ generates a $\Z_{d/2}$ automorphism group of $Q$ acting freely on the nodes and preserving the relations. Under these conditions there exists a Galois cover of algebras (in the sense of Gabriel \cite{gab,galoiscover}; for a nice survey see \cite{pena}\,\footnote{\ Covering techniques in Representation Theory 
have been reviewed and applied to $\cn=2$ QFT in \cite{Cecotti:2015qha}. See in particular \S.2.4.1 where the covering functors are discussed in detail.})
\be
F\colon J\to J/\Z_{d/2}.
\ee 
The crucial point is that this is a very special kind of Galois cover: from eqn.\eqref{unbranchedZ} we see that it is an \emph{unbranched} cover\footnote{\ See  \textsc{appendix \ref{cov-fffunc}} for definitions and properties.}: by this we mean that it acts freely on the indecomposables.
Eqn.\eqref{kkkaz120} implies
\be
J_d=J/\Z_{d/2}.
\ee

$\iota_dT$ is a cluster-tilting object of $\sC$ with the special form\eqref{jjjzas}. In $\sC$ there are standard tilting objects $R$ so that the quiver of $\sC(R,R)$
is a 2-acyclic quiver with superpotential in the mutation class of $Q(\{1^4\},2)$, and $J$
is the corresponding Jacobian algebra.

\subparagraph{The IR picture.}
The BPS states\footnote{\ More precisely: the BPS states stable in a BPS chamber covered by the given quiver \cite{Alim:2011kw}.} of $SU(2)$ $N_f=4$
are given by the stable objects in the category $D^b\mathsf{mod}\,J$. 
Therefore, if there exists a tilting-object $T\in\sC_d$ such that $\iota_d T$ is standard, the
category 
\be
\sD_d\equiv D^b\mathsf{mod}\,J_d= D^b\mathsf{mod}\, J/\Z_{d/2}
\ee 
has the correct properties to describe the BPS states of a $\cn=2$ QFT obtained by gauging a $\Z_{d/2}$ symmetry of $SU(2)$ with $N_f=4$ whose BPS states are described by the derived Jacobian category $D^b\mathsf{mod}\,J$.
This follows from the fact that the Galois covering is unbranched. In this case,
the two covering functors,\footnote{\ We collect some useful properties of the covering functors in \textsc{appendix \ref{cov-fffunc}}.} push-down $F_\lambda$ and pull-up $F^\lambda$
\cite{galoiscover,pena}
\be\label{covfunct}
 \xymatrix{\mathsf{mod}\,J\ar@<0.6ex>[rr]^{F_\lambda} && \mathsf{mod}\,J_d\ar@<0.6ex>[ll]^{F^\lambda}}
 \ee
  set-up a correspondence between the BPS spectrum of the gauged and un-gauged theories which is the physically correct one provided the central charge (stability function) is equivariant
\be\label{zaq12m}
Z(S_dX)= e^{2\pi i/d}\,Z(X),\qquad X\in\mathsf{coh}\,X
\ee
and one uses the $Z_k$-covariant version of the Bridgeland stability condition with respect to the full  subcategory $\boldsymbol{t}$ of objects with
$0\leq \arg Z(X) < \pi/d$. This is equivalent to define a module of $X\in \mathsf{mod}\,J_d$ to be \emph{stable} iff it is the push-down of a module $Y\in\mathsf{mod}\,J$ which is stable with respect to a stability function $Z$ satisfying \eqref{zaq12m}.
The lift $Y\in\mathsf{mod}\,J$ is well-defined up to the $\Z_{d/2}$ action;
  there is an unique element of the orbit with $0\leq \arg Z(Y) < \pi/d$. Fixing this lift amounts to fixing the $\Z_{d/2}$-gauge; each object of $\sC_d$ is represented a unique object in $\boldsymbol{t}$. The Dirac pairing \emph{in the chosen gauge} $\boldsymbol{t}$ is the antisymmetric part of their Euler form $\chi(-,-)$ in $\mathsf{mod}\,J$ restricted to $\boldsymbol{t}$. Everything is well defined since the Galois cover is unbranched. Needless to say, this precisely agrees with the physical definition,
as discussed after eqn.\eqref{pppqwer}.

In terms of the 1d theory on the world-line, we get precisely the set-up predicted in section 2:
$\mathscr{L}$ is the 1d Lagrangian for the ungauged model, with $A(\mathscr{L})=J$, $\bG=\Z_{d/2}$ is the group to be gauged in 4d as well as in 1d, and $A(\mathscr{L}/\bG)=J_d$; to simplify the analysis of the FI terms, it is more convenient to work with the covering Lagrangian $\mathscr{L}$. 
\medskip

We stress that eqn.\eqref{zaq12m} reflects the physical fact that the double-cover of the gauged $\Z_{d/2}$ symmetry is a combination of a  $\Z_d\subset U(1)_R$ symmetry and a $\Z_d\subset SL(2,\Z)$ duality \cite{Argyres7}.

\subparagraph{Comparing with the Seiberg-Witten geometry.} The above discussion mimics what happens on the Seiberg-Witten geometry. The $\Z_{d/2}$ gauge symmetry acts on the Coulomb branch coordinate $u$ of $SU(2)$ with $N_f=4$
as $u\mapsto e^{4\pi i/d}\,u$. Away from the branching points $u=0,\infty$,
the local physics at vacuum $u$ is identical for the gauged and un-gauged physics; the fibers $\ce_{u^{d/2}}$ and $\ce^\prime_u$ are isomorphic elliptic curves, and the SW differentials $\lambda$, $\lambda^\prime$, restricted to these curves, are identified.
Thus the curves in these two fibers which are calibrated by $e^{-i\theta}\lambda|_{\ce_{u^{d/2}}}$ and respectively  $e^{-i\theta}\lambda^\prime|_{\ce^\prime_u}$, are also identified, that is, the BPS spectrum at a generic point in the Coulomb branch is identical for the gauged and ungauged theories. Stated, differently a BPS state, represented by a curve in the SW geometry, is stable in the gauged theory (at a generic vacuum) iff its pull-back to the covering ($\equiv$ ungauged) SW geometry is stable. This is the same condition stated before in terms of the unbranched Galois cover of algebras $J\to J_d$. \bigskip

We conclude:
\begin{itemize}
\item[1.] the interpretation of $\sC_d$ as the UV BPS category of a $\cn=2$ QFT obtained by gauging a $\Z_{d/2}$
of $SU(2)$ with $N_f=4$ requires
the existence of a standard tilting object
$T\in\sC_d$;
\item[2.] if a standard cluster-tilting object $T\in \sC_d$ exists, the endo-quiver $Q_T$ of its pull-up $\iota_dT$ is an element of the mutation class of $SU(2)$ $N_f=4$
($\equiv$ the class of triangulation quivers of the sphere with 4 ordinary punctures \cite{Cecotti:2011rv,fomin}), having an automorphism group at least $\Z_{d/2}$,
with the property that folding $Q_T$ along its $\Z_{d/2}$ symmetry produces an unbranched Galois cover between the IR BPS categories $\mathsf{mod}\,J$ and $\mathsf{mod}\,J_d$.
\end{itemize}

\subparagraph{Explicit cluster-tilting sheafs.} The mutation class of the sphere with 4 punctures contains just 4 quivers.\footnote{\ The mutation class of a quiver may be computed using Keller's Java applet \cite{applet}. The number of quivers in a class alway means their number up to source/sink equivalence.} Two of them have no automorphism. The other two are the $\Z_2$ symmetric quiver
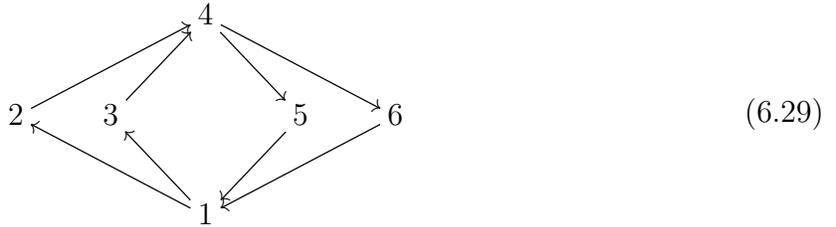
\be\label{sppquiver2}
\begin{gathered}
\xymatrix{&& 4\ar[rd]\ar[rrd]\\
2\ar[rru] & 3\ar[ru]&& 5\ar[dl] & 6\ar[dll]\\
&& 1\ar[lu]\ar[llu]}
\end{gathered}
\ee
whose symmetry acts on the node labels as
$i\mapsto i+3\bmod 6$, and the subtle quiver
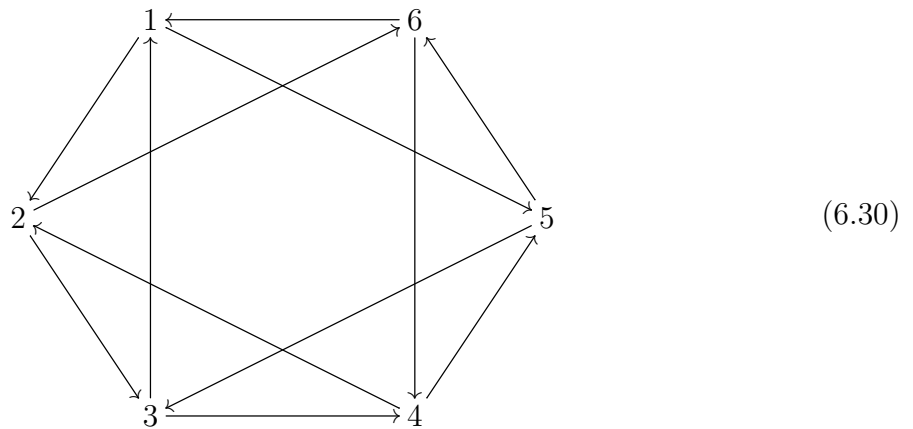
\be\label{z6qui2}
\begin{gathered}
\begin{xy} 0;<1pt,0pt>:<0pt,-1pt>:: 
(200,75) *+{5} ="0",
(150,150) *+{4} ="1",
(50,0) *+{1} ="2",
(0,75) *+{2} ="3",
(50,150) *+{3} ="4",
(150,0) *+{6} ="5",
"1", {\ar"0"},
"2", {\ar"0"},
"0", {\ar"4"},
"0", {\ar"5"},
"1", {\ar"3"},
"4", {\ar"1"},
"5", {\ar"1"},
"2", {\ar"3"},
"4", {\ar"2"},
"5", {\ar"2"},
"3", {\ar"4"},
"3", {\ar"5"},
\end{xy}
\end{gathered}
\ee
which has an obvious $\Z_6$ symmetry, acting freely, generated by $i\mapsto i+1\!\bmod 6$. However this is not the full story: this quiver has many other free automorphisms. Indeed the permutation $\sigma_i$ of the $i$ and $i+3$ nodes is a non-free automomorphism for all $i$'s. Conjugating the obvious $\Z_6$ by these non-free automorphisms, we get additional freely acting $\Z_6$'s. For instance conjugating with $\sigma_2\sigma_3$ the $\Z_3$ action $i\mapsto i+2\bmod 6$ get replaced by the inverse of the node permutation $(1,2,6)(5,3,4)$.

In \textsc{appendix \ref{ctc4}} it is shown explicitly that the sheaf
\be\label{firstclaim}
T=\co\oplus \co(\vec x_1)\oplus \cs_{2,1}
\ee
is a standard cluster-tilting object in $\sC_4$
such that the endo-quiver of $\iota_4T$ in $\sC$
is \eqref{sppquiver2} with the obvious $\Z_2$ symmetry. 
In \textsc{appendix \ref{q11zx}}
we show that the algebra $\mathrm{End}_{\sC_4}(T)$ corresponds to the non-2-acyclic quiver with the quartic superpotential
\be
\xymatrix{\bullet \ar@/^1pc/[r]^a &\bullet\ar@/^1pc/[l]^b\ar@/^1pc/[r]^c & \bullet\ar@/^1pc/[l]^d}
\qquad \cw = \mu_1 (ab)^2+\mu_2 (cd)^2+\varrho\, acdb
\ee
$\mu_1$, $\mu_2$, $\varrho$ being generic complex coefficients. The Jacobian algebra is finite-dimensional.
\medskip

In \textsc{appendix \ref{ctc6}} it is shown  that
\be\label{kkkaqwhhn}
T=\cs_{1,0}\oplus \cs_{2,1}
\ee
 is a standard tilting object in $\sC_6$
such that the endoquiver of $\iota_6T$
is \eqref{z6qui2} with the $\Z_3$ symmetry to be gauged corresponding to $i\to i+2\bmod6$.

This completes the identification of the UV and IR categories for the SCFT theories 
$(4,Spin(7))$ and $(6,G_2)$.
\medskip

\subsubsection{Grothendieck groups}
The Grothendieck group for the cluster category $\sC$ of $SU(2)$ with $N_f=4$
was discussed in \S.\ref{jjaszq}\footnote{\ See also \S.2.9.1 of \cite{Caorsi:2017bnp}.}.
We rewrite it in the sheaf notation. The group $K_0(\sC)$ is generated by the five classes $[\co]$, $[\cs_{a,0}]$, $a=1,2,3,4$,
subjected to the relation
\be\label{grorel}
-2[\co]=\sum_{a=1}^4[\cs_{a,0}].
\ee
Writing a class in $K_0(\sC)$ as 
$\sum_a w_a [\cs_{a,0}]$ yields the isomorphism
\be
K_0(\sC)\cong \Big\{(w_1,w_2,w_3,w_4)\in \Big(\tfrac{1}{2}\Z\Big)^2\;\Big|\; w_a=w_b\mod 1\Big\}\equiv \Gamma_{\text{weights},\,\mathfrak{spin}(8)} 
\ee
and the normalized Euler pairing is exactly the Cartan one on the weigth lattice of $\mathfrak{spin}(8)$ (which is valued in $\tfrac{1}{2}\Z$), i.e.\! $(w,w^\prime)_\mathfrak{f}=\sum_a w_a w_a^\prime\in\Z$.

The Grothendieck group of the IR category\footnote{\ The root category $\sR_\bX$ has the same Grothedieck group as $\sD_\bX$.} $K_0(\sD_\bX)$,
 is the free Abelian group over the six classes
$[\co]$, $[\cs_{a,0}]$ and $[\co(\vec c)]$.
The electric and magnetic charges $e$, $m$ are 
$$
e([\co])=0, \ m([\co])=1,\quad
e([\cs_{a,0}])=1,\ m([\cs_{a,0}])=1,\quad
e([\co(\vec c)])=2,\ m([\co(\vec c)])=1.
$$
The UV group $K_0(\sC_\bX)$ is obtained from the IR group $K_0(\sD_\bX)$ by imposing the relation
$[\co]\simeq [\co(\vec c)]$ and the one
in eqn.\eqref{grorel} \cite{groK}. Now the eletro-magnetic charges $e$, $m$ are well defined only mod 2, i.e.\! 
\be
(e,m)\in\Z_2\times \Z_2\cong Z(Spin(8)).
\ee
 Eqn.\eqref{grorel} says that states of even (odd) magnetic charge
are in tensor (spinor) representations of the flavor group $Spin(8)$, and that dyons of even/odd electric charge have opposite $Spin(8)$ chirality, as known from a direct physical analysis \cite{SW1}. They also say that the $PSL(2,\Z)$ action generated by the telescopic functors, $T$ and $L$, acts on the Grothendieck group $\equiv$ flavor weights by $\mathfrak{S}_3$ triality. 

The Grothendieck groups of $\sC_4$ and $\sC_6$ are obtained by taking the quotient of the $\mathfrak{spin}(8)$ weight lattice by a subgroup $\Z_2\subset\mathfrak{S}_3$ and, respectively, $\Z_3\subset \mathfrak{S}_3$; it is well known that these quotients  produce the weight lattices of $\mathfrak{spin}(7)$ and $G_2$, respectively. This is a manifestation of the universal relation between the Grothendieck group of the UV category and the flavor group of the QFT we claimed in section \ref{ccartoon}.  

\subparagraph{The Grothendieck groups $K_0(\sC_d)$.} We illustrate in detail $K_0(\sC_4)$ the case of $K_0(\sC_6)$ being similar. In \textsc{appendix \ref{kkkaqwe}} we show that in $K_0(\sC)$
\be
[S_4\co]=-[\cs_{3,0}],\qquad [S_4\cs_{a,0}]=[\cs_{a,0}]\ \text{for }\ a=1,2,4.
\ee
Then relation \eqref{grorel} becomes in $K_0(\sC_4)$
\be
-[\co]=[\cs_{3,0}]=\sum_{a\neq3}[\cs_{a,0}].
\ee
Thus $K_0(\sC_4)\simeq \Z^3$ is the free Abelian group
generated by the three classes $[\cs_{1,0}]$, $[\cs_{2,0}]$ and $[\cs_{4,0}]$,  isomorphic to the $\mathfrak{spin}(7)$ weight lattice. The \emph{normalized}
Euler form \cite{Caorsi:2017bnp} is the Cartan one on the $\mathfrak{spin}(7)$  weight lattice. 
For $K_0(\sC_6)$ the same argument gives $F=G_2$ and again $\kappa_F=4$.
Also the generic Higgs branch dimension remains $h=0$ since the covering functors \eqref{covfunct} guarantee that no `pure flavor' indecomposable is generated by the gauging.

\subsubsection{The \underline{un}-gaugeable $\Z_6$ symmetry} \label{obstrr}
It remains to explain the `unexpected' $\Z_6$ symmetry of the quiver \eqref{z6qui2}. We specialize the weighted projective line $\bX$ by fixing its exceptional points to be $\{1,-1,0,\infty\}$.
Then $\mathsf{coh}\,\bX$ and $\sC$ have a order 2 auto-equivalence $\Pi$ permuting the first two points which commutes with the telescopic functors and hence with $S_6$ and (obviously) $\tau$. Then
$H=\tau\Pi S_6$ is an order 6 equivalence of 
$\sC$ with $S_6^{-1}=H^2$. The tilting object in eqn.\eqref{kkkaqwhhn} has the form
\be
\iota_6 T= \bigoplus_{i=0}^5 H^kT
\ee
so the endo-quiver of $T$ must have a freely acting $\Z_6$ symmetry which is what we found in \eqref{z6qui2}. The symmetry is not unique since $\Pi$ corresponds to an element of the flavor Weyl group $\mathrm{Weyl}(D_4)$, and we may replace it by any element in its conjugacy class.

$H$ lifts to the autoequivalence $\tau \Pi TL$ of $\sD_\bX$.
The orbit category
\be\label{kzqw99i}
\sD_\bX/(\tau \Pi TL)^\Z
\ee
is triangulated and 2-periodic,
but it is \underline{not} 2-CY since
$\tau^{-1}\Sigma$ is not a power of
$\tau\Pi TL$. 
One has
\be
(\tau \Pi TL)^6=\tau^{-2} \Sigma^2\quad\text{in }\sD_\bX, 
\ee
so the category \eqref{kzqw99i} is fractional Calabi-Yau of dimension
$4/2$ rather than $2$.
This suggests that the $\Z_6$ symmetry of $SU(2)$ $N_f=4$ cannot be consistently gauged while preserving $\cn=2$ \textsc{susy}.

The last statement is well known to physicists \cite{Argyres7}: $\Pi$ generates a finite subgroup of the flavor group $Spin(8)$, and the gauging of such a subgroup 
is not consistent with $\cn=2$ \textsc{susy} \cite{Argyres7}. 

%\section{General aspects of discrete gaugings}
%
%In this section we formalize what we learned from our warm-up example as a general paradigm for the gauging of a discrete (Abelian) symmetry $\bG$ of a $\cn=2$ parent theory with particular regard to obstructions to gauging certain discrete symmetries as we found in the example of \S.\ref{obstrr}.
%
%Before doing that we have to stress the fundamental difference between two cases:
%$h=0$ (pure Coulomb branch) and $h\geq1$
%(enhanced Coulomb branch).

\section{General facts about discrete gaugings}

\subsection{Discrete gauging of $\cn=2$ theories with $h\geq1$}\label{sub:h1}

We stress a fundamental difference between gauging a discrete subgroup of the duality group  of a $\cn=2$ SCFT having a pure Coulomb branch ($h=0$) or an enhanced Coulomb branch ($h\geq1$). 
In the warm-up examples of section 5 the discrete gaugings turned out to be, in the categorical language, just unbranched Galois covers with Galois group the discrete gauge group $\bG$. This was so because $\bG$ was a group of auto-equivalences of the UV category $\sC$
with the nice property of acting freely on \emph{all} indecomposable objects, see eqn.\eqref{unbranchedZ} for $SU(2)$ with four flavors. This is the general pattern of discrete gaugings for all triangular QFT with pure Coulomb branches,
i.e.\! $h=0$. Indeed, when $h=0$ the electro-magnetic charge $(e,m)$ is non-zero for all BPS states (a state with $(e,m)=0$ would be everywhere light on the Coulomb branch hence part of the generic Higgs branch).
Thus, in a triangular model with $h=0$,  all BPS states have a well-defined slope $\mu$
\be
\mu\equiv \text{\textsf{(electric charge)}}\big/\text{\textsf{(magnetic charge)}}\in \bP^1(\bQ),
\ee
on which the rank-1 electro-magnetic duality acts projectively
\be
\mu\mapsto \frac{a \mu+b}{c\mu+d},\quad \qquad \begin{bmatrix} a & b\\
c & d\end{bmatrix}\in PSL(2,\Z).
\ee
At the categorical level this reflects into the fact that the discrete gauge group $\bG$ acts freely on the indecomposable objects of the UV category $\sC$, as we saw in eqn.\eqref{unbranchedZ} for $SU(2)$ with four flavors. This leads to the pattern we observed in the previous section:

\begin{quote}\it A consistent \underline{true} gauging of a finite Abelian symmetry $\bG$ of a $\cn=2$ triangular theory with $h=0$ corresponds to an automorphism of a 2-acyclic quiver in its mutation class which acts freely on the nodes.  
\end{quote}

On the contrary, if $h>0$, the objects in the generic Higgs branch have electro-magnetic charges $(e,m)=(0,0)$ which \underline{do not} define any point in $\bP^1(\bQ)$. The flavor symmetry acting effectively on the generic Higgs branch is $Sp(2h)$, whose Lie algebra has no outer automorphism,
while $U(1)_R$ acts trivially on the Higgs branch scalars. Therefore the discrete group $\bG$ leaves the Higgs objects  invariant: the nice property of the $h=0$ case that the $\bG$-action is free on the indecomposables is lost. Then

\begin{quote}\it A consistent gauging of a finite Abelian symmetry $\bG$ of a $\cn=2$ triangular theory with $h\geq 1$ does not \emph{necessarily} correspond to an automorphism of a 2-acyclic quiver in its mutation class which acts freely on the nodes.  
\end{quote}

In other words: for $h\geq1$ the existence of a freely acting quiver automorphism is a \emph{sufficient} condition for the existence of the given gauging (provided it is not obstructed for a different reason: cfr.\!  \S.\ref{obstrr}) but not at all \emph{necessary.}

\subsection{$h=0$: Gaugeable symmetries}

Let us focus on the $h=0$ case.
From the IR viewpoint (in a given BPS chamber)
the original theory is described by the module category $\mathsf{mod}\,J$ of the Jacobian algebra $J$,
and the physical system has a symmetry group contained in the automorphisms $\mathrm{Aut}(J)$ of $J$. We wish to gauge a finite subgroup $\bG\subset \mathrm{Aut}(J)$. 
The gauged theory \textit{if it exists}
should correspond (in the IR and in the corresponding chamber) to the module category $\mathsf{mod}\,J_\text{ga}$
of another Jacobian algebra by the 1d argument of section \ref{ccartoon}. 
It is clear that `discrete gauging' should correspond to some functor between these two categories
\be
F_\lambda\colon \mathsf{mod}\,J\to \mathsf{mod}\,J_\text{ga}.
\ee
The natural guess, which revealed to be correct in the explicit examples of the previous section, is
that $F_\lambda$ is the push-down of some covering functor $F\colon J\to J_\text{ga}$ whose Galois group is $\bG$.

However, this is mathematically surprising because it is quite uncommon\footnote{\ We thank B. Keller for pointing out this to us.} that a finite-dimensional 
Jacobian algebra is covered by another algebra which is again Jacobian and finite-dimensional\footnote{\ See however \cite{Cecotti:2015qha}.}.
Dually, from the physics point of view it  would be quite surprising if we were able to gauge all or `most' finite subgroups of $\mathrm{Aut}(J)$. Typically the gauging of a subgroup $\bG\subset \mathrm{Aut}(J)$ is obstructed either by 't Hooft anomalies (so the gauged QFT does not exists at all) or because the gauging breaks $\cn=2$ \textsc{susy} (so its BPS category does not make any sense). Existence of non-trivial consistent discrete gauging preserving $\cn=2$ \textsc{susy} is an unexpected miracle that only recently was discovered to happen. 

Is it possible that the mathematically uncommon feature and the physically  rare `miracle' are two faces of the same coin?
Indeed they are: the discrete gauging of a subgroup $\bG\subset \mathrm{Aut}(J)$ is possible precisely if there is a Galois quotient $J\to J/\bG$ which is unbranched in the sense of definition \ref{unbranched}. Note that being unbranched is much stronger a condition than being admissible \cite{gab,galoiscover,pena}: in the second case $\bG$ is required to act freely on the simple modules of $\mathsf{mod}\,J$ while to be unbranched $\bG$ should act freely on all \emph{indecomposables.} 
In the 1d physical language, the Galois cover being unbranched is equivalent to the statement that
\textit{the gauging \emph{(i.e.\! the push-down functor $F_\lambda$)} preserves the single-trace chiral operators.} 
We conclude:
\begin{cri} Whenever the orbit category of the UV 2-CY category $\sC$,
$\sC/\bG$, is not triangulated, or not 2-CY, or not
Hom-finite, or has no cluster-tilting object $T$, in
the corresponding QFT the gauging of the discrete symmetry $\bG$ (while preserving $\cn=2$ supersymmetry) \textsc{is obstructed}. An unobstructed discrete gauging requires the pull-up $F^\lambda T$ of the cluster-tilting object $T\in\sC/\bG$ to be a cluster-tilting object for $\sC$ such that the Jacobian algebra $\mathrm{End}_\sC(F^\lambda T)$ is an \emph{unbranched} Galois cover of the
Jacobian algebra $\mathrm{End}_{\sC/\bG}(T)$.\end{cri}

In the above statement there are several clauses. They are \emph{not} all on the same footing. Hom-finite 2-CY is a UV `sacred' condition, and cannot be violated without paying the price of getting physical non-sense. The existence of a cluster-tilting object is \emph{not} a condition for physical sense, it is merely a criterion for the construction to be interpretable as a straight-forward discrete gauging. More general constructions exist, see section 8.
%Here it suffices to say that the cluster-tilting object $T$ connects the UV description to the IR one and the RG flow needs not to be as plain as in our warm-up examples.
\medskip   

Let us see the criterion at work.
Physics says that
there are (at least) three necessary criterions in order for a discrete gauging to be consistent:
\textit{i)} the gauge group $\bG$ should acts trough a non-trivial finite group of $S$-duality transformations; \textit{ii)}
the $S$-duality rotations should be accompanied by a $U(1)_R$ rotation.
\textcolor{red}{Since the finite group $\bG$ has a faithful one-dimensional representation, $\bG\simeq \Z_n$ for some $n$;}
\textit{iii)} $\bG$ can act on the flavor Lie algebra only through outer automorphisms,
actions of the flavor Weyl group being obstructed.

Condition \textit{i)} is akin to saying    
that $\bG$ acts on the slope $\mu$ without fixed points; if the indecomposables have a well defined slope (as in section \ref{za10k}) this implies that $\bG$ acts freely on the indecomposables, namely $J\to J/\bG$ is
an unbranched Galois cover of bounded linear categories. Given the usual relations between the slope and the stability function ($\equiv$ the \textsc{susy} central charge) this may be consistent only if the action of $\bG$ on the slope and on $Z$ is equivariant, which is precisely \textit{ii)}.
Finally, the generator $\theta$ of the gauge group $\bG\cong\Z_n$ will act on the flavor root latice $(K_0(\sC)/\text{(torsion)})^\vee$ by an element $\theta\in\mathrm{Weyl}(F)\rtimes \mathrm{Aut}(\textsf{Dyn})$ where $\mathrm{Aut}(\textsf{Dyn})$ are the automorphisms of the Dynkin graph.  $\theta$ permutes the direct summands of the covering cluster-tilting object $F^\lambda T$, i.e.\! the simple roots of $F$. Then it is a non-trivial automorphism of the Dynkin graph, as we saw in the previous examples.

\section{The other 8 discrete gaugings}

The RT description of the discrete gaugings of $SU(2)$ with $N_f=4$ was rather straightforward because we have a very explicit description of the cluster category for the ungauged model as the orbit category of a well-studied derived  hereditary category.
We have a similar explicit description for the cluster categories of the Argyres-Douglas (AD) models (the SCFTs with $\Delta<2$).
The analysis of discrete gaugings of  AD models is then similar to the one in the previous section. In facts, their possible gaugings were already worked out in the math literature \cite{finite1,finite2,finite3} for different purposes.
On the contrary, for the ungauged models
with $\Delta>2$ (the wild case) we do not have such a simple description of the cluster category.  
%However the cluster category of $(3,Sp(6))$ may be related by a trick [cite] to a triangular one, which may be used to draw the relevant conclusions.

\subsection{Discrete gaugings of Argyres-Douglas models}
\subsubsection{$\Z_2$ gauging of $(\tfrac{3}{2},SU(3))$}\label{zzza12}\label{muu1}

According to table \ref{missing} the $\Z_2$ gauging of the SCFT $(\tfrac{3}{2},SU(3))$ produces $(3,SU(3))$ (entry 21). In the quiver mutation class of $(\tfrac{3}{2},SU(3))$ there are two kinds of suggestive quivers: the orientations of the $D_4$ Dynkin graph
and the 4-cyclic quiver with $\cw$ equal to the cycle (a.k.a.\! the `square tensor product' $A_2\square A_2$ of two $A_2$ quivers \cite{kellerper,CNV}). 
\be\label{d4quiv}
\begin{gathered}\xymatrix{& b\ar[d]\\
a\ar[r]& \bullet & c\ar[l]}\\
D_4\ \text{Dynkin}\end{gathered}
\qquad\qquad 
\begin{gathered}\xymatrix{1\ar[d] & 4\ar[l]\\
2\ar[r] & 3\ar[u]}\\
\text{4-cyclic}\end{gathered}
\ee
The category $\ch$ of finite-dimensional moduli of the path algebra $\C D_4$ of the Dynkin quiver is hereditary with an AR translation\footnote{\ For the explicit action of $\tau$ see \textsc{appendix \ref{d4ex}}.} $\tau$. $\tau$ extends to an auto-equivalence of the triangle category $D^b\ch$ with \cite{picard,Caorsi:2016ebt}
\be\label{d444}
\tau^{-3}=\Sigma.
\ee
Hence
\be\label{gaugeable}
(\tau^{-1})^4=\tau^{-1}\tau^{-3}=\tau^{-1}\Sigma,
\ee
and $\tau^{-1}$ generates a $\Z_4$ symmetry of the cluster category $\sC_{D_4}$. There is a related (but different) $\Z_4$ auto-equivalence $\tilde\tau$ which just rotates the 4-cyclic quiver by $\pi/2$.\footnote{\ $\tilde\tau$ is the AR translation in the Abelian category of nilpotent modules of the path-algebra of the cyclic quiver. As an auto-equivalence of $\mathsf{nil}\,\C(\text{4-cycle})$ its satisfies $\tilde\tau^4=\mathrm{Id}$.} In the cluster category $\tilde\tau\cong\theta\tau$, where $\theta\in\mathfrak{S}_3$ is an order 2 automorphism of the Dynkin quiver which we choose to be the interchange of the first two nodes $a\leftrightarrow b$ in the first quiver of \eqref{d4quiv}\footnote{\ For details see \S.6.2.1 of \cite{Caorsi:2016ebt}.}
In particular, $\tau^2\cong \tilde\tau^2$ in the cluster category.

The UV category of the $\Z_2$-gauged SCFT $(3,SU(3))$ is then the orbit category
\be
\sC_{3,SU(3)}=D^b\ch/(\tau^2)^\Z,
\ee
which is triangulated and 2-CY. It is also \emph{symmetric} since $\mathrm{Id}\simeq (\tau^2)^{-3}=\Sigma^2$; this is consistent with the fact that the Coulomb dimension of the gauged theory is an integer. More precisely, the quantum monodromy $\Sigma^2$ is the \emph{cube} of an operator, $\tau^{-2}$, which acts as the identity in the UV category of the gauged model. This happens iff all Coulomb dimensions are not only integral but multiples of 3 \cite{CNV}, as indeed is the case.

Again the universality property of the cluster category $\sC_{D_4}$ yields a triangle gauging functor
relating the UV categories of the un-gauged and gauged SCFTs
\be
g\colon \sC_{D_4}\to \sC_{3,SU(3)}.
\ee

Repeating the arguments of section 5, if $T\in \sC_{3,SU(3)}$ is a cluster-tilting object, then $T\oplus\tau^{-2} T$ is cluster-tilting in $\sC_{D_4}$. Thus $T$ is the direct sum of exactly 2 indecomposables
of $\sC_{D_4}$ in agreement with $\mathrm{rank}\,F=2$. Moreover, the
endo-quiver $\mathrm{End}_{\sC_{D_4}}(T\oplus \tau^{-2}T)$ should be a quiver in the $D_4$-mutation class wich has a cyclic $\Z_2$-symmetry acting freely on the nodes. The only quiver with this property is the 4-cyclic one \eqref{d4quiv}.

Since the $\Z_2$ gauged model has $\mathrm{rank}\,F=2$ we need a cluster-tilting object $T$ with exactly two indecomposable summands. Consider the object
\be
T= \text{\begin{scriptsize}$\begin{bmatrix}&0&\\1& 0&0\end{bmatrix}$\end{scriptsize}}\oplus \theta\tau^{-1} \text{\begin{scriptsize}$\begin{bmatrix}&0&\\1& 0&0\end{bmatrix}$\end{scriptsize}}= \text{\begin{scriptsize}$\begin{bmatrix}&0&\\1& 0&0\end{bmatrix}$\end{scriptsize}}\oplus\text{\begin{scriptsize}$\begin{bmatrix}&0&\\1& 1&1\end{bmatrix}$\end{scriptsize}} \in \sC_{3,SU(3)}
\ee 
where the indecomposables are represented by their dimensions.
The corresponding $\tau^2$-invariant object in $\sC_{D_4}$ is\footnote{\ $S_1$ is the simple module of $\C D_4$
with support on node 1 of the Dynkin quiver \eqref{d4quiv}.}
\be
T\oplus \tau^{-2}T\equiv \bigoplus_{k=0}^3 (\theta\tau^{-1})^k\, S_1=
\text{\begin{scriptsize}$\begin{bmatrix}&0&\\1& 0&0\end{bmatrix}$\end{scriptsize}}\oplus \text{\begin{scriptsize}$\begin{bmatrix}&0&\\1& 1&1\end{bmatrix}$\end{scriptsize}}\oplus \text{\begin{scriptsize}$\begin{bmatrix}&0&\\1& 1&0\end{bmatrix}$\end{scriptsize}}\oplus \text{\begin{scriptsize}$\begin{bmatrix}&1&\\0& 1&0\end{bmatrix}$\end{scriptsize}}[-1]
\ee
is cluster-tilting with endo-quiver the 4-cycle with a $\Z_4$-symmetry acting cyclically on the nodes generated by the permutation $(1234)$,
while the $\Z_2$-symmetry to be gauged is generated by $(13)(24)$. Indeed,
\be
\sC_{D_4}\big(S_1,(\theta\tau^{-1})^k S_1\big)= \begin{cases} \C & k=0\mod 4\\
0 & \text{otherwise.}
\end{cases}
\ee

Therefore also the object
$T\in\sC_4$ is cluster-tilting and its endo-quiver with superpotential $\cw$ is the Galois quotient of the 4-cyclic one by the free $\Z_2$-action
\be
\begin{gathered}
\xymatrix{\bullet\ar@/^1.2pc/[rr]^a&& \bullet\ar@/^1.2pc/[ll]^b}
\end{gathered}\qquad\quad \cw=(ab)^2.
\ee

\subparagraph{Grothendieck group.}
$K_0(\sC_{D_4})$ is described in ref.\cite{groK} and reviewed in \textsc{appendix \ref{d4ex}}. There we show that
$K_0(\sC_{D_4})\cong \Z[S_1]\oplus \Z[\Sigma S_2]$ and in this basis the matrix of the natural quadratic form is
\be
\Delta\, C^{-1},\qquad  C= \text{the $A_2$ Cartan matrix}
\ee
in agreement with the physical expectation.
Using formulae from \textsc{appendix \ref{d4ex}} we see
\be
[\tau^2 S_1]=[S_1],\qquad [\tau^2 \Sigma S_2]=[\Sigma S_2],
\ee
so the gauging quotient preserve the Grothedieck group. Then 
$K_0(\sC_{3,SU(2)})\equiv K_0(\sC_{D_4})$ is the $SU(3)$ lattice with the same quadratic form as in the ungauged model. Thus, in this case, the gauging leaves unchanged the flavor symmetry $SU(3)$.

\subsubsection{$\Z_4$ gauging of $(\tfrac{3}{2},SU(2))$} \label{muu2}

This $\Z_4$ gauging produces $(6,SU(2))$
(entry 11).
In the set-up of the previous subsection, note that the $\Z_4$ symmetry generated
by $\theta\tau$ is gaugeable since $(\theta\tau)^{-4}=\tau^{-1}\Sigma$, while
$\tilde\tau\cong\theta\tau$ generates the $\Z_4$ automorphism of the 4-cyclic quiver. Then the proper UV 2-CY category is
\be
D^b\ch/(\theta\tau)^\Z.
\ee 
Again it is symmetric since $(\theta\tau)^{-6}=\Sigma^2$ and
$(\theta\tau)^{-1}$ is a $\tfrac{1}{6}$-monodromy acting as 1, consistently with the fact that $\Delta=6$.
  
An example of cluster-tilting object is
\be
T\cong \text{\begin{scriptsize}$\begin{bmatrix}&0&\\1& 0&0\end{bmatrix}$\end{scriptsize}}.
\ee
Now $[\theta\tau S_1]=[S_2]$ and $[\theta\tau S_2]=[S_1]$ and $K_0(\sC)\cong\Z[S_3]$. The flavor group is $SU(2)$.

\begin{rem}\label{flagau} $\theta$ acts as an element of the flavor Weyl group of $D_4$ AD. This does not contradict the statement that we cannot gauge a discrete subgroup of flavor while preserving $\cn=2$ \textsc{susy}. In facts $\tau$ acts on the $SU(3)$ flavor weights $w$ as the outer automorphism $w\mapsto -w$, so that the above $\Z_4$ symmetry acts on the flavor by outer automorphisms.
\end{rem}

\subsubsection{$\Z_3$ gauging of $(\tfrac{4}{3},SU(2))$}\label{jjjza12n}

The situation is similar to the one in \S.\ref{zzza12}.
The $\Z_3$ gauging of the SCFT $(\tfrac{4}{3},SU(2))$ produces $(4,SU(2))$ (entry 18). In the mutation class there are the orientations of the $D_3$ Dynkin graph
and the 3-cyclic quiver with $\cw$ the cycle. 
\be
\begin{gathered}\xymatrix{
\bullet& \bullet\ar[l]\ar[r] & \bullet}\\
\text{Dynkin}\end{gathered}
\qquad\qquad 
\begin{gathered}\xymatrix{1\ar[d] \\
2\ar[r] & 3\ar[ul]}\\
\text{3-cyclic}\end{gathered}
\ee
The category $\ch$ of finite-dimensional moduli of the path algebra of the Dynkin quiver is hereditary, with an AR translation $\tau$. $\tau$ extends to an autoequivalence of the triangle category $D^b\ch$ with \cite{picard}
\be
\tau^{-2}=\theta\Sigma,
\ee
where $\theta$ is the order 2 autoequivalence induced by the exchange of the two end nodes, which acts as the non-trivial element of the flavor Weyl group $\mathrm{Weyl}(F)$. One has
\be
(\theta\tau^{-1})^3=\tau^{-1}\theta \tau^{-2}=\tau^{-1}\Sigma,
\ee
and the $\Z_3$ symmetry is gaugeable.
The UV category 
\be
D^b\ch/(\theta\tau^{-1})^\Z
\ee
is triangular, 2-CY and symmetric, since
$(\theta\tau^{-1})^4=\Sigma^2$ with a $\tfrac{1}{4}$-monodromy acting as 1, consistently with $\Delta=4$.
An example of a cluster-tilting object is
\be
T\cong \text{\begin{scriptsize}$\begin{bmatrix}1& 0&0\end{bmatrix}$\end{scriptsize}}.
\ee
Remark \ref{flagau} applies also to this $\Z_3$ gauging.

\begin{rem} The similarity between the last two examples becomes evident if we note that $A_3$ may be seen as $D_3$.
The above formulae yield a $\Z_r$ gauging of all Argyres-Douglas models of type $D_r$. More generally, we may gauge a subgroup $\Z_d$ for all divisors $d\mid r$.
\end{rem}

\subsubsection{No other AD gauging}

Let us show that the above list of gaugings is complete, i.e.\! that there is no other discrete gauging of the rank-1 SCFT with $\Delta<2$. (This also follows from the theorems in refs.\!\cite{finite1,finite2,finite3}, but we prefer to perform a very explicit analysis to illustrate the physical point). 
If $\Gamma$ is a Dynkin graph, we write $D(\Gamma)$ for the bounded derived category of the modules of the path algebra of any orientation\footnote{\ $D(\Gamma)$ is independent of the orientation up to equivalence.} of $\Gamma$, and $\mathrm{Aut}(D(\Gamma))$ for the group of its auto-equivalences.

By the previous discussion, a non-trivial gauging of the Argyres-Douglas model of type $\Gamma\in ADE$ corresponds to $\varrho\in \mathrm{Aut}(D(\Gamma))$
which satisfies two conditions: 
\be\label{conddz}
\begin{aligned}
&1) &&\varrho^n=\tau^{-1}\Sigma,\qquad
n> 1\\
&2) && \text{The orbit category $\sC_\varrho\equiv D(\Gamma)/(\varrho)^\Z$ contains a cluster-tilting object.}
\end{aligned}
\ee

\subparagraph{AD model of type $A_2$.} $\mathrm{Aut}(D(A_2))$ is the Abelian group generated by
$\tau$ and $\Sigma$ with the unique relation $\tau^{-3}=\Sigma^2$.
Then $\Sigma=(\tau\Sigma)^3$, $\tau=(\tau\Sigma)^{-2}$ and $\mathrm{Aut}(D(A_2))$ is the infinite cyclic group $\Z$ generated by $\tau\Sigma$. $\tau^{-1}\Sigma=(\tau\Sigma)^5$, and the only solution to 1) is $\tau\Sigma$ with $n=5$. This corresponds to the $\Z_5$ $S$-duality group of the AD model of type $A_2$ which cyclically permutes the 5 indecomposable objects of its cluster category $\sC$.
Now let $T\in\sC_\varrho$ be any non-zero object\footnote{\ We write $\sC\equiv \sC_{\tau^{-1}\Sigma}$ for the cluster category of the ungauged model.}
\be
\sC_\varrho(T,T[1])\cong\sC_\varrho(T,\tau T)\cong \bigoplus_{k=0}^4 \sC(T,(\tau\Sigma)^{k-2}T)=\sC(T,T)\oplus\cdots \neq 0,
\ee
and condition 2) cannot be satisfied.
In other words, the $\Z_5$ $S$-symmetry of the $A_2$ Ad model is not gaugeable.

\subparagraph{AD model of type $A_3$.}
$\mathrm{Aut}(D(A_3))\cong \Z\times \Z/2\Z$ is the Abelian group generated by $\tau$ and the involution $\theta$. One has
$\Sigma=\theta\tau^{-2}$ and $\tau^{-1}\Sigma=\tau^{-3}\theta$.
The only solution to 1) is the one
associated to the gauging already considered in \S.\ref{jjjza12n}. 

\subparagraph{AD model of type $D_4$.}
$\mathrm{Aut}(D(D_4))\cong \Z\times \mathfrak{S}_3$ where the cyclic group is generated by $\tau$ and $\mathfrak{S}_3$ is the triality automorphism of the Dynkin graph. $\Sigma=\tau^{-3}$ and $\tau^{-1}\Sigma=\tau^{-4}$. Up to conjugacy, there are four solutions to condition 1)
\be\label{solrho}
\tau^{-2},\qquad\theta\tau^{-1},\qquad \tau^{-1}\qquad\theta\tau^{-2},
\ee
where $\theta$ is an element of order 2 in 
$\mathfrak{S}_3$.
The first two solutions correspond to the gaugins in \S.\ref{muu1} and \ref{muu2}.
The orbit category of the third one has no non-zero rigid object since
\be
\sC_{\tau^{-1}}(X,X[1])\cong \sC(X,X)\oplus\cdots
\ee
The fourth solution would correspond to gauging a $\Z_2$ discrete symmetry acting on the flavor by inner automorphisms, 
which is inconsistent on physical grounds. The corresponding RT statement is that $\sC_{\theta\tau^{-2}}$ has no cluster-tilting objects. Let us check that the statement is indeed correct. If $T_1\oplus T_2$ is cluster tilting for $\sC_{\theta\tau^{-2}}$, then 
$\tilde T\equiv T_1\oplus T_2\oplus \theta\tau^2 T_1\oplus \theta\tau^2 T_2$ is cluster-tilting for $\sC$.
Rigid indecomposables of $D(D_4)$
form four orbits
under the subgroup of $\mathrm{Aut}(D(D_4))$ generated by $\tau$ and $\Sigma$ generates by the four simples.
The indecomposables in the orbits of the two $\theta$-invariant simples are rigid in
$\sC_{\theta\tau^2}$ iff they are in $\sC_{\tau^2}$, i.e.\! only if they belong to the $\theta$-invariant peripheral node. The indecomposables in the orbits of the other two peripheral simples are not rigid
We have\footnote{\ We choose 
$\theta\colon\text{\begin{scriptsize}$\begin{bmatrix}&0&\\1& 0&0\end{bmatrix}$\end{scriptsize}} \mapsto \text{\begin{scriptsize}$\begin{bmatrix}&0&\\0& 0&1\end{bmatrix}$\end{scriptsize}} $.}
\begin{equation*}
\sC_{\theta\tau^{-2}}\!\left(\text{\begin{tiny}$\begin{bmatrix}&0&\\1& 0&0\end{bmatrix}$\end{tiny}},\text{\begin{tiny}$\begin{bmatrix}&0&\\1& 0&0\end{bmatrix}$\end{tiny}}[1]\right)\cong \sC\!\left(\text{\begin{tiny}$\begin{bmatrix}&0&\\1& 0&0\end{bmatrix}$\end{tiny}},\theta\tau^{-1}\text{\begin{tiny}$\begin{bmatrix}&0&\\1& 0&0\end{bmatrix}$\end{tiny}}\right)\cong
\sC\!\left(\text{\begin{tiny}$\begin{bmatrix}&0&\\1& 0&0\end{bmatrix}$\end{tiny}},\text{\begin{tiny}$\begin{bmatrix}&1&\\1& 1&0\end{bmatrix}$\end{tiny}}\right)\cong\C
\end{equation*}
so the rigid indecomposables of the cluster category $\sC$ which are rigid in $\sC_{\theta\tau^{-2}}$ are
\be
\text{\begin{tiny}$\begin{bmatrix}&1&\\0& 0&0\end{bmatrix}$\end{tiny}}
\xrightarrow{\ \tau^{-1}\ }
\text{\begin{tiny}$\begin{bmatrix}&0&\\1& 1&1\end{bmatrix}$\end{tiny}}
\xrightarrow{\ \tau^{-1}\ }
\text{\begin{tiny}$\begin{bmatrix}&1&\\0& 1&0\end{bmatrix}$\end{tiny}}
\xrightarrow{\ \tau^{-1}\ }
\text{\begin{tiny}$\begin{bmatrix}&1&\\0& 0&0\end{bmatrix}$\end{tiny}}[1]
\ee
which are isomorphic in pairs as objects of $\sC_{\theta\tau^{-2}}$. Since $X\oplus \tau^{-1}X$ is never rigid, we find precisely two maximal rigid objects which are indecomposables, i.e.\! $S_2$ and $\tau^{-1}S_2$. Now the rigid obejct
$S_2\oplus \theta\tau^{-2}S_2$
is not cluster tilting in $\sC$ so the fourth gauging does not exists.

We conclude that for rank-1 $\Delta<2$ the `discrete gauging' in the sense of RT exactly reproduce the physically allowed ones.

\subsection{The remaining 5 discrete gaugings}

For the remaining cases we have no simple explicit realization of the cluster category for the ungauged theory.\footnote{\ The cluster category of $(3,E_6)$ is the Amiot category of a del Pezzo algebra $\ca_6$. In principle, one should be able to describe the gauging in terms of sheaves on the del Pezzo surface.} We should argue in an \emph{ad hoc} manner. 

\subsubsection{$\Z_2$ gauging of $(3,E_6)$}\label{e6mn}

The $\Z_2$ gauging of the SCFT $(3,E_6)$ produces $(6,F_4)$ (entry 4). 
In the mutation class of the quiver $Q(1^6;3)$ there is a unique
$\Z_2$ symmetric one analogue to \eqref{sppquiver2} 
\be
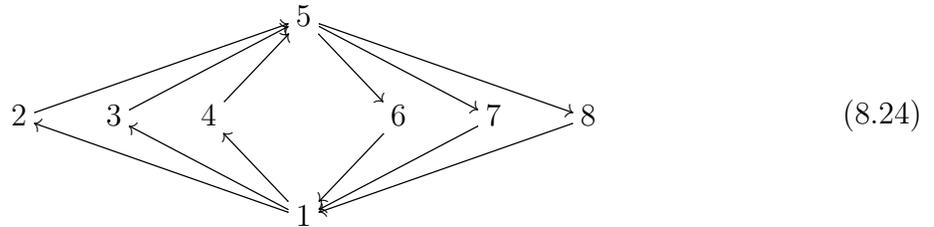
\label{sppquiver3}
\begin{gathered}
\xymatrix{&&& 5\ar[rd]\ar[rrd]\ar[rrrd]\\
2 \ar[rrru] &3\ar[rru] & 4\ar[ru]&& 6\ar[dl] & 7\ar[dll]& 8 \ar[dlll]\\
&&& 1\ar[lu]\ar[llu]\ar[lllu]}
\end{gathered}
\ee
with involution $i\mapsto i+4\bmod8$. The superpotential $\cw$ is invariant under the involution. This action induces a $\Z_2$ symmetry of the cluster category $\sC$ generated by an auto-equivalence\footnote{\ Identifying $E_6$ Minahan-Nemeshanski with the SCFT of type $D_2(Spin(8))$, 
in the notations \S.6.3.1 of \cite{Caorsi:2016ebt}, we have $M=1\otimes \tau$.} $M$ with $M^2=\text{Id}$.  Hence
in cluster category $\sC$ of $(3,E_6)$ there is a  cluster-tilting object $T=\oplus_{i=1}^8T_i$, unique up to auto-equivalence, such that $T_{i+4}\cong MT_i$, whose endo-quiver is \eqref{sppquiver3}. Endowing the quiver with its symmetric superpotential $\cw$, the quiver automorphism $i\mapsto i+4\bmod8$
extends to an automorphism of the Jacobian algebra. The orbit is still a Jacobian algebra, and we may define the `orbit' 2-CY category. 

In \cite{Caorsi:2017bnp} it was shown that the $S$-duality group of $E_6$ MN consists of the semi-direct product of the flavor Weyl group
with the $\Z_2$ outer automorphism of the $E_6$ Lie algebra. Physically we expect that a finite subgroup of the $S$-duality group which acts on the flavor via outer automorphisms should be gaugeable while preserving $\cn=2$ \textsc{susy} as was the case in $SU(2)$ $N_f=4$. Then the subgroup generated by $S$ should be gaugeable, as we found above by the symmetries of the quiver. This argument also shows that no other gauging of $E_6$ MN is possible.

The symmetry of the gauged model is $F_4$ which is the result of folding the $E_6$ Dynkin diagram along its $\Z_2$ automorphism. This is the physically expected result \cite{Argyres7}.

\subsubsection{$\Z_3$ gauging of $SU(2)$ $\cn=2^*$}\label{jjasqwe8878}

This gauging produces $(6,SU(2))$ (entry 9). 

Since $SU(2)$ $\cn=2^*$ is a class $\cs[A_1]$ theory, there is in principle an `explicit' description of its cluster category $\sC$
in terms of ideal triangulations of a torus with a puncture. Gaugings of class $\cs[A_1]$ models will be discussed in Part II of this paper, where the present models will be revisited. Here we work at a more naive level, which does not requires a knowledge of the details of $\sC$.

There is a unique quiver in the mutation class, the Markoff one, with a $\Z_3$
automorphism acting freely on the nodes
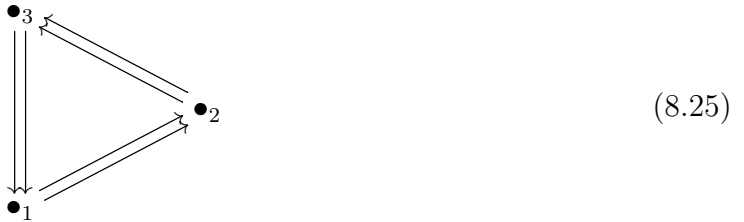
\be\label{markov}
\begin{gathered}
\xymatrix{\bullet_3\ar@<0.4ex>[dd]\ar@<-0.4ex>[dd]\\
&& \bullet_2 \ar@<0.4ex>[ull]\ar@<-0.4ex>[ull]\\
\bullet_1 \ar@<0.4ex>[urr]\ar@<-0.4ex>[urr]}
\end{gathered}
\ee 
Then  the UV category $\sC$ of the ungauged SCFT, $SU(2)$ $\cn=2^*$,
has an autoequivalence $M$ with $M^3\cong\mathrm{Id}$.
Of course this $\Z_3$ group is the obvious finite subgroup of the $S$-duality group $SL(2,\Z)$, and the situation is a close analogue of the $\Z_3$ gauging of $SU(2)$ $N_f=4$ (section 5).
Since the flavor Lie algebra of the parent model,
$\mathfrak{su}(2)$ has no non-trivial outer automorphism, the discrete gauge $\Z_3$ should act trivially on flavor and the gauged model should have the same flavor group as the original model. This is consistent with the action of $\Z_3$ on the Markoff quiver since the flavor charge is represented by the dimension vector $(2,2,2)$ which is invariant under cyclic rotations.

However  now  the ungauged theory has $h>0$,
so there are further subtleties in the categorical approach. These subtleties do not spoil the existence of the gauging, but make the categorical description more involved.
The ``quotient'' category by $M$ is then the UV category for the SCFT in entry 9.

\subsubsection{$\Z_2$ gaugings of $SU(2)$ $\cn=2^*$}

There are two subtly different such gaugings: entries 16 and 17. Both have $F=SU(2)$. 
Thus the $\Z_2$ symmetry to be gauged acts trivially on flavor; this is consistent with $\mathfrak{su}(2)$ not having non-trivial
outer automorphisms.

The Markoff quiver
\eqref{markov} has no $\Z_2$ automorphism. This is not a problem according the discussion in \S.\ref{sub:h1}, since in this case $h=1$.
For $h\geq1$ the only condition is that the  discrete gauge group $\bG$ should be an unobstructed subgroup of the homological  $S$-duality group \cite{Caorsi:2016ebt,Caorsi:2017bnp} equal to the group of auto-equivalences $\mathrm{Aut}(\sC)$ of the UV category $\sC$ of the parent theory (modulo its subgroup acting trivially on physical observables). 

For $SU(2)$ $\cn=2^*$, $\mathrm{Aut}(\sC)/\text{(phy.triv.\!)}$ is a $\Z_2$ extension of the modular group $PSL(2,\Z)$ \cite{Caorsi:2016ebt,Caorsi:2017bnp}. As in the previous examples, a finite subgroup of  $\mathrm{Aut}(\sC)/\text{(phy.triv.\!)}$ should be gaugeable iff it acts on the flavor charges only through outer automorphisms.
This, in particular, applies to the $\Z_2$ subgroup generated by $S\in PSL(2,\Z)$ \cite{SW0}.
In the definition of the $S$-duality group
 $\mathrm{Aut}(\sC)/\text{(phy.triv.\!)}$
the automorphisms of the quiver and the automorphisms of the underlying graph which reverse all arrows enter on the same footing \cite{Caorsi:2017bnp}. The Markoff quiver has a $\Z_2$ morphism of the second kind, corresponding to an elementary mutation at one node, which induces a  freely acting $\Z_2$ automorphism  of the complementary full subquiver of the mutated node. This is the categorical symmetry to be gauged.
 
Let us make it a bit more explicit. The Markoff quiver \eqref{markov} is the Gabriel quiver of the algebra
$\mathrm{End}_\sC(T_1\oplus T_2\oplus T_3)$,
where $T_1\oplus T_2\oplus T_3\in\sC$ is a certain cluster-tilting object and the $T_i$'s are indecomposable. By the Iyama-Yoshino theorem \cite{IH}\footnote{\ See \textsc{appendix}, \textbf{Proposition \ref{IYpro}}.} there are precisely 2 indecomposables $T_3, T_3^*\in \sC$
such that
\be
T_1\oplus T_2\oplus T_3,\quad\text{and}\quad T_1\oplus T_2\oplus T_3^*
\ee  
are cluster-tilting. Mutation at node $\bullet_3$
is the involution corresponding to the interchange
$T_3\leftrightarrow T_3^*$: the quiver of $\mathrm{End}_\sC(T_1\oplus T_2\oplus T_3^*)$ is just the mutation of the quiver \eqref{markov} at the node $\bullet_3$.
The mutated quiver has the same form as the original one up to the permutation $\bullet_1\leftrightarrow \bullet_2$ of nodes. Thus the operation
\be
\sigma\colon \{T_1,T_2,T_3\}\leftrightarrow \{T_2,T_1,T_3^*\}
\ee 
is an involution which
 leaves invariant the quiver $Q$ and its superpotential $\cw$. Therefore it leaves invariant the Ginzburg dg algebra $\Gamma(Q,\cw)$,
  and induces a $\Z_2$ automorphism of the categories one constructs out of it, in particular the  cluster category $\sC$ \cite{Caorsi:2017bnp}. Since $\sigma$ acts trivially on the flavor charge, it should be identified with the unique $\Z_2$ with these properties, that is, with the gaugeable duality $S\in PSL(2,\Z)$.
\medskip

Alternatively, we may see $SU(2)$ $\cn=2^*$ as the class $\cs[A_1]$ theory \cite{classS1,classS2} with Gaiotto surface $C$ the once-punctured torus. In this context, the
$S$-duality group arises as the mapping class group of $C$, which acts on the ideal triangulation in the obvious way. This leads to an explicit action of the modular group $GL(2,\Z)$ on $\sC$ as described in the language of ideal triangulations (for details see  \cite{Caorsi:2017bnp} and references therein).
\medskip

Is the $\Z_2$ gauging of $\cn=2^*$ unique? 

Two inequivalent gaugings will correspond to two  order-2 elements of $\mathrm{Aut}(\sC)/\text{(phy.triv.\!)}$,
$S$ and $S^\prime$. $S(S^\prime)^{-1}$ would be a order-2 autoequivalence acting trivially on both the electro-magnetic and the flavor charges.
Is there such a order-2 auto-equivalence?
To answer this question, we first recall that the class $\cs[A_1]$ theory we are considering actually corresponds to 
$SU(2)$ $\cn=2^*$ plus a free hypermultiplet. Hence the IR physics contains two hypers with zero electric-magnetic charge (cfr.\! \S.4.8 of \cite{Alim:2011kw}). In the IR the automorphism $b$ of the Markoff quiver which fixes all nodes and flips all arrows in the pairs $\rightrightarrows$ interchanges these two hypermultiplets.
Thus using $S$ and $b S$ yields us two subtly different gauged theories both with
$\Delta=4$ and $F=SU(2)$. 
This perfectly matches the subtle difference between the SCFT in entries 16 and 17.

\begin{rem} The rational elliptic surfaces $(\ce,\cf^\vee)$ for entries 16 and 17 are obtained one from the other by a Kodaira quadratic transformation which leaves invariant the fiber at $\cf^\vee$. In particular they have the same functional invariant $\mathscr{J}\!(z)$ and hence the same space of deformations. So they look `almost' the same. Their respective parent ungauged models are the two different versions of $\cn=2^*$.
\end{rem}

\subsubsection{The last one: $\Z_2$
gauging of $(3,Sp(6))$}

It remains only one gauging to discuss:
the one of $(3,Sp(6))$ producing $(6,Sp(4))$. We know no simple description of the cluster category $\sC$ of the parent theory 
besides its general abstract definition in terms of the 
perfect $\mathfrak{Per}\,\Gamma$ and bounded derived $D^b\Gamma$ categories of the Ginzburg dg algebra $\Gamma$ associated to its quiver with superpotential (see \cite{Caorsi:2017bnp} for a review); in fact, since we do not know the higher order terms in the superpotential,
even the general definition is not complete.
Then instead of a precise argument, we shall present circumstantial evidence for the existence (and uniqueness) of the $\Z_2$ gauging from two categorical viewpoints.

\subparagraph{Comparison with $SU(2)$ $\cn=2^*$.} In the mutation class, we have the quiver
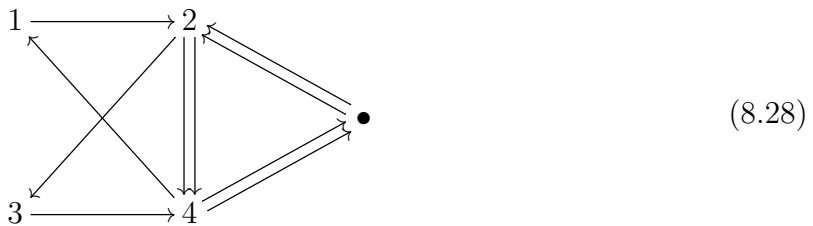
\be\label{markov23}
\begin{gathered}
\xymatrix{1\ar[rr]&&2\ar[ddll]\ar@<0.4ex>[dd]\ar@<-0.4ex>[dd]\\
&&&& \bullet \ar@<0.4ex>[ull]\ar@<-0.4ex>[ull]\\
3\ar[rr]&&4 \ar[uull]\ar@<0.4ex>[urr]\ar@<-0.4ex>[urr]}
\end{gathered}
\ee
which looks like `an extension' of the Markoff one by adding the nodes 1 and 2.
Mutation at the $\bullet$ node (the one associated to $\Z_2$ gauging of $\cn=2^*$) induces a
graph automorphism which is a free
$\Z_2$ automorphism of the complementary subquiver, $i\mapsto i+2\bmod4$ much as in the Markoff case. 
At the level of indecomposable summands of the corresponding
cluster-tilting object we have
\be
\sigma(T_\bullet)=T_\bullet^*,\qquad \sigma(T_j)=T_{j+2}.
\ee
As in the $\cn=2^*$ case this shows that we have a gaugeable $\Z_2$ subgroup of the $S$-duality group.
Since $\mathfrak{sp}(4)$ does not have non-trivial outer automorphisms, the non-Abelian factor in the
original flavor symmetry, $Sp(4)\times U(1)$ should be preserved.  The Abelian part is lost since it is odd under $\sigma$.

Alternatively, we may argue as in the following subsection. 

\subsection{Computer search of consistent gaugings} There exists a purely combinatorial algorithm \cite{Caorsi:2017bnp} to compute the $S$-duality group and its action on the flavor charges for all $\cn=2$ QFT with the BPS-property. This algorithm may be effectively implemented on a computer provided its quiver $Q$ has not too many nodes.
The algorithm may be used to search
for all gaugeable discrete subgroups of
the $S$-duality group. For the models in the previous subsection, the algorithm returns the discrete gaugings we already know. The algorithm may be run for the quiver \eqref{markov23} with the result that the $S$-duality group contains an order 2 element acting on the flavor charges as described above.

\subsection{No more gaugings}

The computer search for $\mathrm{Aut}(\sC)/\text{(phy.triv.\!)}$ shows that there are no gaugings of $E_7$ and $E_8$ Minahan-Nemeshanski. This is already obvious from the dimension formula \eqref{dfffr} since any such gauging will have $\Delta\geq 8$ which is impossible in rank-1 \cite{Caorsi:2018zsq}.

\section{The 5 false-gaugings}

It remains to discuss the 5 base-changes of elliptic surfaces that we dubbed false-gaugings, i.e.\! the rows in table \ref{missing} without the symbol $\checkmark$. 
As we have already observed,
they are precisely the symplectic base-changes of rational elliptic surfaces $\ce$ with fiber configurations satisfying all physical conditions, including the `Dirac quantization of charge', which pull back to elliptic surfaces $\ce^\prime$ which satisfy all requirements except the `Dirac quantization'. Indeed, in all 5 instances the covering surface $\ce^\prime$ has a fiber configuration of the form
\be
\{\cf^\vee; I_{a_1}, I_{a_2}^{m}\}\ \  \text{with } \frac{a_1}{a_2}\in\bQ\ \text{square-free}, \ \  a_1,\, n\mid m,\ \ a_1+m a_2=12-e(F_\infty),
\ee   
$n$ being the degree of the cover.
Thus these models have a perfectly good Seiberg-Witten geometry with a cover which is a nice SW geometry which we do
not consider as a definition of a distinct $\cn=2$ SCFT for tricky  reasons. In table \ref{covvve} we collect some data on the corresponding covers.
In the last column we indicate that two entries, 6 and 8, are under scrutiny by the authors of \cite{Argyres3} since there is no evidence they exist. Here we see that they are precisely the ones with covers having 3-torsion, which is impossible
in rank-1 2-CY. This does not prove that the SCFT $\#6$ and $\#8$ do not exist, but provides further evidence that they are ``strange''.  

\begin{table}
\begin{center}
\begin{tabular}{cccc}\hline\hline
$\#$ & cover fibers & cover Mordell-Weil & existence?\\\hline
5 & $\{IV;I_2,I_1^6\}$ & $A_5^\vee$ &\\
6 & $\{IV;I_2,I_3^2\}$ & $\langle 1/6\rangle\oplus \Z/3\Z$ & in question\\
8 & $\{IV;I_6,I_1^2\}$ & $A_1^\vee\oplus \Z/3\Z$ & in question\\
15 & $\{I_0^*;I_2,I_1^4\}$ & $A_1^\vee\oplus A_1^\vee\oplus A_1^\vee$ & \\
22 & $\{IV^*;I_2,I_1^2\}$ & $\langle 1/6\rangle$ & \\\hline\hline
\end{tabular}
\caption{\label{covvve}Covering rational elliptic surfaces for the 5 rank-1 false-gaugings. $\#$ is model number in table 1 of \cite{Argyres3}. The third column yields the Mordell-Weil group of the covering surfaces (\textbf{notation}: $\mathfrak{g}^\vee$ stands for the weight lattice of the simply-laced, simple Lie algebra $\mathfrak{g}$; $\langle 1/6\rangle$ stands for $\Z$ endowed with the quadratic form $\tfrac{1}{6}x^2$). In the last column we indicate the SCFT whose existence is put in question in ref.\!\!\cite{Argyres3}.}
\end{center}
\end{table}

The simplest interpretation of these covering geometries is to start from a   
generic configuration $\{\cf^\vee;I_1^{12-e(\cf^\vee)}\}$ which corresponds to the quiver $Q(1^{10-e(\cf^\vee)},q)$; its  functional invariant $\cj\!(z)$ is a rational function with $12-e(\cf^\vee)$ simple poles in $\C$ satisfying the Kodaira conditions on zeros and ones
\cite{miranda}. Varying the position of the poles
is equivalent to changing the mass parameters which take value in the Cartan subalgebra of the flavor group $F$.
By suitable fine-tuning of the mass-deformation we may force a group of $b$ poles of $\cj(z)$ to coalesce to form a pole of higher order $b$. An order $b$ pole corresponds to a fiber of type $I_b$. Physically, we fine-tune the masses so that $b$ different hypermultiplet species get massless on the same locus in the Coulomb branch. By such
an operation we may obtain any fiber configuration of type $I_b$ fibers (keeping fixed the fiber at infinity $\cf^\vee$) provided the configuration is present in the tables of allowed fiber configurations 
 \cite{per,mira}.  In this process, the $b$ particle species loose their distinct identity, and this amounts to quotient out the statistic group $\mathfrak{S}_b\equiv\mathrm{Weyl}(A_{b-1})$.
 The visible Weyl group is then the commutant of  $\mathrm{Weyl}(A_{b-1})$ in the Weyl group of the generic surface $\mathrm{Weyl}(E_{10-e(\cf^\vee)})$.  

\subparagraph{Seiberg-Witten viewpoint.}
Let us consider false-gauging from the SW perspective. A BPS particle of the (mass-deformed) SCFT is a (family of) calibrated curve(s) $\gamma$ inside the fiber $\ce_u$ over our Coulomb vacuum $u$, which we assume to be generic. Its pull-back $\phi_\ast^{\ \ast}\gamma$, is a disjoint union of calibrated curves in the covering fibers
$\{\ce^\prime_v\;|\;\phi(v)=u\}$. As in the case of a true gauging, the BPS spectrum becomes a purely local computation at covering vacuum $v$. From this local viewpoint, it is just the BPS spectrum of the generic $E_{10-e(\cf^\vee)}$ theory in a peculiar very fine-tuned limit. Then one expects that, roughly speaking, also their 1d Lagrangian
has the form $\mathscr{L}_\text{fi.tun.}/\bG$ where
$\bG$ is the covering group and $\mathscr{L}_\text{fi.tun.}$ is a fine-tuned version of the 1d Lagrangian for the generic $E_{10-e(\cf^\vee)}$ model. (The fine-tuning may correspond to a singular limit, in which case we need to introduce higher order interactions to regularize it).

At this naive level,
one may model the difference between \emph{true} and \emph{false} gaugings as follow. In false gaugings
we freeze  $b-1$ mass deformations by putting $b$ poles together, and (roughly speaking) this kills
an $A_{b-1}$ sublattice of the flavor lattice.
To produce this effect, the generator $\varrho_\text{false}$ of the false discrete  gauge group $\bG_\text{false}\cong\Z_n$ should be twisted with respect  to the generator  $\varrho_\text{true}$ of a true discrete  gauge group $\bG_\text{true}\cong\Z_n$ by an order-$n$ element $\theta$ of the Weyl group of $E_{10-e(\cf^\vee)}$
\be\label{jjjz0021}
\varrho_\text{false}=\theta\,\varrho_\text{true},\qquad \theta^n=1.
\ee 
Then a natural candidate for the UV category of a false gauging is the twisted-orbit category
\be
\sC_\theta= \Big(\sD_\ca\big/(\theta\varrho)^\Z\Big)_\text{tr.hull}
\ee
which, while still Hom-finite and 2-CY, has no cluster-tilting objects any longer. Correspondingly we don't have the unbranched Galois cover of Jacobian algebra $J\to J_d$ we had in the case of a true gauging. Of course, this should be expected, since that structure precisely realized the physical concept of what a discrete gauging is, which is not the case here, where the SCFT, while covered by SW geometries, are not discrete gauging in any physical sense.

Thus in the false case the most we may hope for is to have non-zero rigid objects in $\sC_\theta$. As already stated, we take the conservative stand that the 2-CY shuld have a non-zero rigid object (even if we haven't a strong argument for this requirement). There are plenty of 2-CY categories without non-trivial rigid objects.

 Without loss of generality, we focus on \emph{maximal} rigid objects. When cluster-tilting objects exist, all maximal rigid objects are cluster-tilting \cite{1004.5475}. 

For rank-1 and $\bG_\text{false}$ non-trivial  we have a simple relation (cfr.\! eqn.\eqref{kkkas1})
\be\label{zzllq9}
\mathrm{rank}\, F\Big|_\text{false-gauged theory}= \#\!\left(\begin{array}{l}\text{indecomposable direct summands}\\ \text{of basic maximally rigid objects}\end{array}\right)\!.
\ee
More precisely, if $T_\text{max}$ is a maximally rigid object, we have
\be
\text{Flavor group root lattice}\cong K_0(\mathsf{add}\,T_\text{max}).
\ee
Since (in rank-1) all false-gaugings have $\Delta\in\bN$, their UV category $\sC_\theta$ is symmetric, and the Weyl-invariant pairing is just
\be
M_{ij}=\dim\sC_\theta(T_i,T_j)+\dim\sC_\theta(T_j,T_i),\qquad T_\text{max}=\bigoplus_i T_i.
\ee

\subsection{Entries 15 and 22}

The elliptic surfaces $\ce$ of entries 15 and 22 of \cite{Argyres3} have covering elliptic surfaces $\ce^\prime$ which may be seen as fine-tuned limits of surfaces associated to SCFT whose UV category is a pretty explicit orbit category $D^b\ch/(S^{-1}\Sigma^2)^\Z$ with $\ch$ hereditary.
Thus the UV categories for these models can be constructed explicitly as orbit categories of the ungauged one $D^b\ch/(S^{-1}\Sigma^2)^\Z$. 

Such a false-gauging orbit category has the general form
$\sC_\varrho$ in eqn.\eqref{conddz}
\be\label{lllaqwe}
\sC_\varrho= D^b\ch\big/(\varrho)^\Z,
\ee
where $\varrho$ is an autoequivalence of $D^b\ch$ satisfying condition 1) of \eqref{conddz} but -- contrary to a true  gauging -- condition 2) is violated. 
More precisely, an orbit category of the form \eqref{lllaqwe} describes a false-gauging iff:\begin{itemize}
\item[1)] $\varrho^n=\tau^{-1}\Sigma$ for some $n>1$;
\item[2)] $\sC_\varrho$ contains \emph{non-zero} maximal rigid objects. A false-gauging is a true-gauging iff the maximal rigid object is cluster-tilting.
\end{itemize}

This is a special case of the problem studied in the appendix of \cite{finite2} and in ref.\!\cite{finite3}.

\subparagraph{Entry 22.}
According to the table of base-changes for elliptic surfaces, this model corresponds to a $\Z_2$ pseudo-gauging of an Argyres-Douglas model of type $D_4$ producing a SCFT with an Abelian flavor
group  of dimension 1, $F=U(1)$. Eqn.\eqref{solrho} gives the list of solutions to condition 1) for $D(D_4)\equiv D^b\mathsf{mod}\,\C D_4$; we conclude that the pseudo-gauging producing the entry 22 SCFT corresponds to the orbit category
\be\label{kkaz1003}
\sC_{\theta\tau^{-2}}\equiv D(D_4)\big/( \theta\tau^{-2})^\Z.
\ee
We already know that this 2-CY category has no cluster-tilting object.
The argument following eqn.\eqref{solrho}
shows that $S_\odot$ (the simple module with support at the central node) is basic maximal rigid for $\sC_{\theta\tau^{-2}}$. The maximal rigid object has a single indecomposable summand, as expected since the flavor group $F\equiv U(1)$ has rank 1. 

This result can also be read from the mathematical literature. Indeed, this is the category (L2) on page 430 of \cite{finite3}
which is a 2-CY category with non-zero maximal rigid objects which are not cluster-tilting. For further details we refer to \cite{finite3}.

As predicted in our general discussion -- cfr.\! eqn.\eqref{jjjz0021} --
the false-gauging differs from the corresponding true-gauging producing the SCFT with $\Delta=3$
and $F=SU(3)$ (entry 21) by a \emph{twisting} of the auto-equivalence $\varrho$ by an order-2 element $\theta$ of the flavor Weyl group of the ungauged model. This is the general pattern we find in examples.

\subparagraph{Entry 15.}
This SCFT is a $\Z_2$ false-gauging of $SU(2)$ with $N_f=4$ producing a SCFT with a rank-2 flavor symmetry, $SU(2)\times SU(2)$ or $SU(2)\times U(1)$.
As in the previous case, we twist the true 
$\Z_2$ gauging by an order-2 element
$\Pi\in\mathfrak{S}_3\subset \mathrm{Weyl}(Spin(8))$, so that we get the 2-CY 
category
\be
\sC_{\Pi TLT}\equiv D^b\mathsf{coh}\,\bX\big/(\Pi TLT)^\Z,
\ee
which satisfies condition 1). To fix conventions, we take $\Pi$ to be the
permutation of the first two special points $1\leftrightarrow 2$ in the weighted projective line $\bX$. The computations in \textsc{appendix \ref{kkkaqwe}} show that the following sheaves
\be\label{kkaqwe1}
T\cong\co\oplus\co(\vec x_4),\qquad
T^\prime\cong\co\oplus \cs_{4,1},
\ee
are examples of basic maximally rigid non cluster-tilting.
They have two direct summands, in agreement with the rank of the flavor group, cfr.\! eqn.\eqref{zzllq9}.
%The subfactor of the orbit category
%$\sC_{\Pi TLT}$ associated with one of the two rigid objects \eqref{kkaqwe1} is then the UV 2-CY category of the entry 15 model. {\tt [check they are mutation equivalent - produce same subfactor]}

Let us see how the maximal rigid objects
work at the level of flavor symmetry.
From eqn.\eqref{zzllq9} we see that the fact that maximal rigid objects $T_\text{max}$ have ``too few'' direct summands to be cluster-tilting reflects the fact that the flavor group $F$ has a smaller rank; this reduction of rank being produced by the freezing of mass-deformations of the cover surface. In an intuitive language: the twist by the Weyl element $\theta$ breaks the flavor group to one of a smaller rank: if the twist models the coalescing $b$ fibers of type $I_1$ to make one of type $I_b$,
the `rank deficit', i.e.\! the number of missing direct summands in the pull-up
$F^\lambda T_\text{max}$ should be $b-1$.
The net effect of $\theta$ is to delete $b-1$ nodes from the Dynkin graph of the flavor group of the associated true-gauging $F_\text{true}$.
 \medskip

For instance, 
 the true gauging (before twisting by $\Pi$)
 associated to entry 15 produces a SCFT with $F=Spin(7)$ 
(i.e.\! entry 14). The twisting reflects the fiber collision $I_1+I_1\leadsto I_2$, so $b-1=1$. After the twisting we should get a flavor symmetry whose Dynkin graph is obtained from the $Spin(7)$ one by deleting a node.
Indeed, we get
$F^\prime=SU(2)\times SU(2)$ which corresponds to the Dynkin graph reduction
\be
\xymatrix{\bullet\ar@{-}[r]& \circ \ar@{=>}[r]& \bullet}\qquad \text{\begin{LARGE}$\leadsto$\end{LARGE}}\qquad \xymatrix{\bullet && \bullet}.
\ee

\subsection{The remaining 3 false-gaugings}

In the remaining 3 cases the 2-CY categories are not very explicit, and we shall 
be less detailed. From a conservative mathematical viewpoint, one may say that we check some necessary conditions for the existence of the categories, rather than proving their existence. In particular, the troublesome entries 6 and 8 pass the text, but the purported categories may well not exist.  

The first two SCFT, entry 5 and 6, may be seen as $\Z_2$ false-gaugings of Minahan-Nemeshanski $E_6$. Note that the parent MN $E_6$ model should be much more fine-tuned in entry 6 than in entry 5. In entry 5 only two of the eight poles of $\mathscr{J}\!(z)$ are made to collide (and the other six being set in a $\Z_2$ symmetric position); instead in entry 6 the remaining six poles coalesce in two triplets. The fine-tuning reduces the number of free parameter, hence the rank of the flavor group.
The third model, entry 8, may be either interpreted as a pseudo-gauging of $E_6$ MN or of
the $Sp(6)$, $\Delta=3$ SCFT.
Note that entry 6 should be seen as a variant of entry 8 and not as an independent SCFT. 
\medskip

As in the previous two model, these three gaugings should correspond to twisting the genuine gauging producing the $\Delta=6$ theory with $F=F_4$ by an order-2 element $\theta$ of the flavor Weyl group $\mathrm{Weyl}(E_6)$ of the parent theory. Since the $\Z_2$-symmetry to be gauged should be a subgroup of the homological $S$-duality $\bS$ acting on the flavor by outer automorphisms, and given that in this case
\be
\bS=\Z_2\ltimes \mathrm{Weyl}(E_6),
\ee
with $\Z_2$ the outer automorphism of the $E_6$ graph, we conclude:
\begin{quote}\it The allowed gaugings and pseudo-gauging of $E_6$ MN should correspond to the conjugacy classes of non-trivial involutions  
\be\label{ooop}
\theta\in\Z_2\ltimes \mathrm{Weyl}(E_6)
\ee 
such that $\theta\not\in \mathrm{Weyl}(E_6)$.
Let $\theta$ be such an involution. 
The centralizer $C_\theta$ of $\theta$ in $ \mathrm{Weyl}(E_6)$ is expected to have the form
\be\label{ecxcvg}
C_\theta=\prod_j \mathfrak{S}_{b_j} \times \mathrm{Weyl}(F).
\ee
\end{quote}
In eqn.\eqref{ecxcvg} the product of symmetric groups
corresponds to the indistinguishability of the fibers $I_1$ of the covering surface $\ce^\prime$ after being coalesced into groups 
of $b_j$ to form Kodaira fibers of type $I_{b_j}$ (fine-tuning). In table \ref{missing2} we list the 5 SCFT whose elliptic surfaces have symplectic double covers $\ce^\prime$ with fiber configurations of the form $\{IV;I_{a_1}, I^m_{a_2}\}$. In the last column we write the predicted centralizer $C_\theta$ according to \eqref{ecxcvg}.

\begin{table}
\begin{center}
\begin{tabular}{cccclccl}\hline\hline
$\#$ & $\Delta$ & $F$ & fibers $\ce$ & fibers $\ce^\prime$ & Galois & true? & $C_\theta$\\\hline
4 & 6 & $F_4$ & $\{II;I_0^*,I_1^4\}$ & $\{IV;I_1^8\}$ & $\Z_2$ & $\checkmark$ &
$\mathrm{Weyl}(F_4)$\\
5 & 6 & $Sp(6)$ & $\{II;I_1^*,I_1^3\}$ & $\{IV; I_2, I_1^6\}$ & $\Z_2$ && $\mathfrak{S}_2\times \mathrm{Weyl}(C_3)$\\
6 & 6 & $SU(2)$ & $\{II;I_1^*,I_3\}$ & $\{IV; I_2,I_3^2\}$ & $\Z_2$&& $(\mathfrak{S}_2\ltimes \mathfrak{S}_3^2)\times\mathrm{Weyl}(A_1)$\\
7 & 6 & $Sp(4)$ & $\{II;I_2^*,I_1^2\}$ & $\{IV; I_4, I_1^4\}$& $\Z_2$ & $\checkmark$
& $\mathfrak{S}_4\times \mathrm{Weyl}(C_2)$\\
8 & 6 & $SU(2)$ & $\{II;I_3^*,I_1\}$ & $\{IV; I_6, I_1^2\}$& $\Z_2$&&
$\mathfrak{S}_6\times \mathrm{Weyl}(A_1)$\\\hline\hline
\end{tabular}
\caption{\label{missing2} The five symplectic coverings with covering surface $\ce^\prime$ which is a fine-tuned limit of the $\Delta=3$ $F=E_6$ Minahan-Nemeshanski. In the last column $C_\theta$ is the expected centralizer of the pseudo-gauged $\Z_2$. A $\checkmark$ means that the false-gauging is actually a gauging.}
\end{center}
\end{table}

Using \textsc{mathematica} we compute the conjugacy classes of involutions $\theta\in\Z_2\ltimes \mathrm{Weyl}(E_6)$
with the required properties and their centralizers $C_\theta$. We found 4
such conjugacy classes with $C_\theta$:
\be
\mathrm{Weyl}(F_4),\quad \mathfrak{S}_2\times \mathrm{Weyl}(C_3),\quad
\mathfrak{S}_4\times \mathrm{Weyl}(C_2),
\quad
\mathfrak{S}_6\times \mathrm{Weyl}(A_1) 
\ee
in perfect agreement with the list of SCFT if one sees 6 and 8 as different realizations of same SCFT (in analogy with the two realizations of $\cn=2^*$). As already mentioned, these two models belong to the ugly class, and may not exist as QFT.

The flavor groups of the SCFT arising from pseudo-gaugings work as expected for a subfactor. For instance, for 
entry 5 we have to delete a single node from the Dynkin graph  
\be
\xymatrix{\circ\ar@{-}[r]&\bullet \ar@{=>}[r]& \bullet\ar@{-}[r]&\bullet}
\qquad \text{\begin{LARGE}$\leadsto$\end{LARGE}}\qquad
\xymatrix{\bullet \ar@{=>}[r]& \bullet\ar@{-}[r]&\bullet}.
\ee

\subsection{No more false-gaugings}

Since $E_7$ and $E_8$ have no gaugings, we have nothing to twist by suitable elements of the flavor Weyl group. Indeed, the $E_7$ and $E_8$ Lie algebras have no outer-automorphisms. Thus we conclude that there are no other false-gaugings besides those we have already listed.    Of course the absence of such (false)-gaugings already follows from the bound $\Delta\leq 6$ on the dimension of a rank-1 SCFT.

\section{Conclusions}

We have listed all 2-CY categories satisfying our \textbf{Criterion/Definition} on page \pageref{cride} and constructed quite explicitly most of them. We have verified that they are in one-to-one correspondence with the entries of table 1 of ref.\!\cite{Argyres3} (plus the four categories associated with asymptotically-free rank-1 QFTs). The correspondence is quite precise; when an entry is \emph{ugly} from the point of view advocated by the authors of \cite{Argyres1,Argyres2,Argyres3,Argyres4,Argyres5,Argyres6,Argyres7} looks also ugly from the categorical perspective, and for a parallel reason.
We have also checked the `functorial' correspondence with the classification of rational elliptic surfaces advocated in 
\cite{Caorsi:2018ahl}. 

%Thus we have  succeeded  in our goal of reproducing the (known) classification. This shows that -- as expected -- the RT methods are indeed very effective in classifying four-dimensional $\cn=2$ QFT. Then:
%\begin{quote}\it
%The way now is open to classification in higher rank $k$.
%\end{quote}
%in the process we described many fine points of the physics of the rank-1 models, computing many of their invariant.
%These results have many interesting implications that we shall describe in separate works \cite{toappear}.
%\medskip
%
%There is, however, one aspect which makes this work a bit unsatisfactory: we have shown the equivalence (say) of the classification of 2-CY categories with the right properties with the classification of rational elliptic surfaces with the appropriate properties by performing the two classifications separately and then comparing them. It would be natural to look for a conceptual proof of the correspondence. We think such a deeper understanding would help a lot in simplifying the analysis in higher rank. 

In a companion paper we put the RT methods at work for the classification of $\cn=2$ SCFT in higher rank $k$.

\section*{Acknowledments}

We thank Michele Del Zotto for sharing with us his deep insights which helped so much to lay the foundations of the RT approach to $\cn=2$ classification. We are grateful to Bernhard Keller, Dirk Kussin and Mario Martone for interesting discussions and clarifications about their respective works.

\appendix

\part*{Technical appendices}

\section{`Reading between the lines' vs.\! Mordell-Weil torsion}
\label{mwtorsion}

\subparagraph{Line-operators and 't Hooft groups.} We recall that the line-operator \emph{classes} in a 4d Lagrangian QFT with semi-simple gauge group $G$ are labelled by elements of the Abelian group $\boldsymbol{Z}\oplus \boldsymbol{Z}^\vee$,
where $\boldsymbol{Z}$ is the center of the universal cover $\widetilde{G}$ of $G=\widetilde{G}/H$, $H\subset \boldsymbol{Z}$ (reading between the 4d  lines of gauge theories \cite{betweenthelines}). For simplicity we assume $G$ to be simple.
In this case $\boldsymbol{Z}$ is $k$-torsion, and we have a canonical skew-symmetric pairing
\be
\wedge^2(\boldsymbol{Z}\oplus \boldsymbol{Z}^\vee)\to \Z_k
\ee
which we call the Weil pairing. Two lines are mutually local iff their pairing vanishes.
If $k$ is a prime, $\boldsymbol{Z}\oplus \boldsymbol{Z}^\vee$ is a vector space $V$ over the field $\bF_k$ with $k$ elements, and the Weil pairing is a non-degenerate symplectic form. A maximal set of mutually local lines is then a Lagrangian subspace of $V$. Let $G=\widetilde{G}/H$
and $\boldsymbol{K}=\boldsymbol{Z}/H$ so that $\boldsymbol{K}$ is Abelian group of non-trivial gauge transformations.
Gauge invariance requires the set of allowed classes to be invariant under
$z\mapsto z+(\boldsymbol{K}\oplus 0)$ and locality 
 $\langle z, (\boldsymbol{K}\oplus 0)\rangle=0$.
 Therefore, we may define the line operator classes as elements of 
 \be\label{thoft}
 \boldsymbol{Z}/\boldsymbol{K}\oplus
 (\boldsymbol{Z}/\boldsymbol{K})^\vee\cong
 \boldsymbol{H}\oplus \boldsymbol{H}^\vee.
 \ee
 The group $\boldsymbol{H}$ should act trivially on the local fields, hence $\boldsymbol{H}\subset \boldsymbol{I}$, with $\boldsymbol{I}\subset \boldsymbol{Z}$ the isotropy group of the local fields in $\widetilde{G}$.
 We call the group in \eqref{thoft} the 't Hooft group: he introduced it to classify the phases of 4d non-Abelian gauge theory \cite{tHooft:1977nqb,tHooft:1979rtg}. 
In a Lagrangian theory this group is a
subgroup of $\boldsymbol{I}\oplus\boldsymbol{I}^\vee$. 
The group $\boldsymbol{I}\oplus\boldsymbol{I}^\vee$ is the largest 't Hooft group consistent with a given local Lagrangian. It is the group of line operators which may be defined for the given Lagrangian. The maximal set of line operator which we may define is a isotropic subspace under the Weil pairing $\langle-,-\rangle$; mathematically, we may always work on a `finite cover of our QFT' (not a QFT\,!!)
where the correlation functions are multivalued in which the group is the full $\boldsymbol{I}\oplus\boldsymbol{I}^\vee$; this is what 't Hooft did. 

The $S$-duality group acts on the 't Hooft group by electro-magnetic duality. In a theory obtained by gauging a discrete subgroup of $S$-dualities of a Lagrangian theory, the 't Hooft group is the appropriate quotient of the 't Hooft group of the parent theory.

\subparagraph{Torsion in the Grothendieck group $K_0(\sC)$ of a cluster category.}
If the $\cn=2$ theory is Lagrangian (and hence BPS-quiver \cite{Alim:2011kw}) we alway get \cite{Caorsi:2017bnp}
\be
K_0(\sC)=\text{(flavor weight lattice)}\oplus \boldsymbol{I}\oplus \boldsymbol{I}^\vee,
\ee 
so that we get the maximal 't Hooft group.
In other words, the cluster category always contains all possible line operators. 
We have a natural skew-symmetric form
$\wedge^2(\boldsymbol{I}\oplus\boldsymbol{I}^\vee)$ which in rank-1
takes values in $\Z_2$. It is identified with the 't Hooft form of the Lagrangian QFT.
E.g.\! in pure $SU(2)$
the UV category is the cluster category of coherent sheaves on $\bP^1$, $\sC=\sC(\bP^1)$ and $K_0(\sC)$ is
the group $\Z_2^{\oplus2}$ generated by the class $[\co]$ of the structure sheaf together with the class $[\cs_0]$ of the skyscraper at 0, subjected to the relations
$2[\co]=2[\cs_0]=0$. The pairing is just the Euler pairing in $\mathsf{coh}\,\bP^1$ taken mod 2
\be
\chi(\co,\cs_0)=1,\qquad \chi(\cs_0,\co)=-1.
\ee

\subparagraph{Mordell-Weil torsion in rational elliptic surfaces.}
On the basis of the correspondence suggested in the introduction, one also expects that the torsion part of the Mordell-Weil group $\mathsf{MW}(\ce)$ matches the torsion part of the Grothendieck group  $K_0(\sC)$.
Comparing table 1 of ref.\!\!\cite{Argyres3}
with the table of Mordell-Weil torsion \cite{bookMW}, we see that of the 28 fiber  configurations corresponding to SCFT only 4 have non trivial $\mathsf{MW}$ torsion, namely
\be
\begin{tabular}{c|cccc}\hline\hline
$\#$ & 24 & 25 & 16 & 17\\
fiber conf. & $\{I_0^*;I_2^3\}$ & $\{I_0^*;I_4,I_1^2\}$ & $\{III;I_2,I_1^*\}$ & $\{III;I_1,I_2^*\}$\\
$\mathsf{MW}$ torsion & $(\Z/2\Z)^{\oplus 2}$ & $\Z/2\Z$ &
$\Z/2\Z$ & $\Z/2\Z$\\\hline\hline 
\end{tabular}
\ee
If we add to the list the asymptotically-free models not covered in \cite{Argyres3} we get two additional fiber configurations with non-trivial $\mathsf{MW}$ torsion groups
\be\label{xxxx10m}
\{I_2^*;I_2^2\}:\ (\Z/2\Z)^{\oplus2},\qquad\{I_4^*;I_1^2\}:\ \Z/2\Z.
\ee
This fact has a simple interpretation. 

\subparagraph{Entry 24:} this SCFT is $SU(2)$ $\cn=2^*$ whose 't Hooft group is
$\Z_2^{\oplus 2}$  (for the gauge group $SU(2)/\Z_2\equiv PSU(2)$); the torsion part of the Grothedieck group of the cluster category is also $\Z_2^{\oplus2}$, see eqn.\eqref{llllc23}. Therefore, in this case, we have a perfect match of the three torsion groups: 't Hooft, Grothendieck, and Mordell-Weil.

Below we see $\Z_2^{\oplus2}$ as the vector space $\bF_2^2$ over the field with 2 elements $\bF_2$.
't Hooft electric/magnetic fluxes are elements of this space $(e,m)\in\bF^2$.

\subparagraph{Entries 16 and 17:} these are the two subtly different $\Z_2$ discrete gaugings of $SU(2)$ $\cn=2^*$. Gauging a subgroup of the $S$-duality group identifies 't Hooft's electric and magnetic fluxes, i.e.
\be
\begin{pmatrix}e \\ m\end{pmatrix} \sim
\begin{pmatrix} 0 & -1\\ 1 & 0\end{pmatrix}
\begin{pmatrix} e\\ m\end{pmatrix},\qquad \begin{pmatrix} e\\ m\end{pmatrix}\in\bF^2_2,
\ee
so that $e=m\in \bF_2$ and the 't Hooft group reduces to a single copy of $\Z_2$,
in full agreement with the Mordell-Weil torsion. Again we have full agreement of the three torsion groups.

Note that the same argument applied to entries 9 and 10, i.e.\! $\Z_3$ discrete gaugings of $\cn=2^*$ yields
\be
\begin{pmatrix}e \\ m\end{pmatrix} =
\begin{pmatrix} 1 & -1\\ 1 & 0\end{pmatrix}
\begin{pmatrix} e\\ m\end{pmatrix}\quad\Rightarrow\quad 0=e=m\in \bF_2,
\ee
and no Mordell-Weil torsion is expected, again in full agreement with the tables of \cite{bookMW}.

\subparagraph{Entry 25.} 
This is the second version of the SW geometry of $SU(2)$ $\cn=2^*$. The Mordell-Weil torsion $\Z/2\Z$ is a subgroup of the maximal 't Hooft group; the natural interpretation is that in this geometry one may realize only some line-operator.
$\Z/2\Z$ is the group of a maximal set of mutually-local line-operators, so one gets all the lines present in the QFT (as contrasted with a multivalued cover).  

\subparagraph{The $a\!f$ model $\{I_2^*;I_2^2\}$.} At its face value this may seem to be a zero-parameter specialization of $SU(2)$ with $N_f=2$, in fact it is pure $SU(2)$ which is double-covered by $SU(2)$ with $N_f=2$ \cite{galoiscover}. This can be seen in three different ways. The two SW curves are \eqref{ppppq12234}
\be
p^2=e^x-e^{-x},\qquad p^2=e^{2x}-e^{-2x},
\ee
and $2x\to x$ transforms one into the other.
At the quiver level: we have the unramified Galois $\Z_2$ cover \cite{Cecotti:2015qha}
\be
Q\equiv\begin{gathered}
\xymatrix{1\ar[r]\ar[d] & 2\\
4 &3\ar[u]\ar[l]}\end{gathered}\quad 
\xrightarrow{\quad Q\big/(i\,\sim\, i+2\bmod2)\quad }\quad \xymatrix{1\equiv 3\ar@/^0.8pc/@<0.1ex>[rr]\ar@/_0.8pc/@<-0.1ex>[rr] && 2\equiv4}.
\ee
In terms of SW geometries: the `usual' form for pure $SU(2)$, corresponding to fiber configuration $\{I_4^*,I_1^2\}$, is the realization of the pure $SU(2)$ monodromy on the twice-punctured Coulomb branch in the form
\be
\overset{\rm dyon}{(L)}\cdot\overset{\rm monopole}{(T^2LT^{-2})}\cdot\overset{\beta\ \text{at }\infty}{(-T^4)}=1,
\ee
where the two factors are the monodromies in $SL(2,\Z)$ (with $T,L$ the matrices in eqn.\eqref{modtras}) of, respectively, the 
massless dyon and the massless monopoles, and the last factor the \emph{inverse} of the monodromy at $\infty$ (weak coupling). Inserting in the above product the identity $-1=-L\cdot L^{-1}$ and suitable parenthesis we get
\be
\overset{I_2^*}{(-L^2)}\cdot \overset{I_2}{(L^{-1}T^2L)}\cdot \overset{I_2}{(T^2)}=1
\ee
where to each element of $SL(2,\Z)$ we associate the Kodaira fiber of its conjugacy class. Thus this is the monodromy relation for the surface we are interested in $\{I_2^*;I_2^2\}$. Now conjugate each matrix
\be
X\mapsto \begin{pmatrix} 1&0\\0&2\end{pmatrix}X\begin{pmatrix} 1&0\\0&2\end{pmatrix}^{\!\!-1}\quad\Rightarrow\qquad T^2\mapsto T,\quad
L\mapsto L^2
\ee
which, in the language of \cite{Argyres3}, is the operation of changing the charge normalization. One gets
\be
\overset{\beta\ \text{at }\infty}{(-L^4)}\cdot \overset{\rm monopole}{(L^{-2}TL^2)}\cdot \overset{\rm dyon}{(T)}=1
\ee
which is the standard presentation of pure $SU(2)$ in the Coulomb branch (by the braid relation \eqref{braid}, $L$, $T$ belong to the same conjugacy class).
Thus the SW geometry $\{I_2^*,I_2^2\}$ is just pure $SU(2)$ with electric charge in an unsual normalization. From \eqref{xxxx10m} we see that in this case the 't Hooft, Grothendieck, and Mordell-Weil torsion groups all agree.
In the usual realization $\{I_4^*,I_1^2\}$ only a local subset of line operators are explicitly realized.

\section{Categories:
notations, definitions, basic facts}\label{adetails}
\subsection{Basic definitions and theorems}

The set of morphisms $X\to Y$ in a category $\sC$ is written $\sC(X,Y)$;
the notation $\mathrm{Hom}(X,Y)$ is reserved to the special case when $\sC$ is Abelian hereditary (see definition below). We write $\mathrm{End}_\sC(X)$ for $\sC(X,X)$ and $\mathrm{End}(X)$ for $\mathrm{Hom}(X,X)$. 
\medskip 

We recall that a category $\sC$ is \emph{$\C$-linear} iff,  for all objects $X,Y\in\sC$\!, the set $\sC(X,Y)$ is a $\C$-space, composition is bilinear, and finite direct sums exist. $\sC$ is 
\emph{Hom-finite} if in addition $\dim \sC(X,Y)<\infty$; in this case $\mathrm{End}_\sC(X)$ is a finite-dimensional associative $\C$-algebra with 1 for all $X$.
$\sC$ is \emph{Krull-Schmidt} if it has split idempotents. Through the paper $D?\equiv \mathrm{Hom}_\C(?,\C)$ stands for duality of $\C$-spaces.

\textbf{Convention:} In this paper we use the term \emph{`category'}  as a synonym for \textit{`\,$\C$-linear, Hom-finite, Krull-Schmidt category'.} 
We may (and do) assume the category to be \emph{connected}.
All categories we use have these properties.
Under these assumptions, all objects in $\sC$ are finite direct sums of indecomposable objects, and the endo-algebra $\mathrm{End}_\sC(X)$ of all indecomposable $X$ is local.

We write $\mathsf{add}\,X$ for \textit{the additive closure of $X$ in $\sC$} namely the full subcategory of $\sC$ whose objects are direct sums of summands of $X$.
An object $X\in\sC$ is said to be \emph{basic} if its indecomposable direct summands are pairwise non-isomorphic; this is the same as saying that its endo-algebra $\mathrm{End}_\sC(X)$ is basic.
Now (see e.g.\! \cite{galoiscover})

\begin{thm}[Gabriel]\label{thmgab} Let $A$ be a finite-dimensional basic $\C$-algebra.
Then there exists a quiver $Q_A$, unique up to isomorphism, such that $A$ is isomorphic to $\C Q_A/I$, where $I$ is an ideal of the path algebra $\C Q_A$ contained in the square of the ideal generated by all arrows of $Q_A$.
There is a bijection $i\mapsto S_i$ between the nodes of $Q_A$ and the iso-classes of simple (right) $A$-modules. The number of arrows between nodes $i$ and $j$ is $\dim\mathrm{Ext}^1(S_j,S_i)$.
\end{thm}

Given a basic object $X\in\sC$ its endo-quiver is the quiver of its endo-algebra,
$Q_{\mathrm{End}_\sC(X)}$.

\subparagraph{Abelian and hereditary categories.} The category $\sC$ is \emph{Abelian} if all morphisms have kernels and cokernels. For an Abelian category we define the spaces $\mathrm{Ext}^k(X,Y)$, $k\in\Z$.
$\mathrm{Ext}^0(X,Y)\equiv\sC(X,Y)$, while
$\mathrm{Ext}^k(X,Y)=0$ for $k<0$.
An object $P$ (resp.\! $I$) is said to be
\emph{projective} (resp.\! \textit{injective})
iff $\mathrm{Ext}^1(P,X)=0$ (resp.\! $\mathrm{Ext}^1(X,I)=0$) for all $X$.

An Abelian category $\sC$ is said to be
\emph{hereditary} iff $\mathrm{Ext}^k(X,Y)=0$ for all $k>1$ and $X,Y\in\sC$. A finite-dimensional algebra is said to be hereditary iff the Abelian category of its finite-dimensional modules is hereditary. For a basic algebra as in \textbf{Theorem \ref{thmgab}} this happens iff the ideal $I\equiv0$.

An object in a hereditary category is said to be \emph{rigid} if $\mathrm{Ext}^1(X,X)=0$ and \textit{maximally rigid} if there is no $Y\not\in\mathsf{add}\,X$ so that
$X\oplus Y$ is rigid. A \emph{tilting object} $T$ in a hereditary category $\ch$ is a basic rigid object such that $\mathrm{Ext}^k(T,X)=0$ for all integers $k\geq 0$ implies $X=0$.
Let $T=\sum_i T_i\in\ch$ be tilting, and $A=\mathrm{End}_\ch(T)$ its endo-algebra.  The indecomposable summands $T_i$ are identified with the indecomposable projectives $P_i$ of the module category
$\mathsf{mod}\,A$ (and thus are in bijection with the nodes of the quiver $Q_A$). One has the derived equivalence
\be
D^b\ch\cong D^b\mathsf{mod}\,A.
\ee

A hereditary category $\ch$ has \textit{Serre duality} if there is an autoequivalence $\tau$ such that
\be\label{serre}
\mathrm{Ext}^1(X,Y)= D\mathrm{Hom}(Y, \tau X).
\ee
Since $\tau$ is an autoequivalence, a category with Serre duality has no non-zero projective and injective. Eqn.\eqref{serre}
is valid also in the module category of a hereditary algebra, with $\tau$ the AR translation, which however is not defined on the projectives, while $\tau^{-1}$ is not defined on the injectives.

\begin{thm}[Happel \cite{lenzing1}] A hereditary category with a tilting object is either the module category of a hereditary algebra,  $\C Q$, or the coherent sheaves $\mathsf{coh}\,\bX$ over a weighted projective line $\bX$
(see {\rm \S.\ref{Acoh}.}).
\end{thm}

\subparagraph{Triangle categories.} 
A category $\sC$ is \emph{triangulated} if it has a suspension autequivalence $\Sigma$ (the `shift') and a set of distinguished triangles
\be\label{tria1}
A\to B\to C\to \Sigma A
\ee
satisfying a certain set of axioms, see e.g.\! ref.\!\cite{triangulated}.
In particular, if \eqref{tria1} is a distinguished triangle so is its rotation
$B\to C\to \Sigma A\to \Sigma B$. We also write $X[n]$ for $\Sigma^n X$. 
A typical example of triangulated category is the bounded derived category $D^b(\ca)$
of an Abelian category $\ca$ \cite{triangulated}.
If $\ca$ is in addition hereditary,
$D^b(\ca)$ is the \emph{repetitive category} of $\ca$ \cite{lenzing1} i.e.\! the category $\sD$ whose indecomposables have the form $X[n]$,
$X\in\ca$, $n\in\Z$ and morphism spaces
$\sD(X[i],Y[j])\cong \mathrm{Ext}^{j-i}(X,Y)$. The AR translation $\tau$ extends to an autoequivalence of $\sD$ so that
\be\label{llxcw}
\sD(X,Y[1])\cong D \sD(Y, \tau X).
\ee

\begin{defn} A triangle category $\sC$ is \textit{$n$-Calabi-Yau} ($n$-CY) if there is bifunctorial isomorphisms \cite{CYKeller}
\be
D\sC(X,Y)\cong \sC(Y,X[n]).
\ee
\end{defn}

We are mainly interested in the case $n=2$.
In this case we have the symmetry $\sC(X,Y[1])\cong D\sC(Y,X[1])$.

\begin{defn} Let $\sC$ be 2-CY. An object $X\in\sC$ is \emph{rigid} if $\sC(X,X[1])=0$, and \textit{maximally rigid} iff it is rigid and $X\oplus Y$ rigid implies $Y\in\mathsf{add}\,X$. An object
$T\in \sC$ is \textit{cluster-tilting} iff it is basic, rigid $\sC(T,T[1])=0$, and
\be
\mathsf{add}\,T=\big\{X\in\sC\;\big|\; \sC(X,T[1])=0\big\}.
\ee 
\end{defn}

A cluster-tilting object is maximally symmetric, but the converse is often false see \cite{bikr} for counter-examples.
There are 2-CY categories without cluster-tilting objects, and even 2-CY categories without any non-zero rigid object. Such categories are also relevant for our purposes.
\medskip

We list some basic results:

\begin{pro}[BIRS \cite{BIRS}] $\sC$ a triangulated 2-CY category. \begin{itemize}
\item[(a)] Let $T$ be a cluster-tilting object. Then for all $X\in\sC$ there exist triangles
$T_1\to T_0\to X$ and $X\to T^\prime_0\to T^\prime_1$ with $T_i$, $T^\prime_i$ in $\mathsf{add}\,T$;
\item[(b)] Let $T$ maximal rigid. Then for all \emph{rigid} $X\in\sC$ (a) holds. 
\end{itemize}
\end{pro}

\begin{pro}[Iyama-Yoshino \cite{IH}]\label{IYpro} Let the 2-CY category $\sC$ have cluster-tilting objects.
1) all such objects have the same number $n$ of indecomposable summands. 2) all maximal rigid objects are cluster-tilting. 3) let $\hat T\equiv \bigoplus_{i=1}^{n-1}T_i$ to be an almost-maximal rigid object.
Then there exist precisely \emph{two} indecomposables, $T_n$ and $T_n^\prime$
such that the objects
\be\label{kkkaqw12}
T=\hat T\oplus T_n,\qquad T^\prime=\hat T\oplus T_n^\prime,
\ee 
are cluster-tilting. Replacing the object $T$ by the object $T^\prime$ is called \emph{mutation at $T_n$.} 
\end{pro}

\begin{thm}[Keller-Reiten \cite{KR}] Let $T$ be a cluster-tilting object of the 2-CY category $\sC$, and $A\equiv \mathrm{End}_\sC(T)$ its endo-algebra. The functor
\be
F_T\colon \sC\to\mathsf{mod}\,A,\qquad X\mapsto \sC(T,X)
\ee
in an equivalence of categories
\be
\sC/(\mathsf{add}\,T[1]) \cong \mathsf{mod}\,A,
\ee
where $(\mathsf{add}\,T[1])$ is the ideal of morphisms factoring through objects in 
$\mathsf{add}\,T[1]$.
\end{thm}

If $Q_A$ is finite, we say that the pair $(\sC,T)$ is a \textit{2-CY realization of the quiver $Q_A$.} Note that the nodes of $Q_A$ are in bijection with the summands of $T$. If the quiver $Q_A$ has
no loop at the $n$-th node, it is related to the quiver of $Q_{A^\prime}$ of the mutated cluster-tilting object $T\leadsto T^\prime$ by the \textit{elementary quiver mutation at node $n$} in the sense of Fomin-Zelevinski \cite{quivermutation}. 

On the modules of the algebra
$A\equiv Q_A/I_A$ we have an
important antisymmetric form:
\begin{pro}[Palu \cite{palu}] Let $\sC$ be 2-CY with cluster-tilting object $T$ and $A$ its endo-algebra.
The antisymmetric form
\be\label{Dform}
\begin{split}
\langle X, Y\rangle_D= &\dim \mathrm{Hom}_{\mathsf{mod}A}(X,Y)-\dim \mathrm{Ext}^1_{\mathsf{mod}A}(X,Y)-\\
&-\dim \mathrm{Hom}_{\mathsf{mod}A}(Y,X)+
\mathrm{Ext}^1_{\mathsf{mod}A}(Y,X)
\end{split}
\ee
is well-defined on the Grothendieck group $K_0(\mathsf{mod}A)$.
\end{pro}

For the physical applications a crucial fact is the \textit{Calabi-Yau reduction:}
\begin{thm}[Iyama-Yoshino \cite{IH}] Let $D$ be a  non-zero rigid object in a 2-CY category $\sC$, and consider the full subcategory
\be
{}^\perp D[1]\overset{\text{def}}=\big\{X\in\sC\;\big|\; \sC(X,D[1])=0\big\}.
\ee 
The factor category
\be
\sU_D= {}^\perp D[1]/\mathsf{add}\,D
\ee
 is triangulated 2-CY . The cluster-tilting objects in $\sU_D$ are in one-one correspondence with the cluster-tilting objects of $\sC$ which have $D$ as a summand.
\end{thm}

In particular, if $D$ is an incomplete
cluster-tilting object, i.e.\! an object such that $T=D\oplus T_1\oplus \cdots\oplus T_j$ is cluster-tilting, and $Q_T$ is the endo-quiver of $T$, the quiver of $\mathrm{End}_{\sU_D}(D)$ is the quiver obtained from $Q_T$ by deleting the nodes corresponding to the summands, $T_1,\cdots, T_j$.

\subparagraph{The 2-acyclic case.}
The nicer situation is when our 2-CY category $\sC$ has a cluster-tilting object $T$ such that $Q_A$ is \textit{2-acyclic}
i.e.\! has no loops $\circlearrowright$ (arrows starting and ending at the same node) nor 2-cycles (opposite pairs of arrow $\leftrightarrows$ between the same two nodes). 
In this case there exists a superpotential $\cw_A$ so that $I_A=(\partial \cw_A)$, that is, the endo-algebra $A$ is the Jacobian algebra of a quiver with superpotential $(Q_A,\cw_A)$.

If $Q_A$ has no loops all maximal rigid objects are automatically cluster-tilting 
\cite{maxrigid}.

In the 2-acyclic case the mutation of summands in the sense of \eqref{kkkaqw12} coincides with the mutation of quivers with superpotential
\cite{quivermutation}. The matrix of the antisymmetric form \eqref{kkkaqw12} in the basis of simples, $B_{ij}=\langle S_i,S_j\rangle_D$, is the exchange matrix of the quiver $Q_A$.

In the 2-acyclic case $\sC$ coincides with the \emph{cluster category} $\sC(Q_A)$
which categorifies the cluster algebra associated with the (mutation class) of the 2-acyclic quiver $Q_A$.

When our $\cn=2$ QFT possesses the BPS-quiver property, it is described by a mutation-class of quivers with superpotential (and a stability function)
as described in ref.\cite{Alim:2011kw} and summarized in section 2. In this special  case the Palu antisymmetric form $\langle-,-\rangle_D$ is the Dirac pairing, and the number $n$ of summands of a cluster-tilting object $T$ is the rank of the flavor group plus twice the dimension of the Coulomb branch. As we stressed in the main body of the paper, this is \emph{not true} if $\sC$ has no cluster-tilting object whose endo-quiver is 2-acyclic.

In general, the cluster category $\sC(Q_A)$ of a 2-acyclic quiver with superpotential $\cw_A$ is constructed with the help of the Ginzburg dg algebra $\Gamma(Q_A,\cw_A)$ of the pair $(Q_A,\cw_A)$, see \cite{keller}. In this case the bounded derived category $D^b\Gamma(Q_A,\cw_A)$ may be identified with the IR category describing the BPS particles\cite{Caorsi:2017bnp}; consequently 
$D^b\Gamma(Q_A,\cw_A)$ is 3-CY.

\paragraph{Cluster characters for 2-CY categories with cluster-tilting.}
The cluster characters for general (Hom-finite) 2-CY categories with cluster-tilting object $T$ are defined by Palu in  \cite{palu}; we refer to the original paper for details.

\subsubsection{The cluster category of a hereditary category} An important special case of 2-CY categories with cluster-tilting objects having 2-acyclic endo-quivers is the 
\textit{cluster category} $\sC_\ch$ of a hereditary category $\ch$ with tilting object $T$.  
Consider the orbit category of the derived category $D^b\ch$ with respect to the auto-equivalence $\tau^{-1}\Sigma$
\be
\sC_\ch\overset{\text{def}}{=} D^b\ch\big/(\tau^{-1}\Sigma)^\Z.
\ee
$\sC_\ch$ is the category with the same objects as $D^b\ch$ and morphism spaces
\be
\sC_\ch(X,Y)\cong\bigoplus_{n\in\Z}\mathrm{Hom}_{D^b\ch}\big(X, (\tau^{-1}\Sigma)^nY\big).
\ee
By construction, $\tau\cong \Sigma$ in $\sC_\ch$. From eqn.\eqref{llxcw} we have
\be
D\sC_\ch(X,Y)\cong \sC_\ch(Y,\tau X[1])\cong \sC_\ch(Y,X[2]),
\ee
so $\sC_\ch(\ch)$ is 2-CY provided it is triangulated.

\begin{thm}[Keller \cite{kellermainthm2}]\label{kellermainthm} $\ch$ a hereditary category, $F$ a standard\,\footnote{\ Telescopic functors are standard equivalences.} equivalence of $\ch$.

 Suppose:
\begin{itemize}
\item[1)] For each indecomposable $U$ of $\ch$, only finitely many objects $F^iU$, $i \in \Z$, lie in $\ch$.
\item[2)] There is an integer $N \geq 0$ such that the $F$-orbit of each indecomposable of $D^b\ch$
contains an object $\Sigma^nU$, for some $0 \geq n \geq N$ and some indecomposable object $U$ of $\ch$.
\end{itemize}
Then the orbit category $D^b\ch/(F)^\Z$ admits a natural triangulated structure such that the projection functor $D^b\ch\to D^b\ch/(F)^\Z$  is triangulated.
\end{thm}

Then $\sC_\ch$ is a Hom-finite, triangulated 2-CY category and the image of $T$ is cluster-tilting. $\mathrm{End}_{\sC_\ch}(T)$ is a Jacobian algebra with a 2-acyclic quiver, so has a \textit{cluster structure} in the sense of \cite{reitenrev}.

\begin{defn} $\sC_\ch$ is the \textit{cluster category} of the hereditary category (with tilting) $\ch$.
\end{defn}

The cluster category $\sC_\ch$ is the solution to an universal problem:

\begin{pro}[The universal property]\label{uniprop} The cluster category $\ch$ is the universal 2-CY category under the derived category $D^b\ch$, i.e.\! let 
$\mathscr{K}$ be 2-CY category  and 
$\phi$ a triangle functor  such that $\phi\circ \tau \Sigma \simeq \Sigma^2\circ\phi$
\be\label{kkkkxxz888}
\begin{gathered}
\xymatrix{\sD_\ch\ar[rrd]^\phi\ar[d]_\pi\\
\sC_\ch\ar@{-->}[rr]&& \mathscr{K}}
\end{gathered}
\ee
Then $\phi$ factors through $\pi$.
\end{pro}
See e.g.\! diagram \eqref{factfact} which defines the discrete gauging functors $g_{d/2}$.

\subsubsection{Cluster categories from
algebras of global dimension $\leq2$}\label{amiotcc}

Let $\ca$ be a finite-dimensional
algebra of global dimension $\leq2$
and let $\sD_\ca\equiv D^b\mathsf{mod}\,\ca$ be its bounded derived category.
$\sD_\ca$ admits a Serre functor $S$ such that
\be
D\sD_\ca(X,Y)\cong \sD_\ca(Y,SX)
\ee 
given by the total derived tensor product of the bi-module $D\ca$, i.e.\! $S=?\overset{L}{\otimes}_\ca D\ca$.
The orbit category
\be
\sD_\ca/(S\Sigma^{-2})^\Z
\ee
is a linear category with a suspension functor $\Sigma$ but it is not triangulated in general, unless $\sD_\ca$ is equivalent to the derived hereditary category.
However it has a universal \emph{triangulated hull} $\sC_\ca$ \cite{kellermainthm2} 
such that there exists an algebraic triangulated functor $\pi\colon \sD_\ca\to \sC_\ca$ with a universal property analogous to \eqref{kkkkxxz888}. $\sC_\ca$ is triangulated; one shows \cite{amiot} that if the triangular hull is Hom-finite is also 2-CY.
This happens iff the functor $\mathrm{Tor}^\ca_2(?,D\ca)$ is nilpotent. Moreover the image of $\ca$ is a cluster-tilting object.

\begin{pro}[Amiot \cite{amiot}]
Write $\ca=\C Q/I$, with $I$ an ideal generated by a finite set of minimal relations $\{r_\alpha\}_{\alpha\in \sigma}$ with starting at $s(r_\alpha)$ and ending at $t(r_\alpha)$. Suppose $\sC_\ca$ to be Hom-finite. Then the quiver of the algebra $\mathrm{End}_{\sC_\ca}(\ca)$
is obtained by adding a new arrow, going in the opposite direction, for each minimal relation, i.e.\! in correspondence to the relation $r_\alpha$  we add an arrow from $t(r_\alpha)$ to $s(r_\alpha)$.
\end{pro} 

It is important to understand how smaller is the full subcategory $\sD_\ca/(S\Sigma^{-2})^\Z$
with respect to the cluster category $\sC_\ca$. The two  categories are equivalent iff
$\sD_\ca$ is equivalent to $D^b\ch$ with $\ch$ hereditary. In general one has

\begin{pro}[Amiot-Oppermann \cite{amiot2}] All rigid objects $X\in\sC_\ca$ belong to the orbit category.
In particular, all cluster-tilting objects belong to the orbit category. 
\end{pro}
We say that an object $X\in \sD_\ca$ is $a/b$-fractional CY iff $S^b X\cong X[a]$. Then
\begin{pro}[Amiot-Oppermann \cite{amiot2}] Assume that there is an indecomposable $X\in\sD_\ca$ which is fractional Calabi-Yau with $a\neq b$. Then
the triangular hull $\sC_\ca$ is strictly larger than the orbit category unless
$\sD_\ca$ is the derived category of a hereditary category.
\end{pro}

\begin{corl} Let $\ca$ be the triangular algebra associated to a quiver $Q(A,q)$ (\S.\ref{4d2d}). If $\Delta\leq2$ then $\sC_\ca$ is the orbit category unless $A=\{2\}$ in which case it is strictly larger.
\end{corl}

When $\Delta>2$ the derived category $\sD(A,q)$ of the quiver $Q(A,q)$ is fractional Calabi-Yau and not equivalent to a derived hereditary category. Hence 

\begin{corl} The cluster category for the 2-acyclic models in \S.\ref{4d2d} are strictly larger than the orbit category iff $\Delta>2$.
\end{corl}

\subsubsection{Properties
of Amiot cluster categories of
fractional CY derived categories}

Let $\ca$ a \emph{connected} finite-dimensional algebra with $\mathrm{gl.dim}\,\ca\leq2$ such that 
$\sD_\ca\equiv D^b\mathsf{mod}\,\ca$ is fractional Calabi-Yau of dimension $a/b<2$.
Then 
\begin{pro}[\S.6.1 of \cite{amiot2}]
We have the following possibilities
\begin{itemize}
\item[1)] the Auslander-Reiten (AR) quiver of $\sD_\ca$ has a unique component of the form $\Z Q$, with $Q$ a Dynkin quiver, and $\ca$ is the path algebra of $Q$. In this case $a/b<1$;
\item[2)] $a=b$ and all connected components $\Gamma$ of the AR quiver of $\sD_\ca$ are stable tubes of periods $p\mid b$;
\item[3)] if $a\neq b$ and $\ca$ is not the path algebra of a Dynkin quiver,
all connected components $\Gamma$ have the form $\Z A_\infty$.
\end{itemize}
\end{pro}

\subsection{Representations of Dynkin quivers and related categories}

Let $\Gamma$ be a quiver obtained by choosing an orientation in a Dynkin graph of type $ADE$. We identity the representations of $\Gamma$ with the module category $\mathsf{mod}_\Gamma$ of the path algebra $\C \Gamma$, which is a hereditary category with finitely-many indecomposables, all rigid, in one-to-one correspondence with the positive roots
$\Delta^+(\Gamma)$ of the underlying Dynkin graph.
The correspondence sends the indecomposable $X\in\mathsf{mod}_\Gamma$ to the root
\be
X\mapsto \sum_{v\in\Gamma} \dim X_v\,\alpha_v.
\ee
The indecomposable projective $P_i$
consists of all paths in $\Gamma$ starting at node $i$; dually the indecomposable injective $I_i$ consists of all paths terminating at $i$. Since the category $\mathsf{mod}_\Gamma$ is hereditary, all modules $X$ have a projective resolution of the form
\be
0\to \bigoplus_{i\in I} P_i \to \bigoplus_{j\in J}P_j \to X \to 0.
\ee

The indecomposables of the derived category $\sD_\Gamma\equiv D^b \mathsf{mod}_\Gamma$  then have the form $X_\beta[n]$, $\beta\in \Delta^+(\Gamma)$, $n\in\Z$. Up to equivalence,
$\sD_\Gamma$ is independent of the chosen orientation. $\sD_\Gamma$ has the triangles
inherited from $\mathsf{mod}_\Gamma$
\be
\bigoplus_{i\in I} P_i \to \bigoplus_{j\in J}P_j \to X \to \bigoplus_{i\in I}\Sigma P_i.
\ee 
The AR translation $\tau$ is then defined by the following triangle
\be\label{ttar}
\tau X\to \bigoplus_{i\in I} I_i \to \bigoplus_{j\in J}I_j\to \Sigma \tau X.
\ee
The cluster category $\sC_\Gamma$ is the orbit category of $\sD_\Gamma$ with respect to $\tau^{-1}\Sigma$. Modulo
isomorphism, its indecomposables are the indecomposables of $\mathsf{mod}_\Gamma$ together with the shifted projectives $P_i[1]$.

\subsection{Coherent sheaves on weighted projective lines}\label{Acoh}

References for this topics are \cite{GL1,lenzing1,lenzing2,lenzingmeltzer,meltzer,dirk,revLINE}. Given a set of  integral weights $\boldsymbol{p}=(p_1,p_2,\dots,p_s)$, $p_i\geq 2$ we define $L(\boldsymbol{p})$ to be the Abelian group
 over the generators
$\vec x_1,\vec x_2,\dots, \vec x_s$
subjected to the relations
\begin{equation}\vec c= p_1\vec x_1=p_2\vec x_2=\cdots=
p_s\vec x_s.\end{equation}
$\vec c$ is called the \emph{canonical} element of 
$L(\boldsymbol{p})$, while the \emph{dual} element is 
\begin{equation}
\vec \omega=(s-2)\vec c-\sum_{a=1}^s \vec x_a\in L(\boldsymbol{p}).\end{equation}
Given the weights $\boldsymbol{p}$ and $s$
distinct points $(\lambda_a:\mu_a)\in\mathbb{P}^1$ we define a ring
graded by $L(\boldsymbol{p})$
\begin{equation}S(\boldsymbol{p})=\bigoplus_{\vec a\in L(\boldsymbol{p})} S_{\vec a}= \C[X_1,X_2,\cdots, X_s,u,v]\Big/\big(X_1^{p_1}-\lambda_1 u -\mu_1 v, \;\cdots,\, X_s^{p_s}-\lambda_s u -\mu_s v\big)\end{equation}
where the degree of $X_a$ is $\vec x_a$
and the degree of $u,v$ is $\vec c$.
The weighted projective line $\mathbb{X}(\boldsymbol{p})$ is defined to be the projective scheme $\mathsf{Proj}\,S(\boldsymbol{p})$. Its Euler characteristic is 
\begin{equation}\label{pppl}\chi(\boldsymbol{p})=2-\sum_{a=1}^s(1-1/p_a).\end{equation}
The Picard group of $\mathbb{X}(\boldsymbol{p})$ (i.e.\! the group of its invertible coherent sheaves $\equiv$ line bundles) is isomorphic to the 
group $L(\boldsymbol{p})$
\begin{equation}\mathsf{Pic}\,\mathbb{X}(\boldsymbol{p})=
\big\{\co(\vec a)\;\big|\; \vec a\in L(\boldsymbol{p})\big\},\end{equation}
i.e.\! all line bundles are obtained from the structure sheaf $\co\equiv \co(0)$ by shifting its degree in $L(\boldsymbol{p})$.
The dualizing sheaf is $\co(\vec\omega)$.
Hence
\begin{equation}\tau\, \co(\vec a)=\co(\vec a+\vec\omega).\end{equation}
One has
\begin{equation}\mathrm{Hom}(\co(\vec a),\co(\vec b))\simeq S_{\vec b-\vec a},\qquad \mathrm{Ext}^1(\co(\vec a),\co(\vec b))\simeq D\,S_{\vec a+\vec\omega-\vec b}.
\end{equation}
Any non--zero morphism between 
line bundles is a monomorphism \cite{lenzing1,lenzing2}. In particular, for all line bundles $L$,
 $\mathrm{End}\,L=\C$. Hence, if
$(\lambda:\mu)\in\mathbb{P}^1$ is \underline{not} one of the special $s$ points $(\lambda_i:\mu_i)$,
we have the exact sequence
\begin{equation}0\to \co\xrightarrow{\lambda u+\mu v}\co(\vec c)\to \cs_{(\lambda:\mu)}\to 0\end{equation}
which defines a coherent sheaf $\cs_{(\lambda:\mu)}$ concentrated at $(\lambda:\mu)\in\mathbb{P}^1$. It is a simple object in the category $\mathsf{coh}\,\mathbb{X}(\boldsymbol{p})$
(the `skyscraper').
At the special points $(\lambda_a:\mu_a)\in\mathbb{P}^1$
the skyscraper is not a simple object but rather it is an indecomposable of length $p_a$. The simple sheaves localized at the $a$--th special point $(\lambda_a:\mu_a)$
are the $\cs_{a,j}$ (where $j\in\Z/p_a\Z$)
defined by the exact sequences
\begin{equation}0\to \co(j\vec x_a)\to \co((j+1)\vec x_a)\to \cs_{a,j}\to 0.\end{equation}
Applying $\tau$ to these sequences we get 
\begin{equation}\tau \cs_{(\lambda;\mu)}=\cs_{(\lambda;\mu)},\qquad \tau\cs_{a,j}=\cs_{a,j-1}.\end{equation}

In conclusion we have\footnote{ The notation in the \textsc{rhs} \cite{lenzing1,lenzing2,ringelbook} stands for two properties: \textit{(i)} all object $X$ of 
$\mathsf{coh}\,\mathbb{X}(\boldsymbol{p})$ has the form $X_+\oplus X_0$ with $X_+\in\ch_+$, $X_0\in\ch_0$, and \textit{(ii)}
$\mathrm{Hom}(\ch_0,\ch_+)=0$.} \cite{lenzing1,lenzing2}
\begin{equation}\mathsf{coh}\,\mathbb{X}(\boldsymbol{p})=\ch_+\vee \ch_0,
\end{equation}
 where $\ch_0$ is the full Abelian subcategory of finite length objects (which is a
 uniserial category) and $\ch_+$
 is the subcategory of \emph{bundles}. 
 Any non--zero morphism from a line bundle $L$ to a bundle $E$ is a monomorphism. For all bundles $E$ we have a filtration\cite{lenzing1,lenzing2}
 \begin{equation}\label{filtration}0=E_0\subset E_1\subset E_2\subset\cdots\subset E_\ell=E,\end{equation}
 with $E_{i+1}/E_i$ line bundles. Then
 we have an additive function $\mathsf{rank}\colon K_0(\mathsf{coh}\,\mathbb{X}(\boldsymbol{p}))\to \Z$, the \emph{rank}, which is $\tau$--invariant, zero on $\ch_0$
 and positive on $\ch_+$. $\mathsf{rank}\,E$ is the length $\ell$ of the filtration \eqref{filtration}; line bundles have
 rank 1. 
  
 We define the additive function \emph{degree}, $\mathsf{deg}\colon K_0(\mathsf{coh}\,\mathbb{X}(\boldsymbol{p}))\to \frac{1}{p}\Z$, by
 \begin{equation}\mathsf{deg}\,\co\!\left(\sum\nolimits_a n_a\vec x_a\right)=\sum_a\frac{n_a}{p_a}.\end{equation}
$\mathsf{deg}$ satisfies the four properties:
 \textit{(i)} the degree is $\tau$ stable;
\textit{(ii)}  $\mathsf{deg}\,\co=0$; \textit{(iii)}  if $\cs$ is a simple of $\tau$--period $q$ one has $\mathsf{deg}\,\cs=1/q$; \textit{(iv)} $\mathsf{deg}\,X>0$ for all non--zero objects in $\ch_0$.

Physically, $\mathsf{rank}$ is the Yang--Mills
magnetic charge while $\mathsf{deg}$
is (a linear combination of) the Yang--Mills electric charge (and matter charges)
normalized so that the $W$ boson has charge $+1$. For the four weighted projective lines $\mathbb{X}_p$ with $\chi(\boldsymbol{p})=0$, eqn.\eqref{pppl}, the Riemann--Roch theorem reduces to the equality \cite{GL1,lenzing1,lenzing2}
\begin{equation}\label{RR}\frac{1}{p}\sum_{j=0}^{p-1}\big\langle [\tau^jX],[Y]\big\rangle_E =\mathsf{rank}\,X\,\mathsf{deg}\,Y-\mathsf{deg}\,X\,\mathsf{rank}\,Y.\end{equation}
The \emph{slope} $\mu(E)$ of a coherent sheaf $E$ is the ratio of its degree and rank\footnote{ By convention, the zero object has all slopes.} 
 \begin{equation}\mu(E)=\mathsf{deg}\,E/\mathsf{rank}\,E.\end{equation}

The hereditary category $\mathsf{coh}\,\bX(\boldsymbol{p})$ has a \emph{canonical} tilting object $T_\text{can}$ whose endomorphism algebra is the Ringel canonical algebra $\Lambda(\boldsymbol{p})$ of type $(\boldsymbol{p})$  \cite{GL1,lenzing1,lenzing2}. $T_\text{can}$ is the direct sum of $n\equiv \sum_i (p_i-1)+2$ line bundles
\begin{equation}\label{bases}\co,\qquad \co(\ell\vec x_i)\ (\text{with }i=1,\dots,s,\ \ell=1,\dots, p_i-1),\qquad \co(\vec c).\end{equation}
By definition of tilting object, $\mathrm{Ext}^1$ vanishes between any pair of sheaves in eqn.\eqref{bases}, while
the only non--zero Hom spaces are
\begin{equation}\dim \mathrm{Hom}(\co,\co(\vec c))=2,\qquad
\dim \mathrm{Hom}(\co(k_i\vec x_i),\co(\ell_i\vec x_i))=1,\quad 0\leq k_i\leq \ell_i\leq p_i,\end{equation}
where, for all $i$, $\co(0\,\vec x_i)\equiv \co$ and $\co(p_i\vec x_i)\equiv \co(\vec c)$.

\subparagraph{Telescopic functors.}

If $\chi(\bX(\boldsymbol{p}))=0$ the derived category $D^b\mathsf{coh}\, \bX(\boldsymbol{p})$ has additional auto-equivalences generated by the telescopic functors $T$ and $L$.

$T$ is simply the functor which shifts the $L(\boldsymbol{p})$ degree of the sheaf by $\vec x_3$\cite{lenzing2,meltzer,dirk}
\begin{equation}
X\longmapsto X(\vec x_3)\equiv T X,\end{equation} where we ordered the weights so that $p_3\equiv p$ is the largest one. 
Explicitly, the action on the generating
sheaves
 $\co$, $\cs_{i,j}$ is given by
\begin{equation}T\co=\co(\vec x_3),\quad T\cs_{3,j}=\cs_{3,j+1}, \quad T\cs_{a,j}=\cs_{a,j}\ \text{for }a\neq 3.\end{equation}
Thus $T$ preserves the $\mathsf{rank}$, while increases the \textsf{degree} by $1/p$ times the $\mathsf{rank}$.
$L$ is defined by the 
triangle 
\begin{equation}\bigoplus_{j=0}^{p-1} \mathrm{Hom}^\bullet(\tau^j \co,X)\otimes \tau^j \co\xrightarrow{\ \mathrm{can_X}\ } X\longrightarrow L X.\end{equation}
One has $\mathsf{deg}\,L X=\mathsf{deg}\,X$ while
$\mathsf{rank}\,L X=\mathsf{rank}\,X-p\,\mathsf{deg}\,X$.
In particular,
\be
\begin{aligned} L\co&=\tau^{-1}\co\equiv\co(-\vec \omega),\\
L\cs_{a,j}&=\co\big(-\vec x_a+(p-1-j)\vec\omega\big)[1]\quad \text{iff }p_a=p.
\end{aligned}\ee

It is easy to see that \cite{lenzingmeltzer,meltzer}
\begin{equation}\label{braidrel2}LTL=TLT.\end{equation}
Useful formulae are \cite{Cecotti:2015hca}
\be\label{useexp}
\begin{aligned}
L\,\co(\vec x_3)&=\cs_{3,0}\\
LTL\,\co&=TLT\,\co=\cs_{3,1}\\
LTL\,\cs_{3,1}&=TLT\,\cs_{3,j}=\tau^{-(j+2)}\co[1].
\end{aligned}
\ee

$T$, $L$ generate the braid group $\cb_3$.
One has $PSL(2,\Z)=\cb_3/Z(\cb_3)$. The center $Z(\cb_3)$ of $\cb_3$ is the infinite cyclic group generated by $(TL)^3$.
One has
\begin{align}(TL)^3(\co)&=TLT\cdot LTL(\co)=
TLT(\cs_{3,1})=\tau^{-3}\co[1],\\
(TL)^3(\cs_{3,j})&=TLT\cdot LTL(\cs_{3,j})=\tau^{-(j+2)} TLT(\co)[1]=\tau^{-(j+2)}\cs_{3,1}[1]=\tau^{-3}\cs_{3,j}[1].\end{align}
So we have the isomorphism of triangle functors
\begin{equation}(TL)^3\simeq \tau^{-3}\Sigma.\end{equation}
If $p\neq3$, $\simeq$ is replaced by $=$.

\section{Computations in the derived category of del Pezzo's}

We recall that an order sequence $\{E_1,\dots, E_r\}$ of objects in a triangle category $\sD$ is \emph{exceptional} iff
\be
\sD(E_i,E_j[k])=0 \quad\forall\; k\ \text{and }i>j,\qquad \sD(E_i,E_i[k])=\begin{cases}\C & k=0\\
0 &\text{otherwise.}
\end{cases}
\ee
An exceptional sequence is \emph{full} if generates $\sD$ as a triangulated category. It is \emph{strongly exceptional} iff, in addition,
\be
\sD(E_i, E_j[k])=0\ \text{for }k\neq0\
\text{and all } i,j.
\ee
We need the following
\begin{lem}[\textbf{Corollary 2.11} of \cite{orlov1}]. Let $\{E_i\}$ be an exceptional sequence of sheaves on a del Pezzo surface $X$, then,
for $i<j$, one has $\mathrm{Ext}^2(E_i,E_j)=0$ and at most one of the two spaces $\mathrm{Hom}(E_i,E_j)$, $\mathrm{Ext}^1(E_i,E_j)$ is non zero.
\end{lem}

Clearly $\mathrm{Hom}(E_i,E_j)\neq0$ (resp.\! $\mathrm{Ext}^1(E_i,E_j)\neq0$) iff the Euler form $\chi(E_i,E_j)>0$ (resp.\! $\chi(E_i,E_j)<0$). One the other hand $\chi(E_j,E_i)=0$ by definition of exceptional sequence. Hence
\be\label{rrrcq}
\chi(E_i,E_j)=\chi(E_i,E_j)-\chi(E_j,E_i)=-r(E_i)\,c_1(E_j)\cdot K_X+r(E_j)\,c_1(E_i)\cdot K_X,
\ee 
where we used Riemann-Roch.

We apply these results to the standard full (non strongly) exceptional  sequence \cite{orlov1} $\{\co_{\ell_1}(-1),\cdots, \co_{\ell_k}(-1),\pi^*\co, \pi^*\co(1),
\pi^*\co(2)\}$. For $s=0,1,2$ we have 
\be
\chi(\co_{\ell_j}(-1),\pi^*\co(s))=c_1(\co_{\ell_j}(-1))\cdot K_X=-1,
\ee 
hence, in the derived category $\sD_X$
\be
\sD_X\big(\co_{\ell_j}(-1),\pi^*\co(s)[k]\big)=\begin{cases}\C & k=1\\
0 &\text{otherwise.}
\end{cases}
\ee
Since $\sD_X\big(\pi^*\co(s),\co_{\ell_j}(-1)[k]\big)=0$ for all $k$, the sequence
\be\label{llllzzzz1m}
\big\{E_i\big\}:=\big\{\co_{\ell_1}(-1),\cdots, \co_{\ell_k}(-1),\pi^*\co[1], \pi^*\co(1)[1],
\pi^*\co(2)[1]\big\}
\ee
is full and strongly exceptional.
The ``Cartan'' $(k+3)\times (k+3)$ matrix is
\be
S^{-1}_{ij}:=\dim\sD_X(E_i,E_j)=\begin{cases}\delta_{ij} &
1\leq i,j\leq k\\
1 & 1\leq i\leq k\ \text{and } j\geq k+1\\
\delta_{i,j}+3(j-i) & k+1\leq i\leq j\leq k+3\\
0 & \text{otherwise}
\end{cases}
\ee
so that
\be
S_{ij}=\delta_{ij}+\begin{cases} -1 & 1\leq i\leq k \ \text{and }j=k+1,\; k+3
\\
+2 & 1\leq i\leq k \ \text{and } j=k+2\\
3(-1)^{j-i} & k+1\leq i<j\leq k+3\\
0& \text{otherwise.}
\end{cases}
\ee
The  number of solid (resp.\! dashed) arrows in the quiver
with relations of the triangular algebra $\cb:=\mathrm{End}(\oplus_i E_i)$ is
\be
\#\{\xymatrix{i \ar[r]& j}\} =\max\{-S_{ij}, 0\},\qquad 
\#\{\xymatrix{i \ar@{..>}[r]& j}\}  =\max\{S_{ji}-\delta_{ij}, 0\}
\ee
so that is
\bigskip

 \be\label{uuuqw1244}
 \begin{gathered}
 \xymatrix{\bullet_1\ar@/^3.6pc/[rrrd] \ar[rr] && \bullet_{k+1}\ar@<0.4ex>[dd]\ar[dd]\ar@<-0.4ex>[dd]\\
 \vdots &&& \bullet_{k+3}\ar@{..>}@<0.4ex>[ul]\ar@{..>}[ul]\ar@{..>}@<-0.4ex>[ul] 
\\
 \bullet_k\ar@/_3.6pc/[rrru]
 \ar[uurr] && \bullet_{k+2}\ar@<0.2ex>@{..>}[uull]\ar@<-0.2ex>@{..>}[uull]\ar@<0.2ex>@{..>}[ll]
 \ar@<-0.2ex>@{..>}[ll]\ar@<0.4ex>[ur]\ar[ur]\ar@<-0.4ex>[ur]
 }
 \end{gathered}
 \ee
 \bigskip
 
 \noindent
 Erasing the last node $\bullet_{k+3}$ we get the quiver with relations $\mathring{Q}$ of the triangular algebra $\ca=\mathrm{End}(\oplus_{i=1}^{k+2}E_i)$. This gives the quiver in figure \eqref{uuuqw1345}.
 If we use the full strong exceptional  
sequence \eqref{jjjqawe} instead of the \eqref{llllzzzz1m} one, the computation is similar except that
from \eqref{rrrcq} we have an extra factor 2
\be
\dim\sD_X(\co_{\ell_i}(-1), \pi^*\ct(-1)[k])=2\;\dim\sD_X
(\co_{\ell_i}(-1), \pi^*\co(1)[k]).
\ee
The Serre functor acts $S$ acts on the Grothendieck Group $K_0(\sD_\ca)$ by the 2d monodromy matrix $H_\ca$ of the algebra $\ca$ ($\equiv$ minus the Coxeter matrix). If $\tilde S$ is the 
$(k+2)\times (k+2)$ matrix obtained by deleting the last row and column of $S$, we have \cite{CV92}
\be
H_\ca= (\tilde S^t)^{-1}\tilde S.
\ee
One easily checks that $H_\ca$ satisfy the correct equation
for the action in the Grothedieck group $K_0(\sD_\ca)$ of an auto-equivalence $S$ such that $S^d=\Sigma^{2(d-1)}$. Indeed,
the minimal equation of $H_\ca$ is
\be
H_\ca^d=1,
\ee
where $d$ is the degree of the associated del Pezzo as a hypersurface in weighted projective space.

\section{Elements of the categorical theory of flavor}\label{catfal}

We consider a cluster category $\sC_A$ arising as in \S.\ref{amiotcc}: one starts from a triangular algebra $A=\C Q/I$ with nilpotent $\mathrm{Tor}_2^A$. Then
\be
\sC_A= \Big(\sD_A\big/(\tau^{-1}\Sigma)^\Z\Big)_{\text{triangle hull}}\qquad A\ \text{is cluster-tilting in }\sC_A.
\ee
where (as usual) we write $\sD_A$ for $D^b \mathsf{mod}\,A$.

We apply the `cutting technique' of
\cite{groK} (see also \cite{Caorsi:2016ebt}). We say that an object $X\in\sD_A$ is $2q$-periodic iff  $\tau^{2q} X=(\tau^{-1}\Sigma)^mX$ for some integer $m$.  If $X$ is $2q$-periodic,  
the function $\lambda_X\colon \sC_A\to\bQ$
\be
\lambda_X(Y)=\langle Y, X\rangle=
\frac{1}{q}\sum_{k=0}^{2q-1} (-1)^k\, \dim \sC_A(Y, \Sigma^k X)
\ee
is well-defined on $K_0(\sC_A)$.
Note that we have an additional factor $1/q$ in front of the \textsc{rhs} with respect to ref.\!\cite{groK} ; this guarantees that the \textsc{lhs} is independent of the chosen $q$ and moreover $\langle Y,X\rangle$ is symmetric if both $X$ and $Y$ are periodic of possibly different periods.

Suppose that $X,Y$ ($X$ periodic) both belong to the 
orbit subcategory; since the embedding
$\sD_A/(\tau^{-1}\Sigma)^\Z \hookrightarrow \sC_A$ is fully faithfull,
we have
\be\label{zxas}
\langle Y, X\rangle =\frac{1}{q}\sum_{k=0}^{2q-1}(-1)^k \dim\sC_A(Y,\tau^kX)=
\frac{1}{q}\sum_{j\in\Z}
\sum_{k=0}^{2q-1}(-1)^k \dim\sD_A(Y,\tau^{k-j}\Sigma^{j}X).
\ee
By definition of triangular hull
\be\label{kkazqw}
 K_0(\sC_A)\cong K_0\big(\sD_A/(\tau^{-1}\Sigma)^\Z\big)
 \ee
so the restricted formula \eqref{zxas} suffices to compute the quadratic $\bQ$-form on $K_0(\sC_A)$.
If $A$ is derived equivalent to a hereditary category $\ch$, this form (with a different normalization) was computed in \cite{groK}.
We quote their result specialized to the case of interest. For brevity we omit the 4 asymptotically-free cases which are also covered by \cite{groK} . Data to the left (right) of the double bar come from physics (mathematics):

\begin{center} 
\begin{minipage}{400pt}
\begin{tabular}{cc||cccc}
$\Delta$ &$F$ & $\ch$ & $K_0(\sC_A)$ & period &quadratic $\bQ$-form\\\hline
$\tfrac{4}{3}$ & $SU(2)$ & $\mathsf{mod}\,\C A_3$
& $\Z[S_1]$ & 6 & $\langle [S_1],[S_1]\rangle=\tfrac{2}{3}$\\
$\tfrac{3}{2}$ & $SU(3)$ & $\mathsf{mod}\,\C D_4$ & $\Z[S_1]\oplus \Z[\Sigma S_2]$ & $8$ 
& $\left[\begin{array}{r}\langle[S_1],[S_1]\rangle=1\\
\langle[\Sigma S_1],[\Sigma S_1]\rangle=1\\
\langle[S_1],[\Sigma S_2]\rangle=\tfrac{1}{2}\end{array}\right.$\\
$2$ & $Spin(8)$ & $\mathsf{coh}\, \bX_{(2,2,2,2)}$ & $\left[\begin{array}{l}\sum_{a=1}^4w_a [\cs_{a,1}]\\
w_a\in\tfrac{1}{2}\Z\\
w_a=w_b\bmod1\\
\end{array}\right.$  & 2 & $\langle [\cs_{a,1}], [\cs_{b,1}]\rangle=2\,\delta_{ab}$\\
\end{tabular}
\end{minipage}
\end{center}

From this table one property is obvious:
\begin{quote}\it 
$K_0(\sC_A)$ is the weight lattice $\Gamma^\text{wgt}_F$ of $F$
and the matrix of the quadratic $\bQ$-form (in the basis in the table) is  equal to
\be
\Delta\, C^{-1}_F
\ee
where $C_F$ is the Cartan matrix of $F$.
\end{quote}
This strange result has a suggestive interpretation. For the above three 
SCFT one has the relation
\be
\Delta= \kappa_F/2,
\ee
so what we actually find is 
\be
\frac{1}{2} \kappa_F\, C^{-1}_{ab},
\ee
which is the physical natural answer since
$\kappa_F$ is the normalization of the
two point flavor currents. 

\subsection{Generalization to Amiot cluster categories}

Let us generalize this construction to the case that we have a triangular algebra $\ca$ satisfying the 4d/2d condition (\S.\ref{amiotcc}). We have the Serre functor
$S:\sD_\ca\to \sD_\ca$ given by $X\mapsto X\overset{L}{\otimes}_\ca D\ca$.
We write $\sD_\ca$ for the derived category and $\sC_\ca$ for the cluster category of $\ca$. 
We start by recalling some definitions and useful relations.

\subsubsection{Fractional Calabi-Yau categories and quantum monodromies}
\label{fracCY}
We recall that a triangle category with Serre functor $S$ is said to have \emph{fractional Calabi-Yau dimension}
$a/b$, or simply to be $a/b$-CY, iff
$S^b=\Sigma^a$  and the positive integers $a$, $b$ are minimal for this property (the `fraction' $a/b$ should not be reduced\,!!). The image of $a/b$ in $\bQ$ is written $\hat c$ and physicists call it the \textit{$2d$ superconformal central charge.}
If the category is associate to a 4d SCFT we must have
$\hat c<2$ \cite{CNV}.

\begin{defn}[\S.3.3 of \cite{Caorsi:2016ebt}]
The \emph{2d (quantum) monodromy} $H$
is the image of $S$ in the 2-periodic 
\be
\sR_\ca\equiv\Big(\sD_\ca/(\Sigma^2)^\Z\Big)_\text{triangular hull}.
\ee
The \emph{4d (quantum) monodromy} $\bM$ is the image of $S$ in the cluster category $\sC_\ca$. 
\end{defn}

We shall us the notation $o(H)$, $o(\bM)$ for the orders of 
$H$ and $\bM$, respectively.

\begin{lem}[\S.3.3 of \cite{Caorsi:2016ebt}] Let $\sD_\ca$ be fractional $a/b$-CY 
with $\hat c<2$. Then
\be
o(H)=\frac{2b}{\gcd(a,2)},\qquad o(\bM)\mid\frac{2b-a}{\gcd(a,2)},\qquad q\equiv o(\bM).
\ee
\end{lem}
The last equality follows from the fact that $q$ is defined as the period of $\Sigma^2$ in $\sC_\ca$, but $\Sigma^2\sim \bM$ in the cluster category. Note that this statement only says that $o(\bM)$ divides $(2b-a)$. However from the `coarse-grained' classification we know that \textit{in rank-1}
\be
o(\bM)=\begin{cases} (2b-a)/\gcd(a,2) & \ca\ \text{Dynkin algebra}\\
1 &\text{otherwise}.
\end{cases}
\ee
In rank-1 we have
\be\label{ghasw}
\Delta= \frac{o(H)}{o(\bM)} \overset{\text{rank-1}}{=}\frac{1}{1-\hat c/2}
\ee
In rank-1 we have only 5 possibilities for $a/b$, namely $2/2$, $1/3$, $2/4$, $2/3$, and $2/6$. 

\begin{exe} Consider the del Pezzo  algebras $\ca$ associated to the quivers $Q(\{1^p\},3)$ with
$p=6,7,8$ respectively. We have 
\begin{equation}\label{zaqwe8}
\begin{tabular}{c|cccc}\hline\hline
$p$ &  rel. in $\sD_\ca$ & CY dim. & $o(H)$ & $\Delta$\\
6 &  $S^3=\Sigma^4$ & $4/3$ & 3  & 3\\
7 &   $S^4=\Sigma^6$ & 6/4 & 4 & 4\\
8 &  $S^6=\Sigma^{10}$ & $10/6$ & 6 & 6\\\hline\hline
\end{tabular}
\end{equation} 
\end{exe}

Let $\ca$ be an Amiot algebra
such that $\sD_\ca$ is $a/b$-CY with
$a<2b$.
\begin{defn} 
The \textit{normalized Euler characteristic} in $\sC_\ca$ is
\be\label{jasqwe}
\langle X, Y\rangle =\frac{1}{o(\bM)}\sum_{k=0}^{2\,o(\bM)-1} (-1)^k\, \dim \sC_{\ca}(X,Y[k])\ee
\end{defn}

There are two possibilities:
either $o(\bM)=1$ or $o(\bM)>1$.
In rank-1 the cases with $o(\bM)>1$ correspond to Dynkin algebras, and are already covered by ref.\!\cite{groK}. We assume $o(\bM)=1$ (i.e.\! $\Delta\in\bN$) and we may further assume $\Delta\geq3$ since $\Delta=2$ is already covered by ref.\!\cite{groK} or it corresponds to $SU(2)$ $\cn=2^*$. In this case $o(H)=\Delta$ and $S^\Delta=\Sigma^2$. 
The cluster-category is symmetric,
hence $S\simeq \mathrm{Id}$.

Under these conditions we have a well-defined functor RG functor
\be
\sR_\ca\equiv \left(\sD_\ca/(\Sigma^2)^\Z\right)_{\text{tr.hull}}\to \sC_\ca\equiv \left(\sD_\ca/(S^{-1}\Sigma^2)^\Z\right)_{\text{tr.hull}}
\ee
and we have
\be
\sC_\ca(X,Y)\cong\bigoplus_{k=0}^{\Delta-1} \sR_\ca(X, S^k Y).
\ee
The `cutted' Euler form in the root category
is 
\be
\chi_{\sR_\ca}(X,Y)=\sum_{k=0}^1(-1)^k \dim\sR_\ca(X,Y[k]).
\ee 
and the normalized Euler form for $\sC_\ca$ is
\be
\langle X, Y\rangle = \sum_{k=0}^{\Delta-1} \dim\sR_\ca(X,S^kY).
\ee
Suppose now that $[X]$ is $S$-invariant class in $K_0(\sR_\ca)$, which is canonically identified with a class in $K_0(\sC_\ca)/(\text{torsion})$. We get
\be
\langle [X],[Y]\rangle= \Delta\, \chi_{\sR_\ca}(X,Y).
\ee
On the other hand, consider the simples
$S_1,\dots, S_r$ of the quiver $Q(\{1^p\},q)$ which are neither source or sink. They form a $\bQ$-basis of $K_0(\sC_A)\otimes \bQ$ and satisfy
\be
\chi_\sC(S_i,S_j)=\delta_{ij}.
\ee

\section{Covering techniques in
Representation Theory}\label{cov-tech}

Let $Q$ be a finite quiver, $I$ and admissible ideal, and $\ca= \C Q/I$ the corresponding basic $\C$-algebra.
$\ca$ may be seen as a \emph{bounded $\C$-linear} category whose objects are the nodes,
and morphisms spaces $\ca(i,j)=e_j \ca e_i$ where $e_i\equiv I_i$ is the idempotent at the $i$-node. In this language a module $X$ of $\ca$ is a functor $X\colon \ca\to \mathsf{mod}\,\C$.

Let $\bG$ be a group of auto-equivalence of the linear category $\ca$; one says that the group $\bG$ is \emph{admissible} iff it acts freely on objects (i.e.\! on the nodes of the quiver $Q$).  $\mathbb{G}$ acts on 
$\mathsf{mod}\,\ca$ by composition of functors
$
X\longmapsto X^g\equiv X\circ g$.
To each $X\in\mathsf{mod}\,\ca$ one associates its
\emph{isotropy subgroup}
 $\mathbb{G}_X\subset \mathbb{G}$
\begin{equation}
\mathbb{G}_X=\Big\{\;g\in\mathbb{G}\; \Big|\; X^g\cong X\;\Big\}.
\end{equation}
Let $\mathbb{H}\subseteq \mathbb{G}$ be a subgroup; we write $\mathsf{mod}^\mathbb{H}\!\ca$ for the full subcategory of $\mathbb{H}$--invariant modules.

In this set-up the orbit category $\cb=\ca/\bG$ is well-defined. Its objects are the orbits $\bG i$ of objects of $\ca$ and morphism
\be
\cb(\bG i,\bG j)=\bigoplus_{g\in\mathbb{G}}\ca(i,gj).
\ee
In this context the canonical projection functor
\be
F\colon \ca\to \ca/\bG\equiv \cb,
\ee  
is called a Galois cover since it behaves very much as a topological Galois cover.

\subsection{Galois covering functors}\label{cov-fffunc}
The Galois cover
$F$ induces two natural functors
between the module categories:
\begin{itemize}
\item the \emph{pull up} functor
$F^\lambda\colon\mathsf{mod}\,\cb\to \mathsf{mod}\,\ca$ defined by composition of functors 
\begin{equation}
F^\lambda\colon X\longmapsto F^\lambda X\equiv X\circ F;
\end{equation}
\item the \emph{push down} functor
$F_\lambda\colon\mathsf{mod}\,\ca\to \mathsf{mod}\,\cb$ is the map which associates to the functor $Y\colon
\ca\to \mathsf{mod}\,\C$  the functor $F_\lambda Y\colon \cb\to\mathsf{mod}\,\C$ 
acting as follows
\begin{align}
&\text{\,$\diamond$ \underline{on objects $\mathbb{G}i$}:}  
\hskip 2.37cm\mathbb{G}i\longmapsto 
F_\lambda Y(\mathbb{G}i)=\bigoplus_{g\in\mathbb{G}}Y(g i)\\ 
&
\begin{aligned}&\text{$\diamond$ \underline{on morphisms $\mathbb{G}i \xrightarrow{\;f\;}\mathbb{G}j$}}:
&&\text{$f=\sum_{g\in\mathbb{G}}f_g$
with $f_g\in \ca(i,gj)$}\\
&&&\text{then $F_\lambda Y(f)=\sum_{g\in\mathbb{G}}Y(f_g)$.}
\end{aligned} 
\end{align}
\end{itemize}
\textbf{Properties} \cite{pena}:
\begin{itemize}
\item[1)] the categories $\mathsf{mod}^\mathbb{G}\!\ca$ and $\mathsf{mod}\,\cb$ are equivalent;
\item[2)] for all $X\in\mathsf{mod}\,\ca$ and all
$g\in\mathbb{G}$ we have $F_\lambda X^g\cong F_\lambda X$ and $F^\lambda F_\lambda X\cong \bigoplus_{g\in\mathbb{G}}X^g$;
\item[3)] $F_\lambda$ and $F^\lambda$ are each other right-- and left--adjoints:
\begin{equation*}
\ca(X,F^\lambda Y)\cong \cb(F_\lambda X,Y),\ \  
\ca(F^\lambda Y, X)\cong \cb(Y,F_\lambda X)\ \  \forall\;X\in\mathsf{mod}\,\ca,\; Y\in\mathsf{mod}\,\cb.
\end{equation*}
\end{itemize}

\begin{pro}[see \cite{pena}] $\mathbb{G}$ an admissible group of automorphisms of $\ca$.
Suppose the $\ca$--module $X$ is indecomposable and $\mathbb{G}_X=(1)$. Then $F_\lambda X$
is indecomposable and for all modules $Y$
with $F_\lambda Y\simeq F_\lambda X$ there is $g\in\mathbb{G}$ such that $Y\simeq X^g$.
\end{pro}

\begin{defn}\label{unbranched} A Galois cover of bounded linear categories, $F\colon \ca\to \ca/\bG$ is said to be \emph{unbranched}
iff, for all indecomposables $X\in\mathsf{mod}\,\ca$, $\bG_X=(1)$.
That is, $\bG$ acts freely on the AR quiver.
\end{defn}

\begin{corl}\label{kz10p} Let $\ca$ be a bounded $\C$-linear category, with an admissible
group of auto-equivalences $\bG$
such that $F\colon \ca\to\ca/\bG$ is \emph{unbranched}.
Then the pair of functors $F^\lambda$, $F_\lambda$ set a
correspondence between the indecomposables of $\mathsf{mod}\,\ca/\bG$ and the indecomposables of
$\mathsf{mod}\,\ca$ well-defined up to the action of $\bG$. The AR quiver of $\mathsf{mod}\,\ca/\bG$ is the $\bG$-orbit quiver of the AR quiver of $\mathsf{mod}\,\ca$.
\end{corl}

\section{An example of non-symplectic base-change}\label{exnonsimpletic}

We consider the situation described by the commutative diagram \eqref{diagramcov}.

The covering $\phi_\ast\colon\ce^\prime\to\ce$ is physically interesting when
both elliptic surfaces, $\ce^\prime$ and $\ce$, satisfy UV and SW completeness as required to be special geometries of a $\cn=2$ QFT. However, not all such coverings correspond to gaugings.
Consider the degree 5 covering of elliptic surfaces
\be\label{falsecover}
\{II;I_1^{10}\}\to\{II^*;I_1^2\}, 
\ee
corresponding to the rational function
$z\mapsto z^5=y$. The covered surface $\ce$ is described explicitly by the functional invariant
$\cj(y)=1/(1-y^2)$.

The surface $\ce$ is associated to the
AD model of type $A_2$ ($\Delta=6/5$) while $\ce^\prime$ to some special limit of MN of type $E_8$ ($\Delta=6$).
 The cover \eqref{falsecover} cannot represent neither a $\Z_5$ gauging nor a \emph{pseudo}-gauging of MN producing the $A_2$ AD; the simplest way to see this is that in eqn.\eqref{falsecover} the
 relation between $\Delta$ and $\Delta^\prime$ is the opposite of the physically correct one, eqn.\eqref{dfffr}. The fibers of both $\ce^\prime$ and $\ce$ over the critical/branching point $0$ are smooth, and this is not consistent with $\phi^*\Omega$ being a symplectic form, as required for a physically consistent gauging. 
 
It is instructive to see how the above discussion translates in the Weierstrass model of the two Special Geometries \cite{Argyres1}
\be\{II;I_1^{10}\}\to y^2=x^3+u^5,\qquad
\{II^*;I_1^2\}\to y^2=x^3+u\ee
which exhibits the $E_8$ MN geometry
as a 5-fold cover of the $A_2$ AD one in agreement with eqn.\eqref{falsecover}.
However the SW differential of the MN model is not the pull-back of SW differential of the AD one, but rather $\lambda^\prime= u^{-4} \phi^*\lambda$ \cite{Argyres1}, where the overall factor $u^{-4}$ is needed in order to to cancel the forth-order zero
of $\phi^*d\lambda$ at the origin in order to make it into a \emph{bona fide} symplectic form.   

From the categorical side, it is also obvious 
that  \eqref{falsecover} does not represent a (false)gauging. E.g.\! the computer procedure introduced in \cite{Caorsi:2017bnp},
interpreted as a search for categorical discrete (false-)gauging, does not return anything relevant in the Minahan-Nemeshanski $E_8$ case.
    
\section{(Cluster-)tilting objects in $\mathsf{coh}\,\bX$}\label{kkkaqwe}

\subsection{Generalities}

Let $\bX$ be a weighted projective line
of type $\boldsymbol{p}=\{p_1,\cdots, p_t\}$.
$\mathsf{coh}\,\bX$ has several 
tilting objects; they are automatically cluster tilting in the corresponding cluster category
$\sC_\bX$ \cite{clucan}. The endo-algebras of some of them are well-studied. For instance,  the \emph{canonical} tilting sheaf
\be
T_\text{can}=\co\oplus\co(\vec c)\oplus\left(\bigoplus_{a=1}^t\bigoplus_{j=1}^{p_a-1}\co(j\vec x_a)\right)
\ee 
$\mathrm{End}(T_\text{can})$ the Ringel canonical algebra of type $\boldsymbol{p}$ \cite{ringelbook}.  There are other convenient tilting objects, such as the \emph{squid} tilting sheaf whose endo-algebras  are the squid algebras, and so on, see e.g. \cite{extre}. $T_\text{can}$
is an example of tilting \emph{bundle}
i.e.\! a tilting sheaf whose direct summands are all vector bundles (in this case line bundles). The basic reference for tilting sheaves for weighted projective lines is
\cite{tiltingTH}. We borrow the following
result that, for simplicity, we state in the special case of type $(2,2,\cdots,2)$

\begin{thm}[Lenzing-Meltzer \cite{tiltingTH}]
\label{thmLM}
Let $\bX$ be a weighted projective line of type $(2,\dots,2)$ with length-2 points $z_a$, $a=1,\dots,t$.
All the tilting sheaves of $\mathsf{coh}\,\bX$ have the form
\be\label{qwer}
T=T_I\oplus\left(\oplus_{a\not\in I}\cs_{a,1}\right)
\ee
where $I\subset\{1,2,\dots,t\}$ is a subset and $T_I$ is a tilting \emph{bundle}
for the weighted projective line with
the $|I|$ length-2 points $\{z_a\}_{a\in I}$.
\end{thm}   

A useful result refers to the case of tilting \emph{bundles}. We state the special case we need:
\begin{thm}[BKL \cite{extre}]\label{vvv42} Let $\bX$ be a weighted projective line of type $(2,\dots,2)$ with $\leq4$ 2's. Let $T=\oplus_i T_i$ be a tilting \emph{bundle} for $\mathsf{coh}\,\bX$. Then the slopes of the summands satisfy
\be
\max_{i,j}\!\Big(\mu(T_i)-\mu(T_j)\Big)\leq 2
\ee
with equality if and only if $T$ is a twist of $T_\text{can}$ by a line bundle.
\end{thm}

In the rest of this appendix $\bX$ is a 
weighted projective line of tubular type $(2,2,2,2)$.

\subsection{Cluster-tilting objects in $\sC_4$}\label{ctc4}

Recall from \S.\ref{aaa1098} than a sheaf $X\in \sC_4$ is basic, rigid, or cluster-tilting iff $\iota_4X$ is respectively basic, rigid, or (cluster-)tilting in $\sC$, equivalently
in $\mathsf{coh}\,\bX$. Therefore cluster-tilting objects in $\sC_4$ are just tilting sheaves which are $S_4$-invariant.
All cluster-tilting objects $T\in \sC_4$ have precisely 3 direct summands.

\begin{lem} No tilting \emph{bundle} $B\in\mathsf{coh}\,\bX$ is $S_4$-invariant.
\end{lem}
\begin{proof} Suppose
$B=\oplus_iB_i$ is invariant. Then if $B$ has a direct summand of slope $\mu(B_i)$ has also a summand of slope $\mu(S_4B_i)=-1/\mu(B_i)$. Therefore,
\be
2\leq \max_i \big|\mu(B_i)+1/\mu(B_i)\big|\leq \max_{i,j} \big(\mu(B_i)-\mu(B_j)\Big)\leq 2,
\ee
 all inequalities are equalities, $\mu(B_i)=\pm 1$ for all $i$, 
and moreover $B=T_\text{can}(\vec \eta)$,
$\vec\eta\in L$,  by theorem \ref{vvv42}.
But the summands of $T_\text{can}(\vec \eta)$ have three distinct slopes, and the condition $\mu(B_i)=\pm1$ cannot be satisfied.
\end{proof}
Therefore the tilting sheaf $\iota_4T$ should have the form \eqref{qwer} for a proper sub-set $I$.

\begin{lem} We have ($a\neq 3$)
\be\label{oo980}
S_4\,\co=\cs_{3,1}, \qquad S_4\,\co(\vec x_a)\equiv\co(-\vec x_a),\qquad S_4\,\cs_{a,j}=\tau^{-j}\co(\vec x_3-\vec x_a).
\ee
\end{lem}
\begin{proof} The first equality follows from
\eqref{useexp}.
For the second, we start from the triangle
\be
\begin{aligned}
&\phantom{mm} &&\co(\vec x_3)\to \co(\vec x_3+\vec x_a)\to \cs_{a,0} \\
\text{apply }L& && \cs_{3,0}\equiv L\co(\vec x_3)\to L\co(\vec x_3+\vec x_a)\to L\cs_{a,0}\equiv \tau\co(-\vec x_a)[1]\\
\text{rotate, then $\tau$}& &&\tau L\co(\vec x_3+\vec x_a)[-1]\to\co(-\vec x_a)\to \cs_{3,1}, 
\end{aligned}
\ee
so that
\be
L\co(\vec x_3+\vec x_a)= \tau^{-1}\co(-\vec x_3-\vec x_a)[1]\quad a\neq3,
\ee
and
\be\label{kkkasqw}
TLT\,\co(\vec x_a)= TL\,\co(\vec x_3+\vec x_a)=\tau^{-1}T\,\co(-\vec x_3-\vec x_a)[1]=
\tau^{-1}\co(-\vec x_a)[1]\simeq \co(-\vec x_a).
\ee
The third equality follows from ($a\neq3$)
\be
TLT\,\cs_{a,j}=TL\cs_{a,j}=T\,\tau^{1-j}\co(-\vec x_a)[1]=\tau^{1-j}\co(\vec x_3-\vec x_a)[1].
\ee
\end{proof}

\begin{corl} The sheaf
\be\label{thisob}
T=\co\oplus \co(\vec x_1)\oplus \cs_{2,1}
\ee
is cluster-tilting in $\sC_4$.
\end{corl}
\begin{proof}
From the lemma
\be
\iota_4 T=T\oplus S_4T=
\cs_{2,1}\oplus \cs_{3,1}\oplus \co(-\vec x_1)\otimes T^\prime_\text{can},
\ee
where $T^\prime_\text{can}\equiv\co\oplus\co(\vec x_1)\oplus \co(\vec x_4)\oplus \co(\vec c)$ is the canonical tilting object for the weighted projective line of type $(2,2)$ obtained from $\bX$ by erasing the second and third special points. By theorem \ref{thmLM} $\iota_4T$ is tilting in $\mathsf{coh}\,\bX$, hence $T$ is cluster-tilting in $\sC_4$.
\end{proof}

\begin{corl} $T$ as in \eqref{thisob}.
The quiver of the concealed-canonical algebra $\mathrm{End}(\iota_4 T)$ is
\be\label{jjzqyi}
\begin{gathered}
\xymatrix{&& \co(\vec x_1)\ar[rd]\ar[rrd]\\
\co\ar[rru]&\co(\vec x_4-\vec x_1)\ar[ru] && \cs_{2,1}\ar@{..>}[dl] & \cs_{3,1}\ar@{..>}[dll]\\
&& \co(-\vec x_1)\ar[ull]\ar[ul]}
\end{gathered}
\ee
where dashed arrows stand for minimal relations. The cluster-category endo-quiver
is the completion of this quiver, i.e.\! the one with all arrows made solid. The $\Z_2$ symmetry generated by $S_4$ corresponds to a rotation by $\pi$ around the center of the figure. 
\end{corl}

This shows all claims related to eqn.\eqref{firstclaim} in the main text.

\subsubsection{The superpotential $\cw$}\label{q11zx}

To get the superpotential for the quiver \eqref{jjzqyi} we consider the ideal triangulation of the sphere with 4 puncture
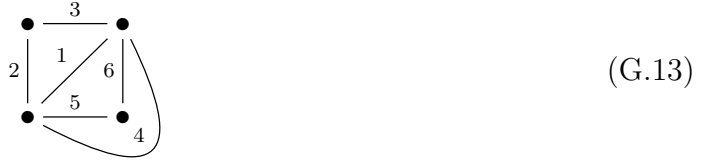
\be
\begin{gathered}
\xymatrix{\bullet\ar@{-}[d]_2 \ar@{-}[r]^3 & \bullet \ar@{-}[d]_6\ar@{-}[dl]_1\ar@{-}@/^3.5pc/[dl]_4\\
\bullet \ar@{-}[r]^5 & \bullet }
\end{gathered}
\ee
and write its non reduced incidence quiver
with potential
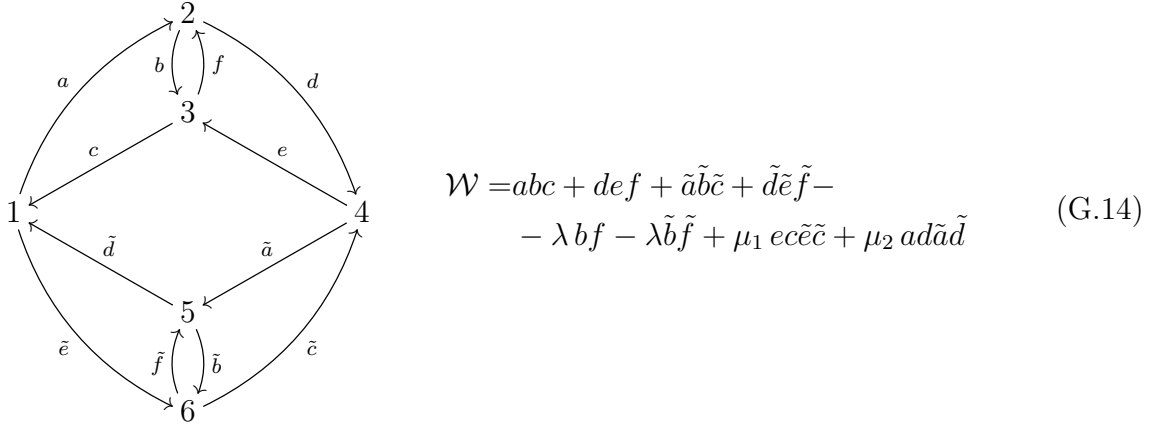
\be
\begin{gathered}
\xymatrix{&& 2\ar@/^1pc/[rrdd]^d\ar@/_0.5pc/[d]_b\\
&& 3\ar[dll]_c\ar@/_0.5pc/[u]_f\\
1\ar@/_1pc/[rrdd]_{\tilde e}\ar@/^1pc/[uurr]^a &&&& 4\ar[ull]_e\ar[lld]_{\tilde a}\\
&& 5\ar[llu]_{\tilde d}\ar@/^0.5pc/[d]^{\tilde b}\\
&& 6\ar@/^0.5pc/[u]^{\tilde f}\ar@/_1pc/[uurr]_{\tilde c}}
\end{gathered}
\qquad \begin{aligned}
\cw=&abc+def+\tilde a\tilde b\tilde c+\tilde d\tilde e\tilde f-\\
&-\lambda\, bf-\lambda \tilde b\tilde f+
\mu_1\, e c\tilde e\tilde c+\mu_2\, ad\tilde a\tilde d
\end{aligned}
\ee
Both the quiver and the potential are invariant under a $\Z_2$ symmetry which acts freely on nodes as $i\mapsto i+3\bmod6$ on the nodes and as $\ell\mapsto \tilde \ell$ on the arrows. $\lambda$, $\mu_a$ are generic coefficients. Integrating away the heavy fields $b$, $f$, $\tilde b$, $\tilde f$ we eliminate the pairs of opposite arrows getting a reduced 2-acyclic quiver of the form \eqref{sppquiver2} with potential
\be
\cw_\text{red}=\lambda^{-1}\,cade+\lambda^{-1}\,\tilde c\tilde a\tilde d\tilde e+\mu_1 ec \tilde e\tilde c+\mu_2\, a d\tilde a\tilde d.
\ee
Taking the quotient with respect to the $\Z_2$ symmetry we get
\be
\xymatrix{2 \ar@/^1pc/[r]^d &1\ar@/^1pc/[l]^a\ar@/^1pc/[r]^e & 3\ar@/^1pc/[l]^c}
\qquad \cw = \mu_1 (ec)^2+\mu_2 (ad)^2+2\lambda^{-1} adec.
\ee

\subsection{Cluster-tilting objects in $\sC_6$}\label{ctc6}

Cluster-tilting sheaves in $\sC_6$ have
just two direct summands $T=T_1\oplus T_2$. From the mutation class of the triangulation quivers for the sphere with 4 punctures, we know that a $\Z_3$ symmetry implies (many)
$\Z_6$ symmetries acting transitively on the nodes. Therefore we may find
cluster-tilting sheaves with $T_2=\tau \Pi T_1$.  $\iota_6T$ must be tilting in $\mathsf{coh}\,\bX$. 

\begin{lem} One has ($a\neq 3$)
\be
\begin{aligned}
S_6\,\co&=\tau\co(\vec x_3), & S_6^2\,\co&=\cs_{3,0},\\
S_6\,\cs_{a,j}&=\tau^j\co(\vec x_3-\vec x_a),
& S_6^2\,\cs_{a,j}&=\tau^j\co(\vec x_a).
\end{aligned}
\ee
\end{lem}
\begin{proof}
First equation: $TL\,\co=\tau T\,\co=\tau\co(\vec x_3)$. For the second
(using \eqref{useexp}):
\be
(TL)^2\co=\tau TL\,\co(\vec x_3)=T\,\cs_{3,1}=\cs_{3,0}.
\ee
Third:
\be
TL\,\cs_{a,j}=T \tau^{1-j}\co(-\vec x_a)[1]=\tau^{1-j}\co(\vec x_3-\vec x_4)\simeq
\tau^{-j}\co(\vec x_3-\vec x_4).
\ee
Forth (using \eqref{oo980}):
\be
(TL)^2\,\cs_{a,j}= T\,LTL\,\cs_{a,j}\simeq
\tau^{-j} T\,\co(\vec x_3-\vec x_a)=\tau^{j}\co(\vec x_a). 
\ee
\end{proof}
\begin{corl} The sheaf
\be
T=\cs_{1,0}\oplus \cs_{2,1}\equiv \cs_{1,0}\oplus \tau \Pi\cs_{1,0}
\ee
is cluster-tilting in $\sC_6$.
\end{corl}
\begin{proof} Using the previous lemma
\be\label{jjz10}
\iota_6 T= \co(-\vec x_1)\otimes\bigg[\Big(\co(\vec x_3)\oplus
 \co(\vec x_4)\oplus \co(\vec c)\oplus \co(\vec x_3+\vec x_4)\Big)\oplus \cs_{1,1}\oplus \cs_{2,1}\bigg]
\ee
By theorem \ref{thmLM} $\iota_6 T$
is tilting in $\mathsf{coh}\,\bX$ iff
the bundle in the large round parenthesis
is a tilting bundle for the weighted projective line of type $(2,2)$ over the third and fourth special points. But this is precisely the well-known tilting object whose endo-quiver is the acyclic affine $\widehat{A}(2,2)$ quiver with the alternating 
orientation
\be\label{affendoquiver}
\begin{gathered}
\xymatrix{\co(\vec c) & \co(\vec x_4)\ar[l]\ar[dd]\\
\\
\co(\vec x_3)\ar[uu]\ar[r] & \co(\vec x_3+\vec x_4)}
\end{gathered}
\ee
and the statement follows.
\end{proof}

Extending the affine endo-quiver
\eqref{affendoquiver} 
by the two simple sheaves inside the large bracket, $\cs_{1,1}$, $\cs_{2,1}$,
and taking into
account the overall twist by $\co(-\vec x_1)$ in the \textsc{rhs} of \eqref{jjz10}, we get the quiver for $\mathrm{End}_{\mathsf{coh}\,\bX}(\ct)$ in the form
\be
\begin{gathered}
\xymatrix{&\co(\vec x_1)\ar[dl]\ar[drr] & \co(\vec x_4-\vec x_1)\ar[l]\ar[dd]\ar@{<..}[dll]\ar@{<..}[dr]\\
\cs_{2,1}&&&\cs_{1,0}\\
&\co(\vec x_3-\vec x_1)\ar@{<..}[ul]\ar@{<..}[urr]\ar[uu]\ar[r] & \co(\vec x_3+\vec x_4-\vec x_1)\ar[ull]\ar[ur]}
\end{gathered}
\ee
where dashed arrow stand for relations as always. The quiver of the cluster category  is obtained by making solid the dashed arrows. It has the expected symmetries.
In particular, the $\Z_3$ symmetry generated by $S_6$
correspond to a $2\pi/3$ rotation of the figure. This completes the justification of the claims
about cluster-tilting in $\sC_6$
and symmetries of endo-quivers
made in section 5.

\section{Cluster-tilting in $\sC_{D_4}$
and its gaugings}\label{d4ex}

We identify the indecomposables of the cluster category $\sC_{D_4}$
with the indecomposable (right) modules of the path algebra of the Dynkin quiver \eqref{d4quiv} together with the shifted indecomposable projectives
$P_i[1]$. The indecomposable modules are in one-to-one correspondence with the positive roots of $D_4$ through their dimension vectors, and we denote them
by the corresponding root written in the form of the quiver.

Eqns.\eqref{ttar}
yield the following equalities in the derived category $\sD(D_4)\equiv D^b\mathsf{mod}\,\C D_4$
\be
\begin{aligned}
\tau\text{\begin{scriptsize}$\begin{bmatrix}&0&\\1& 0&0\end{bmatrix}$\end{scriptsize}}&=\text{\begin{scriptsize}$\begin{bmatrix}&0&\\1& 1&0\end{bmatrix}$\end{scriptsize}}[-1], & \tau\text{\begin{scriptsize}$\begin{bmatrix}&0&\\0& 1&0\end{bmatrix}$\end{scriptsize}}&=\text{\begin{scriptsize}$\begin{bmatrix}&1&\\1& 2&1\end{bmatrix}$\end{scriptsize}},
& \tau\text{\begin{scriptsize}$\begin{bmatrix}&0&\\1& 1&0\end{bmatrix}$\end{scriptsize}}
&=\text{\begin{scriptsize}$\begin{bmatrix}&1&\\0& 1&1\end{bmatrix}$\end{scriptsize}},\\
\tau\text{\begin{scriptsize}$\begin{bmatrix}&1&\\1& 1&1\end{bmatrix}$\end{scriptsize}}&=\text{\begin{scriptsize}$\begin{bmatrix}&0&\\0& 1&0\end{bmatrix}$\end{scriptsize}}[-1], & \tau\text{\begin{scriptsize}$\begin{bmatrix}&0&\\1& 1&1\end{bmatrix}$\end{scriptsize}}&=\text{\begin{scriptsize}$\begin{bmatrix}&1&\\0& 0&0\end{bmatrix}$\end{scriptsize}},
& \tau\text{\begin{scriptsize}$\begin{bmatrix}&1&\\1& 2&1\end{bmatrix}$\end{scriptsize}}
&=\text{\begin{scriptsize}$\begin{bmatrix}&1&\\1& 1&1\end{bmatrix}$\end{scriptsize}},
\end{aligned}
\ee
and other 6 obtained by acting on these ones with the automorphism $\mathfrak{S}_3$ of the quiver. One checks that
$\tau^{-3}=\Sigma$.

Consider, say, the orbit category
\be
\sC_2=\sC/\Z_2\equiv \sD(D_4)/(\tau^2)^\Z
\ee
One has
\be
\begin{split}
\sC_2(X,X[1])\cong &\sC(X,X[1])\oplus \sC(X, \tau^2 X[1])\cong
\sC(X,X[1])\oplus \sC(X, \tau^3 X)\cong\\
&\cong
\sC(X,X[1])\oplus \sC(X, \tau^{-1}X)=\mathrm{Hom}(X,\tau^{-1}X)\oplus \mathrm{Hom}(X,\tau X).
\end{split}
\ee
There are 4 orbits of indecomposables of $\sD(D_4)$ under the group generated by the two auto-equivalences $\tau$ and $\Sigma$, i.e.\! the orbits of the 4 simples $S_i$. 
If the image of an indecomposable $X\in\sD(D_4)$ is rigid in $\sD(D_4)$,
the images of all objects in its $\langle\tau,\Sigma\rangle$-orbit are also rigid.
The rigid objects in $\sC_2$ are the ones in the orbit of the peripherical nodes. Of course, elements of the same orbit get identified in pairs in $\sC_2$. Let $X$, $Y$ be rigid. $\sC_2(X,Y[1])=\sC_2(Y,X[1])=0$ requires $X$ and $Y$ to belong to different peripheral orbits.

\end{document}